%% file: VisGraphs_lipics_MFCS_arxiv.tex
\documentclass[a4paper]{article}

\usepackage[top=1in, bottom=1.25in, left=1.25in, right=1.25in]{geometry} 

\usepackage[utf8]{inputenc}
\usepackage{amsmath}
\usepackage{amssymb}
\usepackage{amsfonts}
\usepackage{amsthm}
\usepackage{tikz} 
\usetikzlibrary{positioning, chains, calc, arrows, decorations.pathreplacing, fit, patterns, graphs, external}
\usepackage{wrapfig}
\usepackage{authblk}
\usepackage{pdflscape}

\DeclareMathOperator{\bigo}{O}
\DeclareMathOperator{\smallo}{o}

\DeclareMathOperator{\LRDU}{\mathsf{LRDU}}
\DeclareMathOperator{\HV}{\mathsf{HV}}
\DeclareMathOperator{\Llabel}{\mathsf{L}}
\DeclareMathOperator{\Rlabel}{\mathsf{R}}
\DeclareMathOperator{\Dlabel}{\mathsf{D}}
\DeclareMathOperator{\Ulabel}{\mathsf{U}}
\DeclareMathOperator{\Hlabel}{\mathsf{H}}
\DeclareMathOperator{\Vlabel}{\mathsf{V}}

\DeclareMathOperator{\npclass}{\mathsf{NP}}
\DeclareMathOperator{\pclass}{\mathsf{P}}
\DeclareMathOperator{\visibility}{\sim}
\DeclareMathOperator{\layoutToGraph}{\mathsf{G}}
\DeclareMathOperator{\unitSquareGraphs}{\mathsf{USV}}
\DeclareMathOperator{\unitSquareGridGraphs}{\mathsf{USGV}}
\DeclareMathOperator{\notallequalthreeSat}{\mathsf{NAE-3SAT}}
\DeclareMathOperator{\threeSat}{\mathsf{3SAT}}
\DeclareMathOperator{\threepartition}{\mathsf{3Part}}
\DeclareMathOperator{\weak}{\mathsf{w}}
\DeclareMathOperator{\visibilityGraphs}{\mathsf{V}}
\DeclareMathOperator{\recognitionProb}{\textsc{Rec}}
\DeclareMathOperator{\Hvis}{\rightarrow}
\DeclareMathOperator{\Vvis}{\downarrow}

\DeclareMathOperator{\symHvis}{\leftrightarrow}
\DeclareMathOperator{\symVvis}{\updownarrow}

\newcommand{\visomorphic}{\textsf{V}-isomorphic}
\newcommand{\visomorphism}{\textsf{V}-isomorphism}
\newcommand{\visomorphically}{\textsf{V}-isomorphically}

\newtheorem{lemma}{Lemma}
\newtheorem{theorem}{Theorem}
\newtheorem{proposition}{Proposition}

\title{Combinatorial Properties and Recognition of Unit Square Visibility Graphs\thanks{This document is a full version (i.\,e., it contains all proofs in the Appendix) of the conference paper \cite{CasFerGriSchWhi2017}.}}

\author[1]{Katrin Casel}
\author[1]{Henning Fernau}
\author[2]{Alexander Grigoriev}
\author[1]{Markus L. Schmid}
\author[3]{Sue Whitesides}

\affil[1]{Fachbereich 4 -- Abteilung Informatikwissenschaften, Universit\"at Trier, 54286 Trier, Germany, \texttt{\{casel,fernau,mschmid\}@uni-trier.de}}
\affil[2]{Maastricht University, School of Business and Economics, P.O.Box 616, 6200 MD Maastricht, The Netherlands, \texttt{a.grigoriev@maastrichtuniversity.nl}}
\affil[3]{University of Victoria, Department of Computer Science, PO Box 1700, STN CSC, Victoria, BC, Canada V8W 2Y2, \texttt{sue@uvic.ca}}

\begin{document}

\maketitle

\begin{abstract}
Unit square (grid) visibility graphs (USV and USGV, resp.) are described by axis-parallel visibility between unit squares placed (on integer grid coordinates) in the plane. We investigate combinatorial properties of these graph classes and the hardness of variants of the recognition problem, i.\,e., the problem of representing USGV with fixed visibilities within small area and, for USV, the general recognition problem.
\end{abstract}

\section{Introduction}

A visibility representation of a graph $G$ is a set $\mathcal{R} = \{R_i \mid 1 \leq i \leq n\}$ of geometric objects (e.\,g., bars, rectangles, etc.) along with some kind of geometric visibility relation $\sim$ over $\mathcal{R}$ (e.\,g., axis-parallel visibility), such that $G = (\{v_i \mid 1 \leq i \leq n\}, \{\{v_i, v_j\} \mid R_i \sim R_j\})$. In this work, we focus on rectangle visibility graphs, which are represented by axis aligned rectangles in the plane and vertical and horizontal axis parallel visibility between them. In particular, we consider the more restricted variant of \emph{unit square visibility graphs} (see~\cite{DeaEHP2008}), and, in addition, we consider the case where the unit squares are placed on an integer grid (an alternative characterisation of the well-known class of graphs with rectilinear drawings).\par
The study of visibility representations is of interest, both for applications and for graph classes, and has remained an active research area\footnote{The 24th International Symposium on Graph Drawing and Network Visualization (GD 2016) featured an entire session on visibility representation (see \cite{ArlBGEGLMMWW2016a,ChaGGKL2016a,ChaLPW2016a,GiaDELMMW2016a}), and the joint workshop day of the Symposium on Computational Geometry (SoCG) and the ACM Symposium on Theory of Computing (STOC) included a workshop on geometric representations of graphs.} mainly because axis-aligned visibilities give rise to graph and network visualizations that satisfy good readability criteria: straight edges, and edges that cross only at right angles. These properties are highly desirable in the design of layouts of circuits and communication paths. Indeed, the study of graphs arising from vertical visibilities among disjoint, horizontal line segments (``bars'') in the plane originated during the 1980's in the context of VLSI design problems; see \cite{DucHVM83,Wis85,TamTol86}.  \par
Because bar visibility graphs are necessarily planar, this model has been extended in various ways in order to represent larger classes of graphs. Such extensions include new definitions of visibility (e.\,g., sight lines that may penetrate up to $k$ bars~\cite{DeaEGLST2007} or other geometric objects~\cite{BabGenKho2015}), vertex representations by other objects (e.\,g., rectangles, L-shapes~\cite{EvaLioMon2016}, and sets of up to $t$ bars~\cite{GauRosWen2016}), extensions to higher dimensional objects (see, e.\,g., \cite{BoseEFHLMRRSWZ98} for visibility representation in 3D by axis aligned horizontal rectangles with vertical visibilities, or \cite{FekHouWhi9596}, which studies visibility representations by unit squares floating parallel to the $x,y$-plane and lines of sight that are parallel to the $z$ axis). The desire for polysemy, that is, the expression of more than one graph by means of one underlying set of objects, has also provided impetus in the study of visibility representations (see for example \cite{BieLioMon2016} and \cite{StrWhi2003}).\par
Rectangle visibility graphs have the attractive property, for visualization purposes, that they yield right angle crossing drawings (RAC graphs (see~\cite{DidEadLio2011}), whose edges are drawn as sequences of horizontal and vertical segments forming a polyline with orthogonal bends), which have seen considerable interest in the graph drawing community. Unit square graphs form a subfamily of L-visibility graphs (see~\cite{EvaLioMon2016}) and their grid variant a subfamily of RACs with no bends (note that RAC recognition for 0-bends is $\npclass$-hard~\cite{ArgBekSym2012}).\par
Using visibilities among objects is but one example of the use of binary geometric relations for this purpose; other geometric relations include intersection relations (e.g., of strings or straight line segments in the plane, of boxes in arbitrary dimension), proximity relations (e.g., of points in the plane), and contact relations. In the literature, for the resulting graph classes, combinatorial aspects, relationships to other graph classes, as well as computational aspects are studied (see \cite{Fel2013} for a survey focusing on contact representations of rectangles).\par
Finally, we note that visibility properties among sets of objects have been studied in a number of contexts, including motion planning and computer graphics. In~\cite{Nil69} it is proposed to find shortest paths for mobile robots moving in a cluttered environment by looking for shortest paths in the visibility graph of the points located at the vertices of polygonal obstacles. This led to a search for fast algorithms to compute visibility graphs of polygons, as well as to a search for finding shortest paths without computing the entire visibility graph. \par
We extend the known combinatorial properties of unit square visibility graphs from~\cite{DeaEHP2008}, and proof their recognition problem to be $\npclass$-hard (this requires a reduction that is highly non-trivial on a technical level with the main difficulty to identify graph structures that can be shown to be representable by unit square layouts in a unique way to gain sufficient control for designing suitable gadgets). With respect to unit square \emph{grid} visibility graphs, we extend known combinatorial properties and consider variants of its recognition problem. \par
Due to space constraints, all results are formally proven in the Appendix; we sketch, however, the proof ideas for our main results. 

\section{Preliminaries}\label{sec:preliminaries}

A \emph{visibility layout}, or simply \emph{layout}, is a set $\mathcal{R} = \{R_i \mid 1 \leq i \leq n\}$ with $n \in \mathbb{N}$, where $R_i$ are closed and pairwise disjoint  axis-parallel rectangles in the plane; the \emph{position} of such a rectangle is the coordinate of its lower left corner. For every $R_i, R_j \in \mathcal{R}$, a closed non-degenerate axis-parallel rectangle $S$ (i.\,e., a non-empty closed rectangle that is not a line segment) is a \emph{visibility rectangle for $R_i$ and $R_j$} if one side of $S$ is contained in $R_i$ and the opposite side in $R_j$. We define $R_i \Hvis_{\mathcal{R}} R_j$ ($R_i \Vvis_{\mathcal{R}} R_j$), if there is a visibility rectangle $S$ for $R_i$ and $R_j$, such that the left side (upper side) of $S$ is contained in $R_i$, the right side (lower side) of $S$ is contained in $R_j$ and $S \cap R_k = \emptyset$, for every $R_k \in \mathcal{R} \setminus \{R_i, R_j\}$. Let $\symHvis_{\mathcal{R}}$ and $\symVvis_{\mathcal{R}}$ be the symmetric closures of $\Hvis_{\mathcal{R}}$ and $\Vvis_{\mathcal{R}}$, respectively. Finally, $R_i \visibility_{\mathcal{R}} R_j$ if $R_i \symHvis_{\mathcal{R}} R_j$ or $R_i \symVvis_{\mathcal{R}} R_j$ ($\visibility_{\mathcal{R}}$ is the \emph{visibility relation} (\emph{with respect to $\mathcal{R}$})). If the layout $\mathcal{R}$ is clear from the context or negligible, we drop the subscript $\mathcal{R}$. We denote $R_i \visibility R_j$, $R_i \symHvis R_j$ and $R_i \Hvis R_j$ also as $R_i$ \emph{sees} $R_j$, $R_i$ \emph{horizontally sees} $R_j$ and $R_i$ \emph{sees} $R_j$ \emph{from the left}, respectively, and analogous terminology applies to vertical visibilities. For $S, T \subseteq \mathcal{R}$, we use $S \Hvis_{\mathcal{R}} T$ as shorthand form for $\bigwedge_{R \in S, R' \in T} R \Hvis_{\mathcal{R}} R'$. \par
A layout $\mathcal{R} = \{R_i \mid 1 \leq i \leq n\}$ \emph{represents} the undirected graph $\layoutToGraph(\mathcal{R}) = (\{v_i \mid 1 \leq i \leq n\}, \{\{v_i, v_j\} \mid 1 \leq i, j \leq n, R_i \visibility R_j\})$, which is then called a \emph{visibility graph}, and the class of visibility graphs is denoted by $\visibilityGraphs$. A graph is a \emph{weak} visibility graph, if it can be obtained from a visibility graph by deleting some edges and the corresponding class of graphs is denoted by $\visibilityGraphs_{\weak}$. As a convention, for a visibility graph $G = (V, E)$ and a layout representing it we denote by $R_v$ the rectangle for $v\in V$ and define $R_{V'} = \{R_x \mid x \in V'\}$ for every $V' \subseteq V$. We call layouts $\mathcal{R}_1$ and $\mathcal{R}_2$ \emph{isomorphic} if $\layoutToGraph(\mathcal{R}_1)$ and $\layoutToGraph(\mathcal{R}_2)$ are isomorphic. Furthermore, we call $\mathcal{R}_1$ and $\mathcal{R}_2$ \emph{\visomorphic{}} if, for some $x \in \{\Hvis_{\mathcal{R}_1}, \Hvis^{-1}_{\mathcal{R}_1}\}$ and $y \in \{\Vvis_{\mathcal{R}_1}, \Vvis^{-1}_{\mathcal{R}_1}\}$, the relational structure $(\mathcal{R}_1, \Hvis_{\mathcal{R}_1}, \Vvis_{\mathcal{R}_1})$ is isomorphic to $(\mathcal{R}_2, x, y)$ or $(\mathcal{R}_2, y, x)$.\footnote{By $\preceq^{-1}$, we denote the inverse of a binary relation $\preceq$.}\par
\emph{Unit square visibility graphs} ($\unitSquareGraphs$) and \emph{unit square grid visibility graphs} ($\unitSquareGridGraphs$) are represented by \emph{unit square layouts}, where every $R \in \mathcal{R}$ is the unit square, and \emph{unit square grid layouts}, where additionally the position of every $R$ is from $\mathbb{N} \times \mathbb{N}$.\footnote{Note that in the grid case, if a unit square is positioned at $(x,y)$, then this is the only unit square on coordinates $(x', y')$, $x' \in \{x-1, x, x+1\}$, $y' \in \{y-1, y, y+1\}$.} The weak classes $\unitSquareGraphs_{\weak}$ and $\unitSquareGridGraphs_{\weak}$ are defined accordingly.\par
For a graph $G = (V, E)$, $N(v)$ is the \emph{neighbourhood} of $v \in V$, $\vec{E}$ denotes an oriented version of $E$, i.\,e., $E = \{\{u, v\} \mid (u, v) \in \vec{E}\}$, and $f\colon \vec{E}\to E, (u,v)\mapsto \{u,v\}$ is a bijection. Let $\Llabel, \Rlabel$ and $\Dlabel, \Ulabel$ be pairs of complementary values (for $X \in \{\Llabel, \Rlabel, \Dlabel, \Ulabel\}$, $\overline{X}$ denotes its complement). An \emph{$\LRDU$-restriction} (for $G$) is a labeling $\sigma \colon \vec{E} \to \{\Llabel, \Rlabel, \Dlabel, \Ulabel\}$ and it is \emph{valid} if, for every $(u, v) \in \vec{E}$ with $\sigma((u, v)) = X$ and every $w \in V \setminus \{u, v\}$, $\sigma((u, w)) \neq X \neq \sigma((w, v))$ and $\sigma((v, w)) \neq \overline{X} \neq \sigma((w, u))$. Obviously, $\LRDU$-restrictions only exist for graphs with maximum degree $4$. A unit square grid visibility layout \emph{satisfies} an $\LRDU$-restriction $\sigma$ if $\sigma((u, v)) = \Llabel$ implies $R_{v} \Hvis R_{u}$, $\sigma((u, v)) = \Rlabel$ implies $R_{u} \Hvis R_{v}$, $\sigma((u, v)) = \Dlabel$ implies $R_{u} \Vvis R_{v}$ and $\sigma((u, v)) = \Ulabel$ implies $R_{v} \Vvis R_{u}$. An \emph{$\HV$-restriction} (for $G$) is a labeling $\sigma \colon E \to \{\Hlabel, \Vlabel\}$ and it is \emph{valid} if, for every $u \in V$ at most two incident edges are labeled $\Hlabel$ and at most two incident edges are labeled $\Vlabel$. A unit square grid visibility layout \emph{satisfies} an $\HV$-restriction $\sigma$ if $\sigma(\{u, v\}) = \Hlabel$ implies $R_{v} \symHvis R_{u}$ and $\sigma(\{u, v\}) = \Vlabel$ implies $R_{v} \symVvis R_{u}$.\par
For a class $\mathfrak{G}$ of undirected graphs, the \emph{recognition problem for $\mathfrak{G}$} (denoted by $\recognitionProb(\mathfrak{G}))$ for short) is the problem to decide, for a given undirected graph $G$, whether or not $G \in \mathfrak{G}$. In the following, we shall consider the problems $\recognitionProb(\unitSquareGridGraphs)$ and $\recognitionProb(\unitSquareGraphs)$.\par
We briefly recall some established geometric graph representations relevant to this work. A \emph{rectilinear drawing} (see~\cite{EadHonPoo0910,ManPPT1011}) of a graph $G = (V, E)$ is a pair of mappings $x, y \colon V \to \mathbb{Z}$, where, for every $v \in V$, $x(v)$ and $y(v)$ represent the $x$- and $y$-coordinates of $v$ on the grid and, for every edge $\{u, v\} \in E$, $(x(u), y(u))$ and $(x(v), y(v))$ are the endpoints of a horizontal or vertical line segment that does not contain any $(x(w), y(w))$ with $w \in V \setminus \{u, v\}$. A graph has \emph{resolution} $\frac{2\pi}{d}$ if it has a drawing in which the degree of the angle between any two edges incident to a common vertex is at least $\frac{2\pi}{d}$. We call such graphs \emph{resolution-$\frac{2\pi}{d}$ graphs} and  are mainly interested in the case $d = 4$, see~\cite{ForHHKLSWW90}. For planar graphs, resolution-$\frac{2\pi}{4}$ graphs are just rectilinear graphs, see~\cite{BodTel2004}. A \emph{bendless right angle crossing} (BRAC) \emph{drawing} of a graph is a straight-line drawing in which every crossing of two edges is at right angles.\footnote{In the literature (e.\,g.,~\cite{DidEadLio2011}), the edges of a RAC-drawing are usually allowed to have bends; the investigated questions are on finding RAC-drawings that minimise the number of bends and crossings.} Note that in a BRAC-drawing or a resolution-$\frac{2\pi}{4}$ drawing, edges are not necessarily axis-parallel (like it is the case for visibility layouts and rectilinear drawings). A graph is called \emph{rectilinear} or \emph{BRAC graph} if it has a rectilinear or  BRAC-drawing, respectively.

\section{Unit Square Grid Visibility Graphs}\label{gridCaseSection}
 
The readability of graph drawings is mainly affected by its \emph{angular resolution} (angles formed by consecutive edges incident to a common node) and its \emph{crossing resolution} (angles formed at edge crossings); see the discussion in~\cite{ArgBekSym1011}. In this regard, resolution-$\frac{\pi}{2}$ graphs and BRAC graphs have an angular resolution and crossing resolution of $\frac{\pi}{2}$, respectively, while rectilinear drawings and unit square grid visibility layouts force both resolutions to be $\frac{\pi}{2}$.\par
The question arises of how these classes relate to each other and in this regard, we first note that $\unitSquareGridGraphs$ and rectilinear graphs coincide. More precisely, a unit square grid layout can be transformed into a rectilinear drawing by replacing every unit square on position $(x, y)$ by a vertex on position $(x, y)$ and translate the former visibilities into straight-line segments. Transforming a rectilinear drawing into a unit square grid layout requires scaling it first by factor $2$ and then replacing each vertex on position $(x, y)$ by a unit square on position $(x, y)$ (without scaling, sides or corners of unit squares may overlap). This only results in a \emph{weak} layout, since visibilities may be created that do not correspond to edges in the rectilinear drawing. However, any weak unit square grid visibility graph can be transformed into a unit square grid visibility graph (as formally stated below in Theorem~\ref{weakStrongEqualityTheorem}).\par
Since all these graphs except the BRAC graphs have maximum degree $4$, we only consider degree-$4$ BRAC graphs. Obviously, resolution-$\frac{\pi}{2}$ graphs and degree-$4$ BRAC graphs are both superclasses of $\unitSquareGridGraphs$ (and rectilinear graphs). Witnessed by $K_3$, the inclusion in degree-$4$ BRAC graphs is proper, while the analogous question w.\,r.\,t. resolution-$\frac{\pi}{2}$ graphs is open. Moreover, $K_3$ is also an example of a degree-$4$ BRAC graph that is not a resolution-$\frac{\pi}{2}$ graph; whether there exist resolution-$\frac{\pi}{2}$ graphs without a BRAC-drawing is open. \par
Due to the equivalence of $\unitSquareGridGraphs$ and rectilinear graphs, results for the latter graph class carry over to the former. In this regard, we first mention that the $\npclass$-hardness proof of recognizing resolution-$\frac{\pi}{2}$ graphs from~\cite{ForHHKLSWW90} actually produces drawings with axis-aligned edges; thus, it also applies to rectilinear graphs (a similar reduction (for rectilinear graphs and presented in more detail) is provided in~\cite{EadHonPoo0910}). As shown in~\cite{EadHonPoo0910}, the recognition problem for rectilinear graphs can be solved in time $\bigo(24^k \cdot k^{2k} \cdot n)$, where $k$ is the number of vertices with degree at least $3$. In~\cite{ManPPT1011}, it is shown that recognition remains $\npclass$-hard if we ask whether a drawing exists that satisfies a given $\HV$-restriction\footnote{The definition of $\HV$- and $\LRDU$-restriction given above naturally extends to rectilinear  drawings.} or a drawing that satisfies a given circular order of incident edges. However, checking the existence of a rectilinear drawing satisfying a given $\LRDU$-restriction can be done in time $\bigo(|E| \cdot |V|)$. Consequently, by trying all such labellings, we can solve the recognition problem for rectilinear graphs in time $2^{\bigo(n)}$. In this regard, it is worth noting that the hardness reduction from~\cite{EadHonPoo0910} can be easily modified, such that it also provides lower complexity bounds subject to the Exponential-Time Hypothesis (ETH), thereby demonstrating that the $2^{\bigo(n)}$ algorithm is optimal subject to ETH (see Appendix for details).

\subsection{Combinatorial Properties of $\mathsf{USGV}$} 

First, we shall see that the class $\unitSquareGridGraphs$ is downward closed w.\,r.\,t. the subgraph relation, i.\,e., if $G \in \unitSquareGridGraphs$, then all its subgraphs are in $\unitSquareGridGraphs$ (intuitively speaking, deletion of edges can be done by moving unit squares, while deletion of a vertex can be realised by deleting the corresponding unit square and then removing unwanted edges introduced by this operation). This observation will be a convenient tool for obtaining other combinatorial results. 

\begin{lemma}\label{removeVerticesEdgesLemma}
Let $G = (V, E) \in \unitSquareGridGraphs$, let $v \in V$ and $e \in E$. Then $(V, E \setminus \{e\}) \in \unitSquareGridGraphs$ and $(V \setminus \{v\}, E) \in \unitSquareGridGraphs$.
\end{lemma}

It is straightforward to prove the following limitations of $\unitSquareGridGraphs$.

\begin{lemma}\label{simpleObsLemma}
Let $G = (V, E) \in \unitSquareGridGraphs$. Then, $\mathbf{(1)}$ the maximum degree of $G$ is $4$, $\mathbf{(2)}$  for every $u, v \in V$, $|N(u) \cap N(v)| \leq 2$, and, $\mathbf{(3)}$ for every $\{u, v\} \in E$, $N(u) \cap N(v) = \emptyset$.
\end{lemma}

A consequence of Lemma~\ref{simpleObsLemma} is that no graph from $\unitSquareGridGraphs$ contains $K_{1,5}$, $K_{2,3}$ or $K_3$ as a subgraph, since they violate the first, second and third condition of Lemma~\ref{simpleObsLemma}, respectively.
Obvious examples for graphs from $\unitSquareGridGraphs$ are subgraphs of a grid; as Lemma~\ref{removeVerticesEdgesLemma} shows, even non-induced subgraphs of a grid. In this context, notice that the problem of deciding if a given graph is such a \emph{partial grid graph} is equivalent to deciding if it admits a unit-length VLSI layout, which, even restricted to trees, is an $\npclass$-hard problem; see \cite{BhaCos87} for details. Yet, $\unitSquareGridGraphs$ contains more, especially non-bipartite graphs, with the smallest example being $C_5$.\par

\begin{wrapfigure}{R}{5cm}
\centering
\begin{tabular}{cc}
\centering
\begin{tikzpicture}[nodes={draw,rectangle},scale=0.5,rotate=270]
\tikzstyle{every node}=[inner sep=0mm,minimum size=4mm,circle,draw]

    \node[scale = 0.7] (1) at (0,0) {$7$};
    \node[scale = 0.7] (2) at (2,0) {$8$};
    \node[scale = 0.7] (3) at (1,-1) {$1$};
    \node[scale = 0.7] (4) at (0,-2) {$6$};
    \node[scale = 0.7] (5) at (1,-2) {$2$};
    \node[scale = 0.7] (6) at (1,-3) {$3$};
    \node[scale = 0.7] (7) at (0,-4) {$5$};
    \node[scale = 0.7] (8) at (2,-4) {$4$};
    
    \draw (1) edge (2);
    \draw (1) edge (4);
    \draw (2) edge (3);
    \draw (2) edge (8);
    \draw (3) edge (5);
    \draw (4) edge (5);
    \draw (4) edge (7);
    \draw (5) edge (6);
    \draw (6) edge (8);
    \draw (7) edge (8);
		
\end{tikzpicture}
&
\centering
\begin{tikzpicture}[nodes={draw,rectangle},scale=0.6,rotate=90]
\tikzstyle{every node}=[inner sep=0mm,outer sep=0.5mm,minimum size=4mm,rectangle,draw]
    \node[scale = 0.7] (1) at (0,0) {$7$};
    \node[scale = 0.7] (2) at (1,0) {$8$};
    \node[scale = 0.7] (3) at (2,0) {$1$};
    \node[scale = 0.7] (4) at (0,-1) {$6$};
    \node[scale = 0.7] (5) at (2,-1) {$2$};
    \node[scale = 0.7] (6) at (0,-2) {$5$};
    \node[scale = 0.7] (7) at (1,-2) {$4$}; 
    \node[scale = 0.7] (8) at (2,-2) {$3$};

    \draw[densely dashed] (1) edge (2);
    \draw[densely dashed] (2) edge (3);
    \draw[densely dashed] (3) edge (5);
    \draw[densely dashed] (5) edge (8);
    \draw[densely dashed] (8) edge (7);
    \draw[densely dashed] (7) edge (6);
    \draw[densely dashed] (6) edge (4);
    \draw[densely dashed] (4) edge (1);
    \draw[densely dashed] (2) edge (7);
    \draw[densely dashed] (4) edge (5);
		
\end{tikzpicture}

\cr\\[-0.35cm]
\centering$(a)$&\centering $(b)$
\end{tabular}
\caption{Necessarily non-planar visibility layout for a planar graph.}
\label{Fig:Nonplanar} 
\end{wrapfigure}
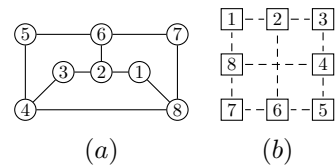

Next, we discuss planarity with a focus on the relationship between the planarity of graphs from $\unitSquareGridGraphs$ and planarity of their respective layouts (where a layout is called \emph{planar} if it does not contain any crossing visibilities). In this regard, we first note that the planarity of a layout is obviously sufficient for the planarity of the represented graph. Moreover, it is trivial to construct non-planar layouts that nevertheless represent planar graphs. Figure~\ref{Fig:Nonplanar}$(a)$ is an example of a planar unit square grid visibility graph, which can only be represented by non-planar layouts (e.\,g., the one of Figure~\ref{Fig:Nonplanar}$(b)$):

\begin{proposition}\label{nonPlanarLayoutProposition}
There exists no planar unit square grid layout for the graph of Fig.~\ref{Fig:Nonplanar}$(a)$.
\end{proposition}

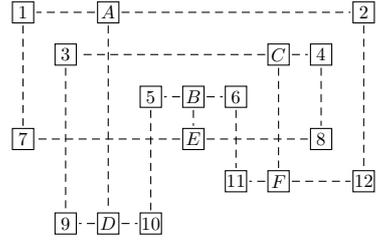
\begin{wrapfigure}{R}{5.2cm}
\begin{tikzpicture}[nodes={draw,rectangle},scale=0.56]
\tikzstyle{every node}=[inner sep=0mm,outer sep=0.5mm,minimum size=4mm,rectangle,draw]
    
\node[scale = 0.7] (B) at (0,0) {$B$};
\node[scale = 0.7] (5) at (-1,0) {$5$};
\node[scale = 0.7] (6) at (1,0) {$6$};
\node[scale = 0.7] (E) at (0,-1) {$E$};
\node[scale = 0.7] (3) at (-3,1) {$3$};
\node[scale = 0.7] (C) at (2,1) {$C$};
\node[scale = 0.7] (4) at (3,1) {$4$};
\node[scale = 0.7] (1) at (-4,2) {$1$};
\node[scale = 0.7] (A) at (-2,2) {$A$};
\node[scale = 0.7] (2) at (4,2) {$2$};
\node[scale = 0.7] (7) at (-4,-1) {$7$};
\node[scale = 0.7] (8) at (3,-1) {$8$};
\node[scale = 0.7] (9) at (-3,-3) {$9$};
\node[scale = 0.7] (D) at (-2,-3) {$D$};
\node[scale = 0.7] (10) at (-1,-3) {$10$};
\node[scale = 0.7] (11) at (1,-2) {$11$};
\node[scale = 0.7] (F) at (2,-2) {$F$};
\node[scale = 0.7] (12) at (4,-2) {$12$};
    
\draw[densely dashed] (A) edge (1);
\draw[densely dashed] (A) edge (2);
\draw[densely dashed] (A) edge (D);
\draw[densely dashed] (B) edge (5);
\draw[densely dashed] (B) edge (E);
\draw[densely dashed] (B) edge (6);
\draw[densely dashed] (C) edge (3);
\draw[densely dashed] (C) edge (F);
\draw[densely dashed] (C) edge (4);
\draw[densely dashed] (D) edge (9);
\draw[densely dashed] (D) edge (10);
\draw[densely dashed] (E) edge (7);
\draw[densely dashed] (E) edge (8);
\draw[densely dashed] (F) edge (11);
\draw[densely dashed] (F) edge (12);
\draw[densely dashed] (1) edge (7);
\draw[densely dashed] (2) edge (12);
\draw[densely dashed] (3) edge (9);
\draw[densely dashed] (4) edge (8);
\draw[densely dashed] (5) edge (10);
\draw[densely dashed] (6) edge (11);

\end{tikzpicture}
\caption{Subdivisions of $K_{3,3}$.}
\label{KFiveSubdivision}
\end{wrapfigure}
It is tempting to assume that graphs in $\unitSquareGridGraphs$ are necessarily planar, but, as demonstrated by Figure~\ref{KFiveSubdivision}, $\unitSquareGridGraphs$ contains a subdivision of $K_{3,3}$ (a layout for a subdivision of $K_5$ can be found in the Appendix). Hence, with Kuratowski's theorem, we conclude: 

\begin{theorem}\label{gridNonPlanarityTheorem}
$\unitSquareGridGraphs$ contains non-planar graphs.
\end{theorem}

Next, we investigate possibilities to characterise $\unitSquareGridGraphs$. In this regard, we first observe that a characterisation by forbidden induced subgraphs is not possible (note that under the assumption $\pclass \neq \npclass$, this also follows from the hardness of recognition).

\begin{theorem}\label{noCharacterisationGridTheorem}
$\unitSquareGridGraphs$ does not admit a characterisation by a finite number of forbidden induced subgraphs.
\end{theorem}

By Lemma~\ref{simpleObsLemma}, the classes of cycles, complete graphs and complete bipartite graphs within $\unitSquareGridGraphs$ are easily characterised: $C_i \in \unitSquareGridGraphs$ if and only if $i \geq 4$, $K_i \in \unitSquareGridGraphs$ if and only if $i \leq 2$, $K_{i, j} \in \unitSquareGridGraphs$ (with $i \leq j$) if and only if ($i = 1$ and $j \leq 4$) or ($i = 2$ and $j = 2$). Furthermore, the trees in $\unitSquareGridGraphs$ have a simple characterisation as well:

\begin{theorem}\label{thm-sgvg-tree}
A tree $T$ is in $\unitSquareGridGraphs$ if and only if the maximum degree of $T$ is at most four. 
\end{theorem}

By definition, $\unitSquareGridGraphs \subseteq \unitSquareGridGraphs_{\weak}$ and every $G' \in \unitSquareGridGraphs_{\weak}$ can be obtained from some $G \in \unitSquareGridGraphs$ by deleting some edges. Consequently, by Lemma~\ref{removeVerticesEdgesLemma}, we conclude the following.

\begin{theorem}\label{weakStrongEqualityTheorem}
$\unitSquareGridGraphs = \unitSquareGridGraphs_{\weak}$.
\end{theorem}

\subsection{Area-Minimisation}

The \emph{area-minimisation} version of the recognition problem is to decide whether a given graph has a drawing or layout of given width and height. The hardness of recognition for $\unitSquareGridGraphs$ and also for $\HV$-restricted $\unitSquareGridGraphs$ carries over to the area-minimisation version, since an $n$-vertex graph has a layout if and only if it has a $(2n-1) \times (2n-1)$ layout. On the other hand, in the $\LRDU$-restricted rectilinear (or unit square grid) case, recognition can be solved in polynomial time, so the authors of~\cite{ManPPT1011} provide a hardness reduction that proves the area-minimisation recognition problem $\npclass$-complete even for $\LRDU$-restricted rectilinear graphs. However, this construction does not carry over to $\unitSquareGridGraphs$, since the non-edges of a rectilinear drawing translate into non-visibilities, which require space as well;\footnote{In general, this space blow-up cannot be avoided, as witnessed by $n$ isolated vertices which have a $1 \times n$ rectilinear drawing, but a smallest unit square grid layout of size $(2n-1) \times (2n-1)$} moreover, it does not even work for the weak case of $\unitSquareGridGraphs$, due to the necessary scaling by factor $2$ to translate a rectilinear drawing into an equivalent weak unit square grid layout. \par
Next, we provide a reduction that shows the hardness of the area-minimisation version of $\recognitionProb(\unitSquareGridGraphs_{\weak})$, which shall also imply several additional results. The problem \textsf{3-Partition} ($\threepartition$) is defined as follows: Given $B \in \mathbb{N}$ and a multi-set $A = \{a_1, a_2, \ldots, a_{3m}\} \subseteq \mathbb{N}$ with $\frac{B}{4} < a_i < \frac{B}{2}$, $1 \leq i \leq 3m$, and $\sum^{3m}_{i = 1} a_i = m B$, decide whether $A$ can be partitioned into $m$ multi-sets $A_1, \ldots, A_m$, such that $\sum_{a \in A_j} a = B$, $1 \leq j \leq m$ (note that the restriction $\frac{B}{4} < a_i < \frac{B}{2}$ enforces $|A_j| = 3$, $1 \leq j \leq m$). Given a $\threepartition$ instance, we construct a \emph{frame graph} (see Figure~\ref{Fig:frameGraph}) $G_f = (V_f, E_f)$ with:
\begin{align*}
V_f = &\{u_{i, j}, v_{i, j}, w_{i, 1}, w_{i, 2} \mid 1 \leq i \leq m, 0 \leq j \leq B\} \cup \{u_{m + 1, 0}, v_{m + 1, 0}, w_{m+1, 1}, w_{m + 1, 2}\}\,,\\
E_f = &\left\{\{u_{i,j}, u_{i,j+1}\}, \{v_{i,j}, v_{i,j+1}\} \mid 1 \leq i \leq m, 0 \leq j \leq B-1\right\} \cup{} \\
&\left\{\{u_{i, B}, u_{i+1, 0}\}, \{v_{i, B}, v_{i+1, 0}\} \mid 1 \leq i \leq m\right\} \cup \left\{\{u_{i,j}, v_{i,j}\}\mid 1 \leq i \leq m, 1 \leq j \leq B\right\} \cup{}\\
&\left\{\{u_{i,0}, v_{i,0}\}, \{v_{i,0}, w_{i,1}\}, \{w_{i,1}, w_{i,2}\} \mid 1 \leq i \leq m+1\right\}\,.
\end{align*}
\label{areaMinimisationReductionPageRef}
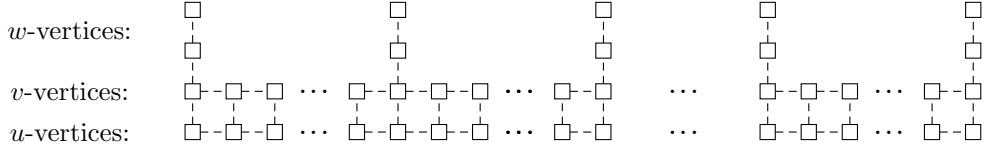
\begin{figure}
\centering
\begin{tikzpicture}[nodes={draw,rectangle},scale=0.6]
\tikzstyle{every node}=[inner sep=0mm,outer sep=0.0mm,minimum size=2mm,rectangle,draw]
     
     \coordinate (xshift) at (0.9,0);
     \coordinate (yshift) at (0,0.9);

    \node (u10) at (0,0){};
    \node (u11) at ($(u10) + (xshift)$){};
    \node (u12) at ($(u11) + (xshift)$){};
    \node (u1B) at ($(u12) + 2*(xshift)$){};
       
    \node (v10) at ($(0,0) + (yshift)$){};
    \node (v11) at ($(v10) + (xshift)$){};
    \node (v12) at ($(v11) + (xshift)$){};
    \node (v1B) at ($(v12) + 2*(xshift)$){};
    
    \node (w11) at ($(v10) + (yshift)$){};
    \node (w12) at ($(w11) + (yshift)$){};
    
    \node[draw=none, scale=0.8] (dots10) at ($(u12) + 1*(xshift)$){\bf \dots};
    \node[draw=none, scale=0.8] (dots11) at ($(v12) + 1*(xshift)$){\bf \dots};
    
    \draw[densely dashed] (u10) edge (v10);
    \draw[densely dashed] (u11) edge (v11);
    \draw[densely dashed] (u12) edge (v12);
    \draw[densely dashed] (u1B) edge (v1B);
    
    \draw[densely dashed] (u10) edge (u11);
    \draw[densely dashed] (u11) edge (u12);
    
    \draw[densely dashed] (v10) edge (v11);
    \draw[densely dashed] (v11) edge (v12);
    
    \draw[densely dashed] (v10) edge (w11);
    \draw[densely dashed] (w11) edge (w12);

    \node (u20) at ($(u1B) + (xshift)$){};
    \node (u21) at ($(u20) + (xshift)$){};
    \node (u22) at ($(u21) + (xshift)$){};
    \node (u2B) at ($(u22) + 2*(xshift)$){};
       
    \node (v20) at ($(u20) + (yshift)$){};
    \node (v21) at ($(v20) + (xshift)$){};
    \node (v22) at ($(v21) + (xshift)$){};
    \node (v2B) at ($(v22) + 2*(xshift)$){};
    
    \node (w21) at ($(v20) + (yshift)$){};
    \node (w22) at ($(w21) + (yshift)$){};
    
    \node (u30) at ($(u2B) + (xshift)$){};
    \node (v30) at ($(u30) + (yshift)$){};
    \node (w31) at ($(v30) + (yshift)$){};
    \node (w32) at ($(w31) + (yshift)$){};
    
    \node[draw=none, scale=0.8] (dots20) at ($(u22) + 1*(xshift)$){\bf \dots};
    \node[draw=none, scale=0.8] (dots21) at ($(v22) + 1*(xshift)$){\bf \dots};
    
    \draw[densely dashed] (u1B) edge (u20);
    \draw[densely dashed] (v1B) edge (v20);
    
    \draw[densely dashed] (u20) edge (v20);
    \draw[densely dashed] (u21) edge (v21);
    \draw[densely dashed] (u22) edge (v22);
    \draw[densely dashed] (u2B) edge (v2B);
    
    \draw[densely dashed] (u20) edge (u21);
    \draw[densely dashed] (u21) edge (u22);
    
    \draw[densely dashed] (v20) edge (v21);
    \draw[densely dashed] (v21) edge (v22);
    
    \draw[densely dashed] (v20) edge (w21);
    \draw[densely dashed] (w21) edge (w22);
    
    \draw[densely dashed] (u2B) edge (u30);
    \draw[densely dashed] (v2B) edge (v30);
    
    \draw[densely dashed] (u30) edge (v30);
    \draw[densely dashed] (v30) edge (w31);
    \draw[densely dashed] (w31) edge (w32);

    \node[draw=none, scale=0.8] (maindots0) at ($(u30) + 2*(xshift)$){\bf \dots};
    \node[draw=none, scale=0.8] (maindots1) at ($(v30) + 2*(xshift)$){\bf \dots};

    \node (um0) at ($(maindots0) + 2*(xshift)$){};
    \node (um1) at ($(um0) + (xshift)$){};
    \node (um2) at ($(um1) + (xshift)$){};
    \node (umB) at ($(um2) + 2*(xshift)$){};
       
    \node (vm0) at ($(um0) + (yshift)$){};
    \node (vm1) at ($(vm0) + (xshift)$){};
    \node (vm2) at ($(vm1) + (xshift)$){};
    \node (vmB) at ($(vm2) + 2*(xshift)$){};
    
    \node (wm1) at ($(vm0) + (yshift)$){};
    \node (wm2) at ($(wm1) + (yshift)$){};
    
    \node[draw=none, scale=0.8] (dots20) at ($(u22) + 1*(xshift)$){\bf \dots};
    \node[draw=none, scale=0.8] (dots21) at ($(v22) + 1*(xshift)$){\bf \dots};

    \draw[densely dashed] (um0) edge (vm0);
    \draw[densely dashed] (um1) edge (vm1);
    \draw[densely dashed] (um2) edge (vm2);
    \draw[densely dashed] (umB) edge (vmB);
    
    \draw[densely dashed] (um0) edge (um1);
    \draw[densely dashed] (um1) edge (um2);
    
    \draw[densely dashed] (vm0) edge (vm1);
    \draw[densely dashed] (vm1) edge (vm2);
    
    \draw[densely dashed] (vm0) edge (wm1);
    \draw[densely dashed] (wm1) edge (wm2);

    \node (umplus1) at ($(umB) + (xshift)$){};
    \node (vmplus1) at ($(vmB) + (xshift)$){};
    
    \node (wmplus11) at ($(vmplus1) + (yshift)$){};
    \node (wmplus12) at ($(wmplus11) + (yshift)$){};
    
    \node[draw=none, scale=0.8] (dotsm0) at ($(um2) + 1*(xshift)$){\bf \dots};
    \node[draw=none, scale=0.8] (dotsm1) at ($(vm2) + 1*(xshift)$){\bf \dots};
    
    \draw[densely dashed] (umB) edge (umplus1);
    \draw[densely dashed] (vmB) edge (vmplus1);
    
    \draw[densely dashed] (umplus1) edge (vmplus1);
    \draw[densely dashed] (vmplus1) edge (wmplus11);
    \draw[densely dashed] (wmplus11) edge (wmplus12);

    \node[draw=none, scale=1] (uvertices) at ($(u10) - 3*(xshift)$){$u$-vertices:};
    \node[draw=none, scale=1] (vvertices) at ($(v10) - 3*(xshift)$){$v$-vertices:};
    \node[draw=none, scale=1] (wvertices) at ($(w11) - 3*(xshift) + 0.5*(yshift)$){$w$-vertices:};

\end{tikzpicture}
\caption{Unit square grid layout for the graph $G_f$.}
\label{Fig:frameGraph}
\end{figure}
\noindent Next, we define a graph $G_A = (V_A, E_A)$ with $V_A = \{b_{i, j}, c_{i, j} \mid 1 \leq i \leq 3m, 1 \leq j \leq a_i\}$ and $E_A = \{\{b_{i, j}, b_{i, j+1}\}, \{c_{i, j}, c_{i, j+1}\} \mid 1 \leq i \leq 3m, 1 \leq j \leq a_i-1\} \cup \{\{b_{i, j}, c_{i, j}\} \mid 1 \leq i \leq 3m, 1 \leq j \leq a_i\}$. Finally, we let $G = (V, E)$ with $V = V_f \cup V_A$ and $E = E_f \cup E_A$.\par
The idea is that $G_f$ forms $m$ size-$B$ compartments and the graphs on $b_{i, j}$, $c_{i, j}$ represent the $a_i$. In a layout respecting the size bounds, the way of allocating the graphs on $b_{i, j}$, $c_{i, j}$ to the compartments corresponds to a partition of $A$ that is a solution for the $\threepartition$-instance.

\begin{lemma}\label{areaMinimisationLemma}
$(B, A)$ is a positive $\threepartition$-instance if and only if $G$ has a $(7 \times (2(mB+m+1)-1))$ unit square grid layout.
\end{lemma}

Since the reduction defined above is polynomial in $m$ and $B$, and $\threepartition$ is strongly $\npclass$-complete (see~\cite[Theorem 4.4]{GarJoh79}), we can conclude the following: 

\begin{theorem}\label{areaMinimisationTheorem}
The area-minimisation variant of $\recognitionProb(\unitSquareGridGraphs_{\weak})$ is $\npclass$-complete.
\end{theorem}

The area minimisation variant implicitly solves the general recognition problem, so the question arises whether it is also hard to decide if a graph from  $\unitSquareGridGraphs_{\weak}$ (given as a layout) can be represented by a layout satisfying given size bounds. Since our reduction always produces a graph in $\unitSquareGridGraphs_{\weak}$ (with an obvious layout), independent of the $\threepartition$-instance, it shows that the hardness remains if the input graph is given directly as a layout. Moreover, the problem is still $\npclass$-complete for the $\LRDU$-restricted variant (the $\LRDU$-restriction then simply enforces the structure shown in Figure~\ref{Fig:frameGraph}).\par
The reduction also yields a (substantially simpler) alternative proof for the hardness of the area-minimisation recognition problem for $\LRDU$-restricted rectilinear graphs~\cite{ManPPT1011} (more precisely, it can be shown that $(B, A)$ is a positive $\threepartition$-instance if and only if $G$ has a $(4 \times (mB+m+1))$ rectilinear drawing), and the hardness also carries over to the variant where the input graph is already given as a rectilinear drawing.\par
We conclude this section by pointing out that it is open whether the $\LRDU$-restricted area-minimisation variant of  $\recognitionProb(\unitSquareGridGraphs)$ can be solved in polynomial-time. Intuitively, reducing the size of a rectilinear drawing is difficult, since space can be saved by placing non-adjacent vertices on the same line, which is not possible for \emph{non-weak} unit square grid layouts. However, computing a size-minimal unit square grid layout includes finding out to what extend the scaling by $2$ is really necessary, which seems difficult as well.

\section{Unit Square Visibility Graphs}

Obviously, a larger class of graphs can be represented if the unit squares are not restricted to integer coordinates (see Figure~\ref{visLayoutsCompleteBipartiteGraphs} for some examples). In~\cite{DeaEHP2008}, cycles, complete graphs, complete bipartite graphs and trees in $\unitSquareGraphs$ are characterised as follows: $C_i \in \unitSquareGraphs$, for every $i \in \mathbb{N}$, $K_i \in \unitSquareGraphs$ if and only if $i \leq 4$, $K_{i, j} \in \unitSquareGraphs$ with $i \leq j$ if and only if ($1 \leq i \leq 2$ and $i \leq j \leq 6$) or ($i = 3$ and $3 \leq j \leq 4$),\footnote{For the more general question of representing bipartite graphs as rectangle visibility graphs, we refer to~\cite{DeaHut97}. In particular, a linear upper bound on the number of edges, compared to the number of vertices, is known.} and a tree $T$ is in $\unitSquareGraphs$ if and only if it is the union of two subdivided caterpillar forests with maximum degree $3$ (note that \cite{GauRosWen2016} provides an algorithm that efficiently checks this property). \par

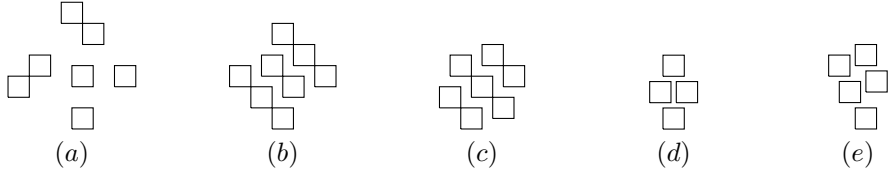
\begin{figure}
\centering
\begin{tabular}{p{2.5cm} p{2.2cm} p{2.2cm} p{2cm} p{2cm}}
\centering
\begin{tikzpicture}[scale=0.7]

\coordinate (width) at (0.4,0);
\coordinate (height) at (0,0.4);
\coordinate (wspace) at (0.2,0);
\coordinate (hspace) at (0,0.2);

\coordinate (1) at (0,0);
\coordinate (6) at ($(1) + 1.5*(width) + (wspace)$);
\coordinate (7) at ($(1) - 1.5*(height) - (hspace)$);
\coordinate (3) at ($(1) + 1.5*(height) + 0.5*(width) + (hspace)$);
\coordinate (2) at ($(3) + (height) - (width)$);
\coordinate (4) at ($(1) - 1.5*(width) + 0.5*(height) - (wspace)$);
\coordinate (5) at ($(4) - (height) - (width)$);

\coordinate (temp) at (1);
\draw[black] ($(temp)$) -- ($(temp) + (width)$) -- ($(temp) + (width) + (height)$) -- ($(temp) + (height)$) -- ($(temp)$);
\coordinate (temp) at (2);
\draw[black] ($(temp)$) -- ($(temp) + (width)$) -- ($(temp) + (width) + (height)$) -- ($(temp) + (height)$) -- ($(temp)$);
\coordinate (temp) at (3);
\draw[black] ($(temp)$) -- ($(temp) + (width)$) -- ($(temp) + (width) + (height)$) -- ($(temp) + (height)$) -- ($(temp)$);
\coordinate (temp) at (4);
\draw[black] ($(temp)$) -- ($(temp) + (width)$) -- ($(temp) + (width) + (height)$) -- ($(temp) + (height)$) -- ($(temp)$);
\coordinate (temp) at (5);
\draw[black] ($(temp)$) -- ($(temp) + (width)$) -- ($(temp) + (width) + (height)$) -- ($(temp) + (height)$) -- ($(temp)$);
\coordinate (temp) at (6);
\draw[black] ($(temp)$) -- ($(temp) + (width)$) -- ($(temp) + (width) + (height)$) -- ($(temp) + (height)$) -- ($(temp)$);
\coordinate (temp) at (7);
\draw[black] ($(temp)$) -- ($(temp) + (width)$) -- ($(temp) + (width) + (height)$) -- ($(temp) + (height)$) -- ($(temp)$);

\end{tikzpicture}
&
\centering
\begin{tikzpicture}[scale=0.7]

\coordinate (width) at (0.4,0);
\coordinate (height) at (0,0.4);

\coordinate (1) at (0,0);
\coordinate (2) at ($(1) + (width) - (height)$);
\coordinate (3) at ($(1) + 1.5*(height) + 0.5*(width)$);
\coordinate (4) at ($(3) + (width) - (height)$);
\coordinate (5) at ($(4) + (width) - (height)$);
\coordinate (6) at ($(1) - 1.5*(width) - 0.5*(height)$);
\coordinate (7) at ($(6) + (width) - (height)$);
\coordinate (8) at ($(7) + (width) - (height)$);

\coordinate (temp) at (1);
\draw[black] ($(temp)$) -- ($(temp) + (width)$) -- ($(temp) + (width) + (height)$) -- ($(temp) + (height)$) -- ($(temp)$);
\coordinate (temp) at (2);
\draw[black] ($(temp)$) -- ($(temp) + (width)$) -- ($(temp) + (width) + (height)$) -- ($(temp) + (height)$) -- ($(temp)$);
\coordinate (temp) at (3);
\draw[black] ($(temp)$) -- ($(temp) + (width)$) -- ($(temp) + (width) + (height)$) -- ($(temp) + (height)$) -- ($(temp)$);
\coordinate (temp) at (4);
\draw[black] ($(temp)$) -- ($(temp) + (width)$) -- ($(temp) + (width) + (height)$) -- ($(temp) + (height)$) -- ($(temp)$);
\coordinate (temp) at (5);
\draw[black] ($(temp)$) -- ($(temp) + (width)$) -- ($(temp) + (width) + (height)$) -- ($(temp) + (height)$) -- ($(temp)$);
\coordinate (temp) at (6);
\draw[black] ($(temp)$) -- ($(temp) + (width)$) -- ($(temp) + (width) + (height)$) -- ($(temp) + (height)$) -- ($(temp)$);
\coordinate (temp) at (7);
\draw[black] ($(temp)$) -- ($(temp) + (width)$) -- ($(temp) + (width) + (height)$) -- ($(temp) + (height)$) -- ($(temp)$);
\coordinate (temp) at (8);
\draw[black] ($(temp)$) -- ($(temp) + (width)$) -- ($(temp) + (width) + (height)$) -- ($(temp) + (height)$) -- ($(temp)$);

\end{tikzpicture}
&
\centering
\begin{tikzpicture}[scale=0.7]

\coordinate (width) at (0.4,0);
\coordinate (height) at (0,0.4);

\coordinate (1) at (0,0);
\coordinate (2) at ($(1) + (width) - (height)$);
\coordinate (3) at ($(2) + (width) - (height)$);
\coordinate (4) at ($(1) + 1.5*(width) + 0.5*(height)$);
\coordinate (5) at ($(4) + (width) - (height)$);
\coordinate (6) at ($(1) - 1.5*(height) - 0.5*(width)$);
\coordinate (7) at ($(6) + (width) - (height)$);

\coordinate (temp) at (1);
\draw[black] ($(temp)$) -- ($(temp) + (width)$) -- ($(temp) + (width) + (height)$) -- ($(temp) + (height)$) -- ($(temp)$);
\coordinate (temp) at (2);
\draw[black] ($(temp)$) -- ($(temp) + (width)$) -- ($(temp) + (width) + (height)$) -- ($(temp) + (height)$) -- ($(temp)$);
\coordinate (temp) at (3);
\draw[black] ($(temp)$) -- ($(temp) + (width)$) -- ($(temp) + (width) + (height)$) -- ($(temp) + (height)$) -- ($(temp)$);
\coordinate (temp) at (4);
\draw[black] ($(temp)$) -- ($(temp) + (width)$) -- ($(temp) + (width) + (height)$) -- ($(temp) + (height)$) -- ($(temp)$);
\coordinate (temp) at (5);
\draw[black] ($(temp)$) -- ($(temp) + (width)$) -- ($(temp) + (width) + (height)$) -- ($(temp) + (height)$) -- ($(temp)$);
\coordinate (temp) at (6);
\draw[black] ($(temp)$) -- ($(temp) + (width)$) -- ($(temp) + (width) + (height)$) -- ($(temp) + (height)$) -- ($(temp)$);
\coordinate (temp) at (7);
\draw[black] ($(temp)$) -- ($(temp) + (width)$) -- ($(temp) + (width) + (height)$) -- ($(temp) + (height)$) -- ($(temp)$);

\end{tikzpicture}
&
\centering
\begin{tikzpicture}[scale=0.7]

\coordinate (width) at (0.4,0);
\coordinate (height) at (0,0.4);

\coordinate (1) at (0,0);
\coordinate (2) at ($(1) + 1.25*(width)$);
\coordinate (3) at ($(1) + 1.25*0.5*(width) + 1.25*(height)$);
\coordinate (4) at ($(1) + 1.25*0.5*(width) - 1.25*(height)$);

\coordinate (temp) at (1);
\draw[black] ($(temp)$) -- ($(temp) + (width)$) -- ($(temp) + (width) + (height)$) -- ($(temp) + (height)$) -- ($(temp)$);
\coordinate (temp) at (2);
\draw[black] ($(temp)$) -- ($(temp) + (width)$) -- ($(temp) + (width) + (height)$) -- ($(temp) + (height)$) -- ($(temp)$);
\coordinate (temp) at (3);
\draw[black] ($(temp)$) -- ($(temp) + (width)$) -- ($(temp) + (width) + (height)$) -- ($(temp) + (height)$) -- ($(temp)$);
\coordinate (temp) at (4);
\draw[black] ($(temp)$) -- ($(temp) + (width)$) -- ($(temp) + (width) + (height)$) -- ($(temp) + (height)$) -- ($(temp)$);

\end{tikzpicture}
&
\centering
\begin{tikzpicture}[scale=0.7]

\coordinate (width) at (0.4,0);
\coordinate (height) at (0,0.4);

\coordinate (1) at ($(0,0)$);
\coordinate (2) at ($(1) + 1.25*(width) + 0.5*(height)$);
\coordinate (3) at ($(1) + 0.5*(width) - 1.25*(height)$);
\coordinate (4) at ($(1) + 1.75*(width) - 0.75*(height)$);
\coordinate (5) at ($(3) + 0.75*(width) - 1.25*(height)$);

\coordinate (temp) at (1);
\draw[black] ($(temp)$) -- ($(temp) + (width)$) -- ($(temp) + (width) + (height)$) -- ($(temp) + (height)$) -- ($(temp)$);
\coordinate (temp) at (2);
\draw[black] ($(temp)$) -- ($(temp) + (width)$) -- ($(temp) + (width) + (height)$) -- ($(temp) + (height)$) -- ($(temp)$);
\coordinate (temp) at (3);
\draw[black] ($(temp)$) -- ($(temp) + (width)$) -- ($(temp) + (width) + (height)$) -- ($(temp) + (height)$) -- ($(temp)$);
\coordinate (temp) at (4);
\draw[black] ($(temp)$) -- ($(temp) + (width)$) -- ($(temp) + (width) + (height)$) -- ($(temp) + (height)$) -- ($(temp)$);
\coordinate (temp) at (5);
\draw[black] ($(temp)$) -- ($(temp) + (width)$) -- ($(temp) + (width) + (height)$) -- ($(temp) + (height)$) -- ($(temp)$);

\end{tikzpicture}\cr
\centering$(a)$&\centering$(b)$&\centering$(c)$&\centering$(d)$&\centering$(e)$
\end{tabular}

\caption{Visibility layouts for $K_{1, 6}$, $K_{2,6}$, $K_{3,4}$, $K_4$ and a $K_5$ with one missing edge.}
\label{visLayoutsCompleteBipartiteGraphs}
\end{figure}

Next, we observe that every graph with at most $4$ vertices is in $\unitSquareGraphs$, while $K_5$ is not (it is not hard to find layouts for graphs with at most $4$ vertices; $K_5 \notin \unitSquareGraphs$ is shown in~\cite{DeaEHP2008}).

\begin{proposition}\label{atMostFourVerticesProposition}
Every graph with at most $4$ vertices is in $\unitSquareGraphs$.
\end{proposition}

A crucial difference between $\unitSquareGridGraphs$ and $\unitSquareGraphs$ is that for the latter, the degree is not bounded, as witnessed by layouts of the following form:
\begin{tikzpicture}[scale=0.5]

\coordinate (width) at (0.4,0);
\coordinate (height) at (0,0.4);

\coordinate (1) at (0,0);
\coordinate (2) at ($(1) + 1.5*(width) + 0.8*(height)$);
\coordinate (3) at ($(2) + 1.5*(width) - 0.1*(height)$);
\coordinate (4) at ($(3) + 1.5*(width) - 0.1*(height)$);
\coordinate (5) at ($(4) + 1.5*(width) - 0.1*(height)$);
\coordinate (6) at ($(5) + 1.5*(width) - 0.1*(height)$);
\coordinate (7) at ($(6) + 1.5*(width) - 0.1*(height)$);

\coordinate (temp) at (1);
\draw[black] ($(temp)$) -- ($(temp) + (width)$) -- ($(temp) + (width) + (height)$) -- ($(temp) + (height)$) -- ($(temp)$);
\coordinate (temp) at (2);
\draw[black] ($(temp)$) -- ($(temp) + (width)$) -- ($(temp) + (width) + (height)$) -- ($(temp) + (height)$) -- ($(temp)$);
\coordinate (temp) at (3);
\draw[black] ($(temp)$) -- ($(temp) + (width)$) -- ($(temp) + (width) + (height)$) -- ($(temp) + (height)$) -- ($(temp)$);
\coordinate (temp) at (4);
\draw[black] ($(temp)$) -- ($(temp) + (width)$) -- ($(temp) + (width) + (height)$) -- ($(temp) + (height)$) -- ($(temp)$);
\coordinate (temp) at (5);
\draw[black] ($(temp)$) -- ($(temp) + (width)$) -- ($(temp) + (width) + (height)$) -- ($(temp) + (height)$) -- ($(temp)$);
\coordinate (temp) at (6);
\draw[black] ($(temp)$) -- ($(temp) + (width)$) -- ($(temp) + (width) + (height)$) -- ($(temp) + (height)$) -- ($(temp)$);
\coordinate (temp) at (7);
\draw[black] ($(temp)$) -- ($(temp) + (width)$) -- ($(temp) + (width) + (height)$) -- ($(temp) + (height)$) -- ($(temp)$);
\end{tikzpicture}. However, if a unit square sees at least $7$ other unit squares, then these must be placed in such a way that visibilities or ``paths'' between some of them are enforced (note that any $K_{1,n}$ may exist as induced subgraph, as can be demonstrated by modifying the above example layout such that between each two consecutive  neighbours another ``visibility-blocking'' unit square is inserted). In \cite{DeaEHP2008}, it is formally proven that in graphs from $\unitSquareGraphs$ any vertex of degree at least $7$ must lie on a cycle. In particular, these observations point out that an analogue of Lemma~\ref{removeVerticesEdgesLemma} is not possible for $\unitSquareGraphs$.\par
For the class of trees within $\unitSquareGraphs$, as long as we consider trees with maximum degree strictly less or larger than $6$, a much simpler characterisation (compared to the one mentioned at the beginning of this section) applies:

\begin{theorem}\label{nonGridTreeTheorem}
Let $T$ be a tree with maximum degree $k$. If $k \leq 5$, then $T \in \unitSquareGraphs$, and if $k \geq 7$, then $T \notin \unitSquareGraphs$.
\end{theorem}

Figure~\ref{Fig:degreeSixTree}$(a)$ shows an example of a tree from $\unitSquareGraphs$ with maximum degree $6$ and Figure~\ref{Fig:degreeSixTree}$(b)$ its representing layout. It can be easily verified that any node of degree $6$ must be represented \visomorphically{} to Figure~\ref{visLayoutsCompleteBipartiteGraphs}$(a)$ (note that this also holds for nodes $A$ and $B$ in Figures~\ref{Fig:degreeSixTree}$(a)$~and~$(b)$). Figure~\ref{visLayoutsCompleteBipartiteGraphs}$(a)$ also demonstrates that not all trees with maximum degree $6$ can be represented: let $R$ denote the square below the central square in the layout, then it is impossible for $R$ to see $5$ additional unit squares that exclusively see $R$.
On the other hand, $\unitSquareGraphs$ contains trees with arbitrarily many degree-$6$ vertices, e.\,g., trees of the form depicted in Figure~\ref{Fig:degreeSixTree}$(c)$ (it is straightforward to see that they can be represented as the union of two forests of caterpillars with maximum degree $3$).
This reasoning shows that not all planar graphs are in $\unitSquareGraphs$, while it follows from \cite{Wis85} that all planar graphs are (non-unit square) rectangle visibility graphs (also see \cite{TamTol86}).\par
Finally, we note that $\unitSquareGraphs$ is a proper subset of $\unitSquareGraphs_{\weak}$ (e.\,g., $K_{1, 7}$ is a separating example):

\begin{theorem}\label{weakStrongSubsetTheorem}
$\unitSquareGraphs \subsetneq \unitSquareGraphs_{\weak}$.
\end{theorem}

\begin{figure}
\centering
\begin{tabular}{ccc}

\centering
\begin{tikzpicture} 
\tikzstyle{every node}=[inner sep=0mm,minimum size=4mm,circle,draw]
     
    \node[scale = 0.7] (A) at (0,0.5) {$A$};
    \node[scale = 0.7] (B) at (0,-0.5) {$B$};
    \node[scale = 0.7] (1) at ($(A) + (0.5,0)$) {$1$};
    \node[scale = 0.7] (2) at ($(A) + (0.5,0.5)$) {$2$};
    \node[scale = 0.7] (3) at ($(A) + (0,0.5)$) {$3$};
    \node[scale = 0.7] (4) at ($(A) + (-0.5,0.5)$) {$4$};
    \node[scale = 0.7] (5) at ($(A) + (-0.5,0)$) {$5$};
    \node[scale = 0.7] (6) at ($(B) + (-0.5,0)$) {$6$};
    \node[scale = 0.7] (7) at ($(B) + (-0.5,-0.5)$) {$7$};
    \node[scale = 0.7] (8) at ($(B) + (0,-0.5)$) {$8$};
    \node[scale = 0.7] (9) at ($(B) + (0.5,-0.5)$) {$9$};
    \node[scale = 0.7] (10) at ($(B) + (0.5,0)$) {$10$};

    \draw (A) edge (1);
    \draw (A) edge (2);
    \draw (A) edge (3);
    \draw (A) edge (4);
    \draw (A) edge (5);
    \draw (A) edge (B);
    \draw (B) edge (6);
    \draw (B) edge (7);
    \draw (B) edge (8);
    \draw (B) edge (9);
    \draw (B) edge (10);

\end{tikzpicture}
&
\hspace{0.5cm}
\centering
\begin{tikzpicture} [scale=0.7]

\coordinate (width) at (0.4,0);
\coordinate (height) at (0,0.4);  
\coordinate (labelshift) at (0.2, 0.2);

\coordinate (A) at (0,0);
\coordinate (1) at ($(A) - 1.75*(width)$);
\coordinate (2) at ($(A) - 0.45*(width) + 1.75*(height)$);
\coordinate (3) at ($(A) + 0.55*(width) + 2.85*(height)$);
\coordinate (4) at ($(A) + 1.75*(width) + 0.5*(height)$);
\coordinate (5) at ($(A)- 2*(height)$);

\coordinate (B) at ($(A) + 5*(width) - 0.5*(height)$);
\coordinate (6) at ($(B) + 4.4*(height)$);
\coordinate (7) at ($(B) + 1.75*(width)$);
\coordinate (8) at ($(B) - 0.45*(width) - 2.75*(height)$);
\coordinate (9) at ($(B) + 0.55*(width) - 3.85*(height)$);
\coordinate (10) at ($(B) - 1.75*(width) - 0.5*(height)$);

\coordinate (temp) at (A);
\draw[black] ($(temp)$) -- ($(temp) + (width)$) -- ($(temp) + (width) + (height)$) -- ($(temp) + (height)$) -- ($(temp)$);
\node[scale=0.8] (label) at ($(temp) + (labelshift)$)  {$A$};
\coordinate (temp) at (1);
\draw[black] ($(temp)$) -- ($(temp) + (width)$) -- ($(temp) + (width) + (height)$) -- ($(temp) + (height)$) -- ($(temp)$);
\node[scale=0.8] (label) at ($(temp) + (labelshift)$)  {$1$};
\coordinate (temp) at (2);
\draw[black] ($(temp)$) -- ($(temp) + (width)$) -- ($(temp) + (width) + (height)$) -- ($(temp) + (height)$) -- ($(temp)$);
\node[scale=0.8] (label) at ($(temp) + (labelshift)$)  {$2$};
\coordinate (temp) at (3);
\draw[black] ($(temp)$) -- ($(temp) + (width)$) -- ($(temp) + (width) + (height)$) -- ($(temp) + (height)$) -- ($(temp)$);
\node[scale=0.8] (label) at ($(temp) + (labelshift)$)  {$3$};
\coordinate (temp) at (4);
\draw[black] ($(temp)$) -- ($(temp) + (width)$) -- ($(temp) + (width) + (height)$) -- ($(temp) + (height)$) -- ($(temp)$);
\node[scale=0.8] (label) at ($(temp) + (labelshift)$)  {$4$};
\coordinate (temp) at (5);
\draw[black] ($(temp)$) -- ($(temp) + (width)$) -- ($(temp) + (width) + (height)$) -- ($(temp) + (height)$) -- ($(temp)$);
\node[scale=0.8] (label) at ($(temp) + (labelshift)$)  {$5$};
\coordinate (temp) at (B);
\draw[black] ($(temp)$) -- ($(temp) + (width)$) -- ($(temp) + (width) + (height)$) -- ($(temp) + (height)$) -- ($(temp)$);
\node[scale=0.8] (label) at ($(temp) + (labelshift)$)  {$B$};
\coordinate (temp) at (6);
\draw[black] ($(temp)$) -- ($(temp) + (width)$) -- ($(temp) + (width) + (height)$) -- ($(temp) + (height)$) -- ($(temp)$);
\node[scale=0.8] (label) at ($(temp) + (labelshift)$)  {$6$};
\coordinate (temp) at (7);
\draw[black] ($(temp)$) -- ($(temp) + (width)$) -- ($(temp) + (width) + (height)$) -- ($(temp) + (height)$) -- ($(temp)$);
\node[scale=0.8] (label) at ($(temp) + (labelshift)$)  {$7$};
\coordinate (temp) at (8);
\draw[black] ($(temp)$) -- ($(temp) + (width)$) -- ($(temp) + (width) + (height)$) -- ($(temp) + (height)$) -- ($(temp)$);
\node[scale=0.8] (label) at ($(temp) + (labelshift)$)  {$8$};
\coordinate (temp) at (9);
\draw[black] ($(temp)$) -- ($(temp) + (width)$) -- ($(temp) + (width) + (height)$) -- ($(temp) + (height)$) -- ($(temp)$);
\node[scale=0.8] (label) at ($(temp) + (labelshift)$)  {$9$};
\coordinate (temp) at (10);
\draw[black] ($(temp)$) -- ($(temp) + (width)$) -- ($(temp) + (width) + (height)$) -- ($(temp) + (height)$) -- ($(temp)$);
\node[scale=0.8] (label) at ($(temp) + (labelshift)$)  {$10$};

\end{tikzpicture}
&
\hspace{0.5cm}
\centering
\begin{tikzpicture} 
\tikzstyle{every node}=[inner sep=0mm,minimum size=4mm,circle,draw,scale=0.5]
     
     \coordinate (xshift) at (0.8,0);
     \coordinate (yshift) at (0,0.5);

    \node (1) at (0,0){};
    \node (2) at ($(1) + (xshift)$){};
    \node (3) at ($(2) + (xshift)$){};
    \node (4) at ($(3) + (xshift)$){};
    \node (5) at ($(4) + (xshift)$){};
    \node (6) at ($(5) + (xshift)$){};
    \node (7) at ($(6) + (xshift)$){};
    \node (8) at ($(7) + (xshift)$){};
    \node (9) at ($(8) + (xshift)$){};
        
    \node (2a) at ($(2) + (yshift) - 0.5*(xshift)$){};
    \node (2b) at ($(2) + (yshift) + 0.5*(xshift)$){};
    \node (2c) at ($(2) - (yshift) - 0.5*(xshift)$){};
    \node (2d) at ($(2) - (yshift) + 0.5*(xshift)$){};
    
    \node (4a) at ($(4) + (yshift) - 0.5*(xshift)$){};
    \node (4b) at ($(4) + (yshift) + 0.5*(xshift)$){};
    \node (4c) at ($(4) - (yshift) - 0.5*(xshift)$){};
    \node (4d) at ($(4) - (yshift) + 0.5*(xshift)$){};
    
    \node (6a) at ($(6) + (yshift) - 0.5*(xshift)$){};
    \node (6b) at ($(6) + (yshift) + 0.5*(xshift)$){};
    \node (6c) at ($(6) - (yshift) - 0.5*(xshift)$){};
    \node (6d) at ($(6) - (yshift) + 0.5*(xshift)$){};
    
    \node (8a) at ($(8) + (yshift) - 0.5*(xshift)$){};
    \node (8b) at ($(8) + (yshift) + 0.5*(xshift)$){};
    \node (8c) at ($(8) - (yshift) - 0.5*(xshift)$){};
    \node (8d) at ($(8) - (yshift) + 0.5*(xshift)$){};

    \draw (1) edge (2);
    \draw (2) edge (3);
    \draw (3) edge (4);
    \draw (4) edge (5);
    \draw (5) edge (6);
    \draw (6) edge (7);
    \draw (7) edge (8);
    \draw (8) edge (9);
    
    \draw (2) edge (2a);
    \draw (2) edge (2b);
    \draw (2) edge (2c);
    \draw (2) edge (2d);
    
    \draw (4) edge (4a);
    \draw (4) edge (4b);
    \draw (4) edge (4c);
    \draw (4) edge (4d);
    
    \draw (6) edge (6a);
    \draw (6) edge (6b);
    \draw (6) edge (6c);
    \draw (6) edge (6d);

    \draw (8) edge (8a);
    \draw (8) edge (8b);
    \draw (8) edge (8c);
    \draw (8) edge (8d);

\end{tikzpicture}\cr
\centering$(a)$&\centering$(b)$&\centering$(c)$
\end{tabular}

\caption{Illustration for trees from $\unitSquareGraphs$ with maximum degree $6$.}
\label{Fig:degreeSixTree}
\end{figure}
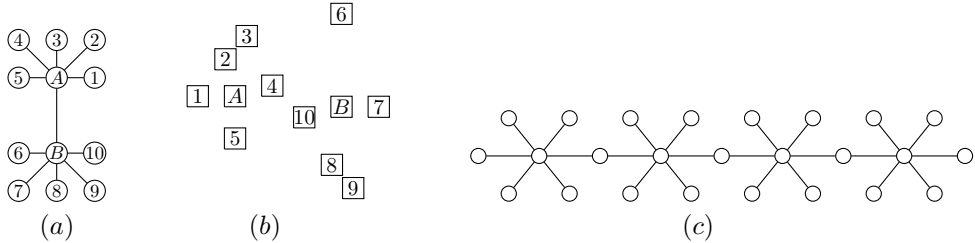

\subsection{The Recognition Problem}

The recognition problem for $\unitSquareGraphs$ consists in checking whether a given graph can be represented by a unit square layout. We first observe that this problem is in $\npclass$ (note that this is not completely trivial, since we cannot naively guess a layout) and the main result of this section shall be its hardness (see Theorem~\ref{nonGridNPCompletenessTheorem}). 

\begin{theorem}\label{NPMembershipTheorem}
$\recognitionProb(\unitSquareGraphs) \in \npclass$.
\end{theorem}

We prove the $\npclass$-hardness by a reduction from $\notallequalthreeSat$, i.\,e., the not-all-equal $3$-satisfiability problem~\cite{Sch78}. To this end, let $F = \{c_1, \ldots, c_m\}$ be a $3$-CNF formula over the variables $x_1,\dots,x_n$, such that no variable occurs more than once in any clause, and, for the sake of convenience, let $c_i = \{y_{i, 1}, y_{i, 2}, y_{i, 3}\}$, $1 \leq i \leq m$.\par 
The general idea of the reduction is as follows: We identify graph structures that can be shown to have a (more or less) unique representation as a unit square layout. With these main building blocks, we construct a sequence of clause and variable gadgets, called \emph{backbone} (see Figure~\ref{backbone}), that can only be represented by a layout in a linear way, say horizontally. Furthermore, every clause gadget is vertically connected to its three literals, two of which are below and the other one above the backbone, or the other way around. The allocation of literal vertices to a variable $x_i$ is done by a path of all literal vertices corresponding to $x_i$ that is connected to the variable vertex for $x_i$. Such paths must lie either completely above or below the backbone. Interpreting the situation that a path lies above the backbone as assigning \emph{true} to the corresponding literal, yields a not-all-equal satisfying assignment, as it is not possible that all the paths for a clause lie on the same side of the backbone.\par
We assume that each clause of $F$ contains at most one negated variable, which is no restriction to not-all-equal satisfiability as a clause over literals $l_1,l_2,l_3$ is not-all-equal satisfied by an assignment if and only if a clause over literals $\bar l_1,\bar l_2,\bar l_3$ is. Furthermore, we also assume that every literal occurs at least three times in the formula. 
We first transform $F$ into $F' = \{c_1, \ldots, c_{2m}\}$, where $c_{m+i}=c_i$ for $i=1,\dots,m$. Then, we transform $F'$ into a graph $G=(V,E)$ as follows. The set of vertices is defined by $V = V_c \cup V_x\cup V_h$, where 
\begin{align*}
V_c =&\:\{c_j, c^1_j,c^2_j \mid 0 \leq j \leq 2m-1\} \cup \{c_{2m}\} \cup \{l_j^1,l_j^2,l_j^3 \mid 1 \leq j \leq 2m\}\,,\\
V_x =&\:\{x_i, x^1_i, x^2_i \mid 1 \leq i \leq n + 1\} \cup \{t_i, \overset{_{\rightarrow}}{t_i}, \overset{_{\leftarrow}}{t_i}, f^1_i, \overset{_{\rightarrow}}{f_i^1}, \overset{_{\leftarrow}}{f_i^1}, f^2_i, \overset{_{\rightarrow}}{f_i^2}, \overset{_{\leftarrow}}{f_i^2} \mid 1 \leq i \leq n\}\,,\\
V_h=&\:\{h_{t_i}^r,h_{f_i^1}^r,h_{f_i^2}^r \mid 1 \leq i \leq n, 0 \leq r \leq 4\}\,.
\end{align*}

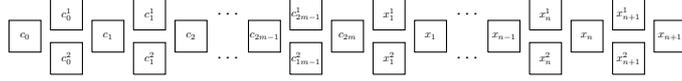
\begin{figure}
\centering
\begin{tikzpicture}[scale=0.7, transform shape]
\coordinate (width) at (0.6,0);
\coordinate (height) at (0,0.6);
\coordinate (labelshift) at (0.3, 0.3);

\coordinate (1) at ($(0,0)$);
\coordinate (2) at ($(1) + 2.6*(width)$);
\coordinate (3) at ($(1) + 0.7*(height) + 1.3*(width)$);
\coordinate (4) at ($(1) - 0.7*(height) + 1.3*(width)$);
\coordinate (5) at ($(2) + 2.6*(width)$);
\coordinate (6) at ($(2) + 0.7*(height) + 1.3*(width)$);
\coordinate (7) at ($(2) - 0.7*(height) + 1.3*(width)$);

\coordinate (8) at ($(6) + 3*(width) + 0.5*(height)$);
\coordinate (9) at ($(7) + 3*(width) + 0.5*(height)$);

\coordinate (10) at ($7.5*(width)$);
\coordinate (11) at ($(10) + 2.6*(width)$);
\coordinate (12) at ($(10) + 0.7*(height) + 1.3*(width)$);
\coordinate (13) at ($(10) - 0.7*(height) + 1.3*(width)$);
\coordinate (14) at ($(11) + 2.6*(width)$);
\coordinate (15) at ($(11) + 0.7*(height) + 1.3*(width)$);
\coordinate (16) at ($(11) - 0.7*(height) + 1.3*(width)$);

\coordinate (17) at ($(15) + 3*(width) + 0.5*(height)$);
\coordinate (18) at ($(16) + 3*(width) + 0.5*(height)$);

\coordinate (101) at ($15*(width)$);
\coordinate (111) at ($(101) + 2.6*(width)$);
\coordinate (121) at ($(101) + 0.7*(height) + 1.3*(width)$);
\coordinate (131) at ($(101) - 0.7*(height) + 1.3*(width)$);
\coordinate (141) at ($(111) + 2.6*(width)$);
\coordinate (151) at ($(111) + 0.7*(height) + 1.3*(width)$);
\coordinate (161) at ($(111) - 0.7*(height) + 1.3*(width)$);

\coordinate (temp) at (1);
\draw[black] ($(temp)$) -- ($(temp) + (width)$) -- ($(temp) + (width) + (height)$) -- ($(temp) + (height)$) -- ($(temp)$);
\node[scale=0.6] (label) at ($(temp) + (labelshift)$)  {$c_{0}$};

\coordinate (temp) at (2);
\draw[black] ($(temp)$) -- ($(temp) + (width)$) -- ($(temp) + (width) + (height)$) -- ($(temp) + (height)$) -- ($(temp)$);
\node[scale=0.6] (label) at ($(temp) + (labelshift)$)  {$c_{1}$};

\coordinate (temp) at (3);
\draw[black] ($(temp)$) -- ($(temp) + (width)$) -- ($(temp) + (width) + (height)$) -- ($(temp) + (height)$) -- ($(temp)$);
\node[scale=0.6] (label) at ($(temp) + (labelshift)$)  {$c_{0}^1$};

\coordinate (temp) at (4);
\draw[black] ($(temp)$) -- ($(temp) + (width)$) -- ($(temp) + (width) + (height)$) -- ($(temp) + (height)$) -- ($(temp)$);
\node[scale=0.6] (label) at ($(temp) + (labelshift)$)  {$c_{0}^2$};

\coordinate (temp) at (5);
\draw[black] ($(temp)$) -- ($(temp) + (width)$) -- ($(temp) + (width) + (height)$) -- ($(temp) + (height)$) -- ($(temp)$);
\node[scale=0.6] (label) at ($(temp) + (labelshift)$)  {$c_{2}$};

\coordinate (temp) at (6);
\draw[black] ($(temp)$) -- ($(temp) + (width)$) -- ($(temp) + (width) + (height)$) -- ($(temp) + (height)$) -- ($(temp)$);
\node[scale=0.6] (label) at ($(temp) + (labelshift)$)  {$c_{1}^1$};

\coordinate (temp) at (7);
\draw[black] ($(temp)$) -- ($(temp) + (width)$) -- ($(temp) + (width) + (height)$) -- ($(temp) + (height)$) -- ($(temp)$);
\node[scale=0.6] (label) at ($(temp) + (labelshift)$)  {$c_{1}^2$};

\coordinate (temp) at (10);
\draw[black] ($(temp)$) -- ($(temp) + (width)$) -- ($(temp) + (width) + (height)$) -- ($(temp) + (height)$) -- ($(temp)$);
\node[scale=0.6] (label) at ($(temp) + (labelshift)$)  {$c_{2m-1}$};

\coordinate (temp) at (11);
\draw[black] ($(temp)$) -- ($(temp) + (width)$) -- ($(temp) + (width) + (height)$) -- ($(temp) + (height)$) -- ($(temp)$);
\node[scale=0.6] (label) at ($(temp) + (labelshift)$)  {$c_{2m}$};

\coordinate (temp) at (12);
\draw[black] ($(temp)$) -- ($(temp) + (width)$) -- ($(temp) + (width) + (height)$) -- ($(temp) + (height)$) -- ($(temp)$);
\node[scale=0.6] (label) at ($(temp) + (labelshift)$)  {$c_{2m-1}^1$};

\coordinate (temp) at (13);
\draw[black] ($(temp)$) -- ($(temp) + (width)$) -- ($(temp) + (width) + (height)$) -- ($(temp) + (height)$) -- ($(temp)$);
\node[scale=0.6] (label) at ($(temp) + (labelshift)$)  {$c_{1m-1}^2$};

\coordinate (temp) at (14);
\draw[black] ($(temp)$) -- ($(temp) + (width)$) -- ($(temp) + (width) + (height)$) -- ($(temp) + (height)$) -- ($(temp)$);
\node[scale=0.6] (label) at ($(temp) + (labelshift)$)  {$x_1$};

\coordinate (temp) at (15);
\draw[black] ($(temp)$) -- ($(temp) + (width)$) -- ($(temp) + (width) + (height)$) -- ($(temp) + (height)$) -- ($(temp)$);
\node[scale=0.6] (label) at ($(temp) + (labelshift)$)  {$x_1^1$};

\coordinate (temp) at (16);
\draw[black] ($(temp)$) -- ($(temp) + (width)$) -- ($(temp) + (width) + (height)$) -- ($(temp) + (height)$) -- ($(temp)$);
\node[scale=0.6] (label) at ($(temp) + (labelshift)$)  {$x_1^2$};

\coordinate (temp) at (101);
\draw[black] ($(temp)$) -- ($(temp) + (width)$) -- ($(temp) + (width) + (height)$) -- ($(temp) + (height)$) -- ($(temp)$);
\node[scale=0.6] (label) at ($(temp) + (labelshift)$)  {$x_{n-1}$};

\coordinate (temp) at (111);
\draw[black] ($(temp)$) -- ($(temp) + (width)$) -- ($(temp) + (width) + (height)$) -- ($(temp) + (height)$) -- ($(temp)$);
\node[scale=0.6] (label) at ($(temp) + (labelshift)$)  {$x_n$};

\coordinate (temp) at (121);
\draw[black] ($(temp)$) -- ($(temp) + (width)$) -- ($(temp) + (width) + (height)$) -- ($(temp) + (height)$) -- ($(temp)$);
\node[scale=0.6] (label) at ($(temp) + (labelshift)$)  {$x_n^1$};

\coordinate (temp) at (131);
\draw[black] ($(temp)$) -- ($(temp) + (width)$) -- ($(temp) + (width) + (height)$) -- ($(temp) + (height)$) -- ($(temp)$);
\node[scale=0.6] (label) at ($(temp) + (labelshift)$)  {$x_n^2$};

\coordinate (temp) at (141);
\draw[black] ($(temp)$) -- ($(temp) + (width)$) -- ($(temp) + (width) + (height)$) -- ($(temp) + (height)$) -- ($(temp)$);
\node[scale=0.6] (label) at ($(temp) + (labelshift)$)  {$x_{n+1}$};

\coordinate (temp) at (151);
\draw[black] ($(temp)$) -- ($(temp) + (width)$) -- ($(temp) + (width) + (height)$) -- ($(temp) + (height)$) -- ($(temp)$);
\node[scale=0.6] (label) at ($(temp) + (labelshift)$)  {$x_{n+1}^1$};

\coordinate (temp) at (161);
\draw[black] ($(temp)$) -- ($(temp) + (width)$) -- ($(temp) + (width) + (height)$) -- ($(temp) + (height)$) -- ($(temp)$);
\node[scale=0.6] (label) at ($(temp) + (labelshift)$)  {$x_{n+1}^2$};

\node[draw=none]  at ($(8)$) {$\cdots$};
\node[draw=none]  at ($(9)$) {$\cdots$};
\node[draw=none]  at ($(17)$) {$\cdots$};
\node[draw=none]  at ($(18)$) {$\cdots$};
\end{tikzpicture}
\caption{The backbone-gadget.}
\label{backbone}
\end{figure}

Vertices $c_j$ and $x_i$ represent the corresponding clauses and variables and the  vertices $c^{r}_j$, $x^r_{i}$, $r \in \{1, 2\}$ are used to enforce the  \emph{backbone} structure as described at the beginning of this section.
The corresponding edges are implicitly defined, by requiring, for every  $0 \leq i \leq 2m-1$ and  $1 \leq i \leq n$, the following groups of $4$ vertices to form a $K_4$: $\{c_j, c^1_j, c^2_j, c_{j + 1}\}$, $\{x_i, x^1_{i + 1}, x^2_{i + 1}, x_{i + 1}\}$, 
and $\{c_{2m}, x^1_1, x^2_1, x_1\}$. Also, for every $j \in \{1, 2\}$, the vertices $c^j_0, c^j_1, \ldots, c^j_{2m - 1}, x^j_1, x^j_2, \ldots, x^j_{n + 1}$ form a path in this order. Consequently, these vertices form the subgraph represented by the layout in Figure~\ref{backbone}, which shall be the \emph{backbone}.
Vertices $t_i$, represent the literal $x_i$, $f^1_{i}$ represent the literal $\overline{x_i}$ in the first $m$ clauses, and $f^2_{i}$ represent the literal $\overline{x_i}$ in the remaining clauses. Vertices $l^1_{j}, l^2_{j}, l^3_{j}$ represent the literals of clause $c_j$. These roles are reflected with edges $\{x_i, t_i\}$, $\{x_i, f^1_i\}$, $\{x_i, f^2_i\}$ for all  $1 \leq i \leq n$ and $\{c_j, l^r_j\}$ for all $1 \leq j \leq 2m$ and $1 \leq r \leq 3$.
The connection between literals and variable assignments is build by turning $l_{j, r}$ with $y_{j, r} = x_i$ into a path connected to $t_i$; analogously, $l_{j, r}$ with $y_{j, r} = \overline{x_i}$ in the first (the last, respectively) $m$ clauses form a path connected to $f^1_{i}$ ($f^2_{i}$, respectively). More precisely, for every  $1 \leq j \leq 2m$,  $1 \leq i \leq n$ and  $1 \leq r \leq 3$:
 \begin{itemize}
 \item if $y_{j, r} = x_i$, there are edges $\{l^r_j, \overset{_{\rightarrow}}{t_i}\}$, $\{l^r_j, \overset{_{\leftarrow}}{t_i}\}$,
 \item if $y_{j,r} = \overline{x_i}$ and $1 \leq j \leq m$, there are edges $\{l^r_j, \overset{_{\rightarrow}}{f^1_i}\}$, $\{l^r_j, \overset{_{\leftarrow}}{f^1_i}\}$ and  $\{l^r_{j+m}, \overset{_{\rightarrow}}{f^2_i}\}$, $\{l^r_{j+m}, \overset{_{\leftarrow}}{f^2_i}\}$,
\item there are edges $\{t_i, \overset{_{\rightarrow}}{t_i}\}$, $\{t_i, \overset{_{\leftarrow}}{t_i}\}$ and $\{\overset{_{\rightarrow}}{t_i},h_{t_i}^p\}$,$\{\overset{_{\leftarrow}}{t_i},h_{t_i}^p\}$ for all $0\leq p\leq 4$,
\item there are edges $\{f_i^s, \overset{_{\rightarrow}}{f_i^s}\}$, $\{f_i, \overset{_{\leftarrow}}{f_i^s}\}$  and $\{\overset{_{\rightarrow}}{f_i^s},h_{f_i^s}^p\}$,$\{\overset{_{\leftarrow}}{f_i^s},h_{f_i^s}^p\}$ for all $0\leq p\leq 4$, $s\in \{1,2\}$,
\end{itemize}
Moreover, for every $i$, $1 \leq i \leq n$,  
\begin{itemize}
\item if $N(\overset{_{\rightarrow}}{t_i}) = \{h_{t_i}^1,h_{t_i}^2,l^{r_1}_{j_1}, l^{r_2}_{j_2}, \ldots, l^{r_{q}}_{j_{q}},h_{t_i}^0,t_i,h_{t_i}^3,h_{t_i}^4\}$ with $j_1 < j_2 < \ldots < j_q$, then these vertices form a path in this order,
\item if $N(\overset{_{\rightarrow}}{f^s_i}) = \{h_{f_i^s}^1,h_{f_i^s}^2,l^{r_1}_{j_1}, l^{r_2}_{j_2}, \ldots, l^{r_{q}}_{j_{q}},h_{f_i^s}^0,f_i^s,h_{f_i^s}^3,h_{f_i^s}^4\}$ with $j_1 < j_2 < \ldots < j_q$ and $s\in \{1,2\}$, then these vertices form a path in this order,
\end{itemize}

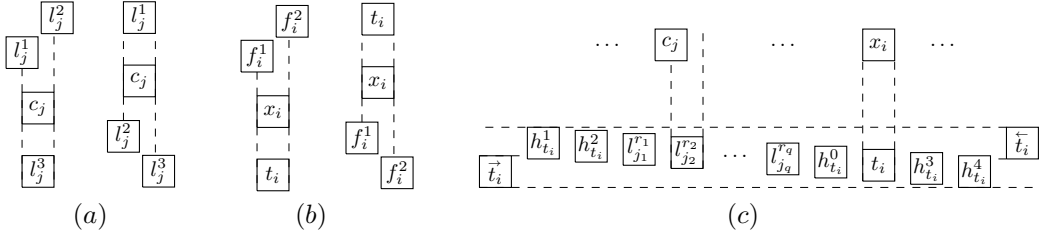
\begin{figure}
\centering
\begin{tabular}{ccc}
\begin{tikzpicture} [scale=0.7]

\coordinate (width) at (0.6,0);
\coordinate (height) at (0,0.6);  
\coordinate (labelshift) at (0.3, 0.3);

\coordinate (cj) at (0,0);
\coordinate (lj1) at ($(cj) - 0.5*(width) + 1.75*(height)$);
\coordinate (lj2) at ($(cj) + 0.6*(width) + 2.85*(height)$);
\coordinate (lj3) at ($(cj)- 2*(height)$);

\coordinate (temp) at (cj);
\draw[black] ($(temp)$) -- ($(temp) + (width)$) -- ($(temp) + (width) + (height)$) -- ($(temp) + (height)$) -- ($(temp)$);
\node[scale=0.8] (label) at ($(temp) + (labelshift)$)  {$c_j$};
\coordinate (temp) at (lj1);
\draw[black] ($(temp)$) -- ($(temp) + (width)$) -- ($(temp) + (width) + (height)$) -- ($(temp) + (height)$) -- ($(temp)$);
\node[scale=0.8] (label) at ($(temp) + (labelshift)$)  {$l^1_j$};
\coordinate (temp) at (lj2);
\draw[black] ($(temp)$) -- ($(temp) + (width)$) -- ($(temp) + (width) + (height)$) -- ($(temp) + (height)$) -- ($(temp)$);
\node[scale=0.8] (label) at ($(temp) + (labelshift)$)  {$l^2_j$};
\coordinate (temp) at (lj3);
\draw[black] ($(temp)$) -- ($(temp) + (width)$) -- ($(temp) + (width) + (height)$) -- ($(temp) + (height)$) -- ($(temp)$);
\node[scale=0.8] (label) at ($(temp) + (labelshift)$)  {$l^3_j$};

\draw[dashed, very thin] ($(cj)$) -- ($(cj |- lj1)$);
\draw[dashed, very thin] ($(cj) + (width)$) -- ($(cj |- lj2) + (width)$);
\draw[dashed, very thin] ($(cj)$) -- ($(cj |- lj3)$);
\draw[dashed, very thin] ($(cj) + (width)$) -- ($(cj |- lj3) + (width)$);

\end{tikzpicture}
\hspace{0.2cm}
\begin{tikzpicture} [scale=0.7]

\coordinate (width) at (0.6,0);
\coordinate (height) at (0,0.6);  
\coordinate (labelshift) at (0.3, 0.3);

\coordinate (cj) at (0,0);
\coordinate (lj1) at ($(cj) + 2*(height)$);
\coordinate (lj2) at ($(cj) - 0.5*(width) - 1.75*(height)$);
\coordinate (lj3) at ($(cj) + 0.6*(width) - 2.85*(height)$);

\coordinate (temp) at (cj);
\draw[black] ($(temp)$) -- ($(temp) + (width)$) -- ($(temp) + (width) + (height)$) -- ($(temp) + (height)$) -- ($(temp)$);
\node[scale=0.8] (label) at ($(temp) + (labelshift)$)  {$c_j$};
\coordinate (temp) at (lj1);
\draw[black] ($(temp)$) -- ($(temp) + (width)$) -- ($(temp) + (width) + (height)$) -- ($(temp) + (height)$) -- ($(temp)$);
\node[scale=0.8] (label) at ($(temp) + (labelshift)$)  {$l^1_j$};
\coordinate (temp) at (lj2);
\draw[black] ($(temp)$) -- ($(temp) + (width)$) -- ($(temp) + (width) + (height)$) -- ($(temp) + (height)$) -- ($(temp)$);
\node[scale=0.8] (label) at ($(temp) + (labelshift)$)  {$l^2_j$};
\coordinate (temp) at (lj3);
\draw[black] ($(temp)$) -- ($(temp) + (width)$) -- ($(temp) + (width) + (height)$) -- ($(temp) + (height)$) -- ($(temp)$);
\node[scale=0.8] (label) at ($(temp) + (labelshift)$)  {$l^3_j$};

\draw[dashed, very thin] ($(cj)$) -- ($(cj |- lj1)$);
\draw[dashed, very thin] ($(cj) + (width)$) -- ($(cj |- lj1) + (width)$);
\draw[dashed, very thin] ($(cj)$) -- ($(cj |- lj2) + (height)$);
\draw[dashed, very thin] ($(cj) + (width)$) -- ($(cj |- lj3) + (width) + (height)$);

\end{tikzpicture}
&
\hspace{0.3cm}
\begin{tikzpicture} [scale=0.7]

\coordinate (width) at (0.6,0);
\coordinate (height) at (0,0.6);  
\coordinate (labelshift) at (0.3, 0.3);

\coordinate (xi) at (0,0);
\coordinate (fi1) at ($(xi) - 0.5*(width) + 1.75*(height)$);
\coordinate (fi2) at ($(xi) + 0.6*(width) + 2.85*(height)$);
\coordinate (ti) at ($(xi)- 2*(height)$);

\coordinate (temp) at (xi);
\draw[black] ($(temp)$) -- ($(temp) + (width)$) -- ($(temp) + (width) + (height)$) -- ($(temp) + (height)$) -- ($(temp)$);
\node[scale=0.8] (label) at ($(temp) + (labelshift)$)  {$x_i$};
\coordinate (temp) at (ti);
\draw[black] ($(temp)$) -- ($(temp) + (width)$) -- ($(temp) + (width) + (height)$) -- ($(temp) + (height)$) -- ($(temp)$);
\node[scale=0.8] (label) at ($(temp) + (labelshift)$)  {$t_i$};
\coordinate (temp) at (fi1);
\draw[black] ($(temp)$) -- ($(temp) + (width)$) -- ($(temp) + (width) + (height)$) -- ($(temp) + (height)$) -- ($(temp)$);
\node[scale=0.8] (label) at ($(temp) + (labelshift)$)  {$f^1_i$};
\coordinate (temp) at (fi2);
\draw[black] ($(temp)$) -- ($(temp) + (width)$) -- ($(temp) + (width) + (height)$) -- ($(temp) + (height)$) -- ($(temp)$);
\node[scale=0.8] (label) at ($(temp) + (labelshift)$)  {$f^2_i$};

\draw[dashed, very thin] ($(xi)$) -- ($(xi |- fi1)$);
\draw[dashed, very thin] ($(xi) + (width)$) -- ($(xi |- fi2) + (width)$);
\draw[dashed, very thin] ($(xi)$) -- ($(xi |- ti)$);
\draw[dashed, very thin] ($(xi) + (width)$) -- ($(xi |- ti) + (width)$);

\end{tikzpicture}
\hspace{0.2cm}
\begin{tikzpicture} [scale=0.7]

\coordinate (width) at (0.6,0);
\coordinate (height) at (0,0.6);  
\coordinate (labelshift) at (0.3, 0.3);

\coordinate (xi) at (0,0);
\coordinate (ti) at ($(xi) + 2*(height)$);
\coordinate (fi1) at ($(xi) - 0.5*(width) - 1.75*(height)$);
\coordinate (fi2) at ($(xi) + 0.6*(width) - 2.85*(height)$);

\coordinate (temp) at (xi);
\draw[black] ($(temp)$) -- ($(temp) + (width)$) -- ($(temp) + (width) + (height)$) -- ($(temp) + (height)$) -- ($(temp)$);
\node[scale=0.8] (label) at ($(temp) + (labelshift)$)  {$x_i$};
\coordinate (temp) at (ti);
\draw[black] ($(temp)$) -- ($(temp) + (width)$) -- ($(temp) + (width) + (height)$) -- ($(temp) + (height)$) -- ($(temp)$);
\node[scale=0.8] (label) at ($(temp) + (labelshift)$)  {$t_i$};
\coordinate (temp) at (fi1);
\draw[black] ($(temp)$) -- ($(temp) + (width)$) -- ($(temp) + (width) + (height)$) -- ($(temp) + (height)$) -- ($(temp)$);
\node[scale=0.8] (label) at ($(temp) + (labelshift)$)  {$f^1_i$};
\coordinate (temp) at (fi2);
\draw[black] ($(temp)$) -- ($(temp) + (width)$) -- ($(temp) + (width) + (height)$) -- ($(temp) + (height)$) -- ($(temp)$);
\node[scale=0.8] (label) at ($(temp) + (labelshift)$)  {$f^2_i$};

\draw[dashed, very thin] ($(xi)$) -- ($(xi |- ti)$);
\draw[dashed, very thin] ($(xi) + (width)$) -- ($(xi |- ti) + (width)$);
\draw[dashed, very thin] ($(xi)$) -- ($(xi |- fi1) + (height)$);
\draw[dashed, very thin] ($(xi) + (width)$) -- ($(xi |- fi2) + (width) + (height)$);
\end{tikzpicture}
&
\hspace{0.3cm}
\begin{tikzpicture}[scale=0.7]

\coordinate (width) at (0.6,0);
\coordinate (height) at (0,0.6);

\coordinate (tright) at (0,0);
\coordinate (ht1) at ($(tright) + 1.5*(width) + 0.9*(height)$);
\coordinate (ht2) at ($(ht1) + 1.5*(width) - 0.1*(height)$);
\coordinate (l1) at ($(ht2) + 1.5*(width) - 0.1*(height)$);
\coordinate (l2) at ($(l1) + 1.5*(width) - 0.1*(height)$);
\coordinate (dots) at ($(l2) + 1.5*(width) - 0.1*(height)$);
\coordinate (lq) at ($(dots) + 1.5*(width) - 0.1*(height)$);
\coordinate (ht0) at ($(lq) + 1.5*(width) - 0.1*(height)$);
\coordinate (t) at ($(ht0) + 1.5*(width) - 0.1*(height)$);
\coordinate (ht3) at ($(t) + 1.5*(width) - 0.1*(height)$);
\coordinate (ht4) at ($(ht3) + 1.5*(width) - 0.1*(height)$);
\coordinate (tleft) at ($(ht4) + 1.5*(width) + 0.9*(height)$);

\coordinate (temp) at (tright);
\draw[black] ($(temp)$) -- ($(temp) + (width)$) -- ($(temp) + (width) + (height)$) -- ($(temp) + (height)$) -- ($(temp)$);
\node[scale=0.8] (label) at ($(temp) + (labelshift)$)  {$\overset{_{\rightarrow}}{t_i}$};

\coordinate (temp) at (ht1);
\draw[black] ($(temp)$) -- ($(temp) + (width)$) -- ($(temp) + (width) + (height)$) -- ($(temp) + (height)$) -- ($(temp)$);
\node[scale=0.8] (label) at ($(temp) + (labelshift)$)  {$h^1_{t_i}$};

\coordinate (temp) at (ht2);
\draw[black] ($(temp)$) -- ($(temp) + (width)$) -- ($(temp) + (width) + (height)$) -- ($(temp) + (height)$) -- ($(temp)$);
\node[scale=0.8] (label) at ($(temp) + (labelshift)$)  {$h^2_{t_i}$};

\coordinate (temp) at (l1);
\draw[black] ($(temp)$) -- ($(temp) + (width)$) -- ($(temp) + (width) + (height)$) -- ($(temp) + (height)$) -- ($(temp)$);
\node[scale=0.8] (label) at ($(temp) + (labelshift)$)  {$l^{r_1}_{j_1}$};

\coordinate (temp) at (l2);
\draw[black] ($(temp)$) -- ($(temp) + (width)$) -- ($(temp) + (width) + (height)$) -- ($(temp) + (height)$) -- ($(temp)$);
\node[scale=0.8] (label) at ($(temp) + (labelshift)$)  {$l^{r_2}_{j_2}$};

\coordinate (temp) at (dots);
\node[scale=0.8] (label) at ($(temp) + (labelshift)$)  {$\ldots$};

\coordinate (temp) at (lq);
\draw[black] ($(temp)$) -- ($(temp) + (width)$) -- ($(temp) + (width) + (height)$) -- ($(temp) + (height)$) -- ($(temp)$);
\node[scale=0.8] (label) at ($(temp) + (labelshift)$)  {$l^{r_q}_{j_q}$};

\coordinate (temp) at (ht0);
\draw[black] ($(temp)$) -- ($(temp) + (width)$) -- ($(temp) + (width) + (height)$) -- ($(temp) + (height)$) -- ($(temp)$);
\node[scale=0.8] (label) at ($(temp) + (labelshift)$)  {$h^0_{t_i}$};

\coordinate (temp) at (t);
\draw[black] ($(temp)$) -- ($(temp) + (width)$) -- ($(temp) + (width) + (height)$) -- ($(temp) + (height)$) -- ($(temp)$);
\node[scale=0.8] (label) at ($(temp) + (labelshift)$)  {$t_i$};
 
\coordinate (temp) at (ht3);
\draw[black] ($(temp)$) -- ($(temp) + (width)$) -- ($(temp) + (width) + (height)$) -- ($(temp) + (height)$) -- ($(temp)$);
\node[scale=0.8] (label) at ($(temp) + (labelshift)$)  {$h^3_{t_i}$};

\coordinate (temp) at (ht4);
\draw[black] ($(temp)$) -- ($(temp) + (width)$) -- ($(temp) + (width) + (height)$) -- ($(temp) + (height)$) -- ($(temp)$);
\node[scale=0.8] (label) at ($(temp) + (labelshift)$)  {$h^4_{t_i}$};

\coordinate (temp) at (tleft);
\draw[black] ($(temp)$) -- ($(temp) + (width)$) -- ($(temp) + (width) + (height)$) -- ($(temp) + (height)$) -- ($(temp)$);
\node[scale=0.8] (label) at ($(temp) + (labelshift)$)  {$\overset{_{\leftarrow}}{t_i}$};
 
\draw[dashed, very thin] ($(tright)$) -- ($(tleft |- tright) + (width)$);
\draw[dashed, very thin] ($(tright) + (height)$) -- ($(ht1 |- tright) + (height)$);
\draw[dashed, very thin] ($(tleft)$) -- ($(ht4 |- tleft) + (width)$);
\draw[dashed, very thin] ($(tleft) + (height)$) -- ($(tright |- tleft) + (height)$);
 
\coordinate (xi) at ($(t |- tright) + 4*(height)$);
\coordinate (cj) at ($(l2 |- tright) + 4*(height) - 0.5*(width)$);

\coordinate (temp) at (xi);
\draw[black] ($(temp)$) -- ($(temp) + (width)$) -- ($(temp) + (width) + (height)$) -- ($(temp) + (height)$) -- ($(temp)$);
\node[scale=0.8] (label) at ($(temp) + (labelshift)$)  {$x_i$};

\coordinate (temp) at (cj);
\draw[black] ($(temp)$) -- ($(temp) + (width)$) -- ($(temp) + (width) + (height)$) -- ($(temp) + (height)$) -- ($(temp)$);
\node[scale=0.8] (label) at ($(temp) + (labelshift)$)  {$c_j$};

\draw[dashed, very thin] ($(t)$) -- ($(t |- xi)$);
\draw[dashed, very thin] ($(t) + (width)$) -- ($(t |- xi) + (width)$);
\draw[dashed, very thin] ($(l2)$) -- ($(l2 |- cj)$);
\draw[dashed, very thin] ($(l2) + (width)$) -- ($(l2 |- cj) + (width) + (height)$);
 
\coordinate (dots1) at ($(cj) - 2*(width)$);
\coordinate (dots2) at ($(cj) + 3.5*(width)$);
\coordinate (dots3) at ($(xi) +2*(width)$);

\node[scale=0.8] (label) at ($(dots1) + (labelshift)$)  {$\ldots$};
\node[scale=0.8] (label) at ($(dots2) + (labelshift)$)  {$\ldots$};
\node[scale=0.8] (label) at ($(dots3) + (labelshift)$)  {$\ldots$};
 
\end{tikzpicture} 
\cr
\centering$(a)$&\centering$(b)$&\centering$(c)$

\end{tabular}

\caption{Possible placements of literal vertices, possible placements of assignment vertices, and the clause path for $x_i$.}
\label{gadgetsFigure}

\end{figure}

Next, we assume that the formula $F'$ is not-all-equal satisfiable and show how a layout for $G$ can be constructed. First, we represent the backbone as illustrated in Figure~\ref{backbone}. If a variable $x_i$ is assigned the value \emph{true}, then we place the unit squares $R_{\{x_i, t_i, f^1_{i}, f^2_{i}\}}$ as illustrated on the left side of Figure~\ref{gadgetsFigure}$(b)$, and otherwise as illustrated on the right side. 
The edges for the vertices $t_i, \overset{_{\rightarrow}}{t_i}, \overset{_{\leftarrow}}{t_i}, h_{t_i}^r$, $0 \leq r \leq 4$, and all $l^r_{j}$ with $y_{j, r} = x_i$ can be realised as illustrated in Figure~\ref{gadgetsFigure}$(c)$ (either placed above or below the backbone, according to the position of $R_{t_i}$). An analogous construction applies to the unit squares for $l^r_{j}$ with $y_{j, r} = \overline{x_i}$, with the only difference that we have two such paths (one for the first $m$ clauses and one for the remaining clauses) and that they both lie on the opposite side of the backbone with respect to $R_{t_i}$. Moreover, in these paths, the $R_{l^r_{j}}$ must be horizontally shifted such that they can see their corresponding $R_{c_j}$ from above or from below, according to whether the path lies above or below the backbone (as indicated in Figure~\ref{gadgetsFigure}$(c)$). As long as not all paths for the three literals of the same clause lie all above or all below the backbone, this is possible by arranging the unit squares as illustrated in Figure~\ref{gadgetsFigure}$(a)$. However, if for some clause all paths lie on the same side of the backbone, then the literals of the clause are either all set to \emph{true} or all set to \emph{false}, which is a contradiction to the assumption that the assignment is not-all-equal satisfiable. Consequently, we can represent $G$ as described. A formal definition of the layout and a full example of the reduction can be found in the Appendix.

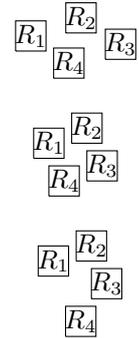
\begin{wrapfigure}{R}{2.4cm}
\centering
\begin{tikzpicture}

\coordinate (width) at (0.4,0);
\coordinate (height) at (0,0.4);
\coordinate (labelshift) at (0.2, 0.2);

\coordinate (1) at (0,0);
\coordinate (2) at ($(1) + 1.5*(height) + 0.4*(width)$);
\coordinate (3) at ($(1) + 1.25*0.5*(height) + 1.7*(width)$);
\coordinate (4) at ($(1) + 0.9*(height) - 1.25*(width)$);

\coordinate (temp) at (1);
\draw[black] ($(temp)$) -- ($(temp) + (width)$) -- ($(temp) + (width) + (height)$) -- ($(temp) + (height)$) -- ($(temp)$);
\node (label) at ($(temp) + (labelshift)$)  {$R_4$};

\coordinate (temp) at (2);
\draw[black] ($(temp)$) -- ($(temp) + (width)$) -- ($(temp) + (width) + (height)$) -- ($(temp) + (height)$) -- ($(temp)$);
\node (label) at ($(temp) + (labelshift)$)  {$R_2$};

\coordinate (temp) at (3);
\draw[black] ($(temp)$) -- ($(temp) + (width)$) -- ($(temp) + (width) + (height)$) -- ($(temp) + (height)$) -- ($(temp)$);
\node (label) at ($(temp) + (labelshift)$)  {$R_3$};

\coordinate (temp) at (4);
\draw[black] ($(temp)$) -- ($(temp) + (width)$) -- ($(temp) + (width) + (height)$) -- ($(temp) + (height)$) -- ($(temp)$);
\node (label) at ($(temp) + (labelshift)$)  {$R_1$};

\end{tikzpicture}
\vspace{0.3cm}

\begin{tikzpicture}

\coordinate (width) at (0.4,0);
\coordinate (height) at (0,0.4);
\coordinate (labelshift) at (0.2, 0.2);

\coordinate (1) at ($(0,0)$);
\coordinate (2) at ($(1) + 1.25*(width) + 0.5*(height)$);
\coordinate (3) at ($(1) + 0.5*(width) - 1.25*(height)$);
\coordinate (4) at ($(1) + 1.75*(width) - 0.75*(height)$);

\coordinate (temp) at (1);
\draw[black] ($(temp)$) -- ($(temp) + (width)$) -- ($(temp) + (width) + (height)$) -- ($(temp) + (height)$) -- ($(temp)$);
\node (label) at ($(temp) + (labelshift)$)  {$R_1$};

\coordinate (temp) at (2);
\draw[black] ($(temp)$) -- ($(temp) + (width)$) -- ($(temp) + (width) + (height)$) -- ($(temp) + (height)$) -- ($(temp)$);
\node (label) at ($(temp) + (labelshift)$)  {$R_2$};

\coordinate (temp) at (3);
\draw[black] ($(temp)$) -- ($(temp) + (width)$) -- ($(temp) + (width) + (height)$) -- ($(temp) + (height)$) -- ($(temp)$);
\node (label) at ($(temp) + (labelshift)$)  {$R_4$};

\coordinate (temp) at (4);
\draw[black] ($(temp)$) -- ($(temp) + (width)$) -- ($(temp) + (width) + (height)$) -- ($(temp) + (height)$) -- ($(temp)$);
\node (label) at ($(temp) + (labelshift)$)  {$R_3$};

\end{tikzpicture}
\vspace{0.3cm}

\begin{tikzpicture}

\coordinate (width) at (0.4,0);
\coordinate (height) at (0,0.4);
\coordinate (labelshift) at (0.2, 0.2);

\coordinate (1) at ($(0,0)$);
\coordinate (2) at ($(1) + 1.25*(width) + 0.5*(height)$);
\coordinate (3) at ($(1) + 0.9*(width) - 2*(height)$);
\coordinate (4) at ($(1) + 1.75*(width) - 0.75*(height)$);

\coordinate (temp) at (1);
\draw[black] ($(temp)$) -- ($(temp) + (width)$) -- ($(temp) + (width) + (height)$) -- ($(temp) + (height)$) -- ($(temp)$);
\node (label) at ($(temp) + (labelshift)$)  {$R_1$};

\coordinate (temp) at (2);
\draw[black] ($(temp)$) -- ($(temp) + (width)$) -- ($(temp) + (width) + (height)$) -- ($(temp) + (height)$) -- ($(temp)$);
\node (label) at ($(temp) + (labelshift)$)  {$R_2$};

\coordinate (temp) at (3);
\draw[black] ($(temp)$) -- ($(temp) + (width)$) -- ($(temp) + (width) + (height)$) -- ($(temp) + (height)$) -- ($(temp)$);
\node (label) at ($(temp) + (labelshift)$)  {$R_4$};

\coordinate (temp) at (4);
\draw[black] ($(temp)$) -- ($(temp) + (width)$) -- ($(temp) + (width) + (height)$) -- ($(temp) + (height)$) -- ($(temp)$);
\node (label) at ($(temp) + (labelshift)$)  {$R_3$};

\end{tikzpicture}
\caption{Re- presenting $K_4$.}\label{visLayoutsCompleteBipartiteGraphsFigure}
\end{wrapfigure}

\begin{lemma}\label{nonGridReductionEasyDirectionLemma}
If $F$ is not-all-equal satisfiable, then $G \in \unitSquareGraphs$.                                       
\end{lemma}

Proving that a layout for $G$ translates into a satisfying not-all-equal assignment for $F$, is much more involved. The general idea is to show that any layout for $G$ must be \visomorphic{} to the layout constructed above. However, this cannot be done separately for the individual gadgets, e.\,g., showing that the backbone must be represented as in Figure~\ref{backbone} (in fact, the structure of the backbone alone does not enforce such a layout) and the literal vertices must form a path as in Figure~\ref{gadgetsFigure}$(c)$ and so on. Instead, the desired structure of the layout is only enforced by a rather complicated interplay of the different parts of $G$.\par
A main building stone is that a $K_4$ can only be represented in $3$ different ways (up to \visomorphism{}), which are illustrated in Figure~\ref{visLayoutsCompleteBipartiteGraphsFigure}. This observation is important, since the backbone is a sequence of $K_4$.

\begin{lemma}\label{K4Lemma}
Every layout for $K_4$ is \visomorphic{} to one of the three layouts of Figure~\ref{visLayoutsCompleteBipartiteGraphsFigure}.
\end{lemma}

We now assume that $G$ can be represented by some layout $\mathcal{R}$. For every $j$, $1 \leq j \leq m$, we define $L_j = \{l^1_{j}, l^2_{j}, l^3_{j}\}$, for every $i$, $1 \leq i \leq n$, we define $A_{i} = \{t_i, f^1_i, f^2_i\}$, and, for every $j$, $1 \leq j \leq m-1$, we define $C^l_{j} = \{c_j, c_{j-1}, c^1_{j-1}, c^2_{j-1}\}$, $C^r_{j} = \{c_j, c_{j+1}, c^1_j, c^2_j\}$ and $C_j = C^l_j \cup C^r_j$.\par
We shall prove the desired structure of $\mathcal{R}$ by first considering the neighbourhood of $c_j$; once we have fixed the layout for this subgraph, the structure of the whole layout can be concluded inductively. The neighbourhood of $c_j$ consists of $C^l_{j}$ and $C^r_{j}$ (two $K_4$ joined by $c_j$) and $L_j$, where all vertices of the two $K_4$ (except $c_j$) are not connected to any vertex of $L_j$. Intuitively speaking, this independence between $L_j$ and the $K_4$ of the backbone will force the backbone to expand along one dimension, say horizontally (as depicted in Figure~\ref{backbone}), while the visibilities between $L_j$ and $c_j$ must then be vertical (as depicted in Figure~\ref{gadgetsFigure}$(a)$). However, formally proving this turns out to be quite complicated.\par
The general proof idea is to somehow place the unit squares of~$R_{L_{j}}$ in such a way that they see $R_{c_j}$ without creating unwanted visibilities. Then, the areas of visibility for the $R_{L_j}$ are blocked for any unit squares from the backbone-neighbourhood $R_{C_{j}}$, since these are independent of~$R_{L_j}$. 
For example, consider the situation depicted in Figure~\ref{proofSketchFigure}. Here, placing unit squares from $R_{C_j}$ in the grey areas implies that they are within visibility of some unit squares from $R_{L_j}$. This leaves only few possibilities to place the unit squares from $R_{C_j}$ and by applying arguments of this type, it can be concluded, by exhaustively searching all possibilities and under application of Lemma~\ref{K4Lemma}, that the only possible layouts have the above described form.\par
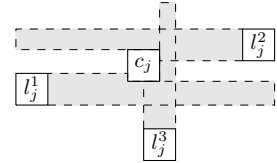
\begin{wrapfigure}{R}{0.3\textwidth}
\centering
\begin{tikzpicture}[scale=0.7]

\coordinate (width) at (0.6,0);
\coordinate (height) at (0,0.6);
\coordinate (labelshift) at (0.3, 0.3);

\coordinate (cj) at (0,0);
\coordinate (lj1) at ($(cj) - 3.5*(width) - 0.75*(height)$);
\coordinate (lj2) at ($(cj) + 3.6*(width) + 0.65*(height)$);
\coordinate (lj3) at ($(cj)- 2.5*(height) + 0.5*(width)$);

\draw[fill, black!10] ($(lj1) + (width)$) -- ($(lj2 |- lj1) + (width)$) -- ($(lj2 |- cj) + (width)$) -- ($(cj |- cj) + (width)$) -- ($(cj |- cj)$) -- ($(cj |- lj1) + (height)$) -- ($(lj1 |- lj1) + (height) + (width)$);
\draw[fill, black!10] ($(lj3) + (height)$) -- ($(lj3 |- cj)$) -- ($(cj |- cj) + (width)$) -- ($(cj |- lj2) + (width) + (height) + (0, 0.5)$) -- ($(lj3 |- lj2) + (width) + (height) + (0, 0.5)$) -- ($(lj3 |- lj3) + (width) + (height)$);
\draw[fill, black!10] ($(lj2)$) -- ($(cj |- lj2) + (width)$) -- ($(cj |- cj) + (width) + (height)$) -- ($(lj1 |- cj) + (height)$) -- ($(lj1 |- lj2) + (height)$) -- ($(lj2 |- lj2) + (height)$);

\draw[dashed, very thin] ($(lj1) + (width)$) -- ($(lj2 |- lj1) + (width)$) -- ($(lj2 |- cj) + (width)$) -- ($(cj |- cj) + (width)$) -- ($(cj |- cj)$) -- ($(cj |- lj1) + (height)$) -- ($(lj1 |- lj1) + (height) + (width)$);
\draw[dashed, very thin] ($(lj3) + (height)$) -- ($(lj3 |- cj)$) -- ($(cj |- cj) + (width)$) -- ($(cj |- lj2) + (width) + (height) + (0, 0.5)$) -- ($(lj3 |- lj2) + (width) + (height) + (0, 0.5)$) -- ($(lj3 |- lj3) + (width) + (height)$);
\draw[dashed, very thin] ($(lj2)$) -- ($(cj |- lj2) + (width)$) -- ($(cj |- cj) + (width) + (height)$) -- ($(lj1 |- cj) + (height)$) -- ($(lj1 |- lj2) + (height)$) -- ($(lj2 |- lj2) + (height)$);

\coordinate (temp) at (cj);
\draw[black] ($(temp)$) -- ($(temp) + (width)$) -- ($(temp) + (width) + (height)$) -- ($(temp) + (height)$) -- ($(temp)$);
\node[scale=0.8] (label) at ($(temp) + (labelshift)$)  {$c_j$};
\coordinate (temp) at (lj1);
\draw[black] ($(temp)$) -- ($(temp) + (width)$) -- ($(temp) + (width) + (height)$) -- ($(temp) + (height)$) -- ($(temp)$);
\node[scale=0.8] (label) at ($(temp) + (labelshift)$)  {$l^1_j$};
\coordinate (temp) at (lj2);
\draw[black] ($(temp)$) -- ($(temp) + (width)$) -- ($(temp) + (width) + (height)$) -- ($(temp) + (height)$) -- ($(temp)$);
\node[scale=0.8] (label) at ($(temp) + (labelshift)$)  {$l^2_j$};
\coordinate (temp) at (lj3);
\draw[black] ($(temp)$) -- ($(temp) + (width)$) -- ($(temp) + (width) + (height)$) -- ($(temp) + (height)$) -- ($(temp)$);
\node[scale=0.8] (label) at ($(temp) + (labelshift)$)  {$l^3_j$};
\end{tikzpicture}
\caption{Possible placement of literal vertices for $c_j$.}
\label{proofSketchFigure}
\end{wrapfigure}
However, this argument is flawed: it is possible to place a unit square $R_x$ within the grey areas, as long as the forbidden visibilities are blocked by other unit squares. This type of blocking would require a path between $x$ and $c_j$ or some vertex from $L_j$, respectively, which does exist as structure in $G$. Consequently, in order to make the above described argument applicable, we first have to show that the existence of such visibility-blocking unit squares leads to a contradiction. This substantially increases the combinatorial depth of the already technical proof idea described above.\par
For the next lemma, which is the main tool in proving how the neighbourhood of $c_j$ is represented, we need some notation. Let $R_i$, $R_j$, $R_k$ be unit squares. 
If some (or every) visibility rectangle for $R_i$ and $R_k$ intersects $R_j$, then $R_j$ is \emph{strictly between} $R_i$ and $R_k$ (or $R_j$ \emph{blocks the view} between $R_i$ and $R_k$, respectively).

\begin{lemma}\label{between_all}
For all $1\leq i\leq 2m$ and $r\in\{1,2,3\}$ and every $z\in N(c_i)\setminus\{l_i^r\}$ there exists no visibility rectangle for $R_{l_i^r}$ and $R_z$ that is not intersected by $R_{c_i}$. In particular, this implies: $R_z$ is not strictly between $R_{c_i}$ and $R_{l_i^r}$, $R_{l_i^r}$ is not strictly between $R_{c_i}$ and $R_z$, and, if $R_{c_i}$ is strictly between $R_{l_i^r}$ and $R_z$, then $R_{c_i}$ blocks the view between $R_{l_i^r}$ and $R_z$.
\end{lemma}

By applying Lemma~\ref{between_all}, we can now show that $R_{C^l_{j}}$ and $R_{C^r_{j}}$ cannot all see $R_{c_j}$ from the same side, which can then be used in order to prove that either all $R_{L_j}$ see $R_{c_j}$ vertically or all of them see $R_{c_j}$ horizontally:

\begin{lemma}\label{notAllOneSideLemma}
For every $j$, $1 \leq j \leq 2m - 1$ and $y \in C_j \setminus \{c_j\}$, $R_{c_j} \Hvis R_{C_j \setminus \{y, c_j\}}$ is not possible.
\end{lemma} 

\begin{lemma}\label{literalsLemma}
For every $j$, $1 \leq j \leq m$, either $R_{c_j} \symHvis R_{L_j}$ or $R_{c_j} \symVvis R_{L_j}$.
\end{lemma}

We are now able to combine these lemmas in order to prove that a layout for $G$ translates into a not-all-equal satisfying assignment for the formula $F$. To this end, we first note that the neighbourhood of a variable vertex $x_i$ has an identical structure as the neighbourhood of the clause vertices, which implies that Lemmas~\ref{between_all},~\ref{notAllOneSideLemma}~and~\ref{literalsLemma} also apply to this part of the graph. By combining Lemmas~\ref{between_all}~and~\ref{literalsLemma}, we can show that for each clause $c_j$, either $R_{c_j}\symHvis R_{C_j \setminus \{c_j\}}$ or  $R_{c_j}\symVvis R_{C_j \setminus \{c_j\}}$. By Lemma~\ref{K4Lemma}, this means that the two corresponding induced $K_4$ are represented as shown in Figure~\ref{backbone}, and, furthermore, an inductive application of Lemma~\ref{notAllOneSideLemma} forces them to form the shown horizontal or vertical backbone. Due to Lemma~\ref{literalsLemma}, the literal vertices and the assignment vertices corresponding to the same variable must all form a path on the same side of the backbone. We can now assign $x_i$ the value \emph{true} if and only if $R_{t_i}$ is below the backbone. As long as, for the variables occurring in some clause $c_j$, $R_{f^1_{i}}$ is on the opposite side of $R_{t_i}$, clause $c_j$ is not-all-equal satisfied, because then literals are set to \emph{true} if and only if they are below the backbone and, due to Lemma~\ref{between_all}, it is not possible that they all lie on the same side. However, if $R_{f^1_{i}}$ lies on the same side as $R_{t_i}$, which is possible, then $R_{f^2_{i}}$, again due to Lemma~\ref{between_all}, must lie on the opposite side of $R_{t_i}$ and, by the same argument, it follows that $c_{j + m}$, which is a copy of $c_j$,  is not-all-equal satisfied (note that every clause has at most one negated variable). 

\begin{lemma}\label{final}
If $G \in \unitSquareGraphs$, then $F$ is not-all-equal satisfiable.
\end{lemma}

\begin{theorem}\label{nonGridNPCompletenessTheorem}
$\recognitionProb(\unitSquareGraphs)$ is $\npclass$-complete.
\end{theorem}

Since in our reduction the size of the graph is linear in the size of the formula, we can also conclude ETH-lower bounds for $\recognitionProb(\unitSquareGraphs)$.

\section{Conclusions}

The hardness of $\recognitionProb(\unitSquareGraphs_{\weak})$ is still open (note that in our reduction, we heavily used the argument that certain constellations yield forbidden edges, which falls apart in the weak case) and we conjecture it to be $\npclass$-hard as well. Two open problems concerning graph classes related to $\unitSquareGridGraphs$ are mentioned in Section~\ref{gridCaseSection}: (1) are $\unitSquareGridGraphs$ and the class of resolution-$\frac{\pi}{2}$ graphs identical, (2) are there resolution-$\frac{\pi}{2}$ graphs without BRAC-drawing? {Note that a positive answer to (2) gives a negative answer to (1).}\par
From a parameterised complexity point of view, our $\npclass$-completeness result shows that the number of different rectangle shapes (considered as a parameter) has no influence on the hardness of recognition. Another interesting parameter to explore would be the step size of the grid, i.\,e., for $k \in \mathbb{N}$, let $\unitSquareGridGraphs^k$ be defined like $\unitSquareGridGraphs$, but for a $\{\frac{\ell}{k} \mid \ell \in \mathbb{N}\}^2$ grid. We note that these classes form an infinite hierarchy between $\unitSquareGridGraphs = \unitSquareGridGraphs^1$ and $\unitSquareGraphs = \bigcup_k \unitSquareGridGraphs^k$, and it is hard to define them in terms of extensions of rectilinear graphs. Another interesting observation is that the hardness reduction for the recognition problem of rectilinear graphs from~\cite{EadHonPoo0910}, if interpreted as reduction for $\recognitionProb(\unitSquareGridGraphs)$, does not work for $\unitSquareGridGraphs^2$. The classes $\unitSquareGridGraphs^k$ might be practically more relevant, since placing objects in the plane with discrete distances is more realistic.

\subsection*{Acknowledgements}
We acknowledge the support of the first author by the Deutsche Forschungsgemeinschaft, grant FE 560/6-1, and the support of the last author by the NSERC Discovery Grant program of Canada. We thank the organizers of the Lorentz Center workshop `Fixed Parameter Computational Geometry' for the great atmosphere that stimulated this project.

\bibliographystyle{plainurl}
\bibliography{main_biblio}

\section*{Appendix}

\subsection*{A Minor Modification of the Reduction from~\cite{EadHonPoo0910}}

The reduction from~\cite{EadHonPoo0910}, which proves the recognition problem for rectilinear graphs $\npclass$-hard, transforms a $\threeSat$ instance with $n$ variables and $m$ clauses into a graph of size $\bigo(n \cdot m)$.\footnote{Since rectilinear graphs have maximum degree $4$, we measure their size in the number of vertices.} The main part of this graph is an L-shaped frame of size $\bigo(n + m)$ (containing $n$ connecting ports in its horizontal and $m$ connecting ports in its vertical arm) and, for every variable $x_i$, a tower with $m$ levels. These levels are aligned with the $m$ clause-ports and are connected by edges only if the clause contains this variable or its negation. Consequently, in every variable tower for $x_i$, only those levels matter that correspond to clauses which contain $x_i$ (or $\overline{x_i}$) and the rest can be ignored. In fact, simply removing those superfluous levels result in a reduction that works in the same way, but constructs a graph of size $\bigo(m)$.\par
With this linear reduction from $\threeSat$, we can conclude that the recognition problem cannot be solved in time $2^{\smallo(n)}$, unless ETH fails. On the other hand, as explained at the beginning of Section~\ref{gridCaseSection}, a $2^{\bigo(n)}$ algorithm follows easily by enumerating all possible $\LRDU$-restrictions and then applying the algorithm from~\cite{ManPPT1011} for solving recognition of $\LRDU$-restricted graphs (in time $\bigo(|E| \cdot |V|)$).

\subsection*{Proof of Lemma~\ref{removeVerticesEdgesLemma}}

%\begin{lemma}
%\noindent\emph{Let $G = (V, E) \in \unitSquareGridGraphs$, let $v \in V$ and $e \in E$.\\ Then $(V, E \setminus \{e\}) \in \unitSquareGridGraphs$ and $(V \setminus \{v\}, E) \in \unitSquareGridGraphs$.} 
%\end{lemma}

\begin{proof}
We first prove the first statement. To this end, let $e = \{u, v\}$, where $u$ and $v$ are represented by unit squares $R_u$ and $R_v$ at coordinates $(x_u, y_u)$ and $(x_v, y_v)$, respectively, and, without loss of generality, we assume that $R_u \Vvis R_v$ (note that this implies $x_u = x_v$). We now modify the layout as follows. Every unit square $R$ on a coordinate $(x, y)$ is moved one unit to the right, if $x > x_v$ or $x = x_v$ and $y \leq y_v$ (note that this means that $R_{v}$ is also moved to the right, but $R_{u}$ is not). Obviously, this modification cannot create any new visibilities and the only visibilities that are destroyed are between unit squares $R$ and $R'$ on coordinates $(x_v, y)$ and $(x_v, y')$ with $y > y_v$ and $y' \leq y_v$, but the only unit squares that satisfy this condition are $R_u$ and $R_v$. Consequently, the modified layout represents $(V, E \setminus \{e\})$.\par
In order to show the second statement, we observe that removing $R_v$, the unit square for $v$, from the layout results in a layout for $(V \setminus \{v\}, E \cup E')$, where $E'$ is a set of at most two edges not present in $(V \setminus \{v\}, E)$. These additional edges can successively be deleted as described above, in order to obtain a layout for $(V \setminus \{v\}, E)$. 
\end{proof}

\subsection*{Proof of Lemma~\ref{simpleObsLemma}}

%\begin{lemma}
%\noindent\emph{Let $G = (V, E) \in \unitSquareGridGraphs$. Then,\\ \textbf{(1)} the maximum degree of $G$ is $4$,\\ \textbf{(2)} for every $u, v \in V$, $|N(u) \cap N(v)| \leq 2$, and,\\ \textbf{(3)} for every $\{u, v\} \in E$, $N(u) \cap N(v) = \emptyset$.}
%\end{lemma}

\begin{proof}
In a grid layout, any unit square can see at most $4$ other squares; thus, the maximum degree of $G$ is $4$. \par
Let $u, v \in V$ be represented by unit squares $R_u$ and $R_v$ on coordinates $(x_u, y_u)$ and $(x_v, y_v)$, respectively. If $x_u = x_v$ or $y_u = y_v$, then there is at most one unit square than can see both $R_u$ and $R_v$. If $x_u \neq x_v$ and $y_u \neq y_v$, then there are at most two unit squares that can see both $R_u$ and $R_v$. This implies the second statement.\par
If $R_u$ sees $R_v$, then it is impossible for any unit square to see both $R_u$ and $R_v$, which implies the third statement.
\end{proof}

\subsection*{Proof of Proposition~\ref{nonPlanarLayoutProposition}}

%\begin{proposition}\label{nonPlanarLayoutProposition}
%The graph of Fig.~\ref{Fig:Nonplanar}$(a)$ cannot be represented by a planar unit square grid layout.
%\end{proposition}

\begin{proof}
The proof shall be illustrated by Figure~\ref{Fig:NonplanarApp}. We first consider the $C_5$ on the vertices 1,2,6,7,8 which requires a visibility layout \visomorphic{} to Figure~\ref{Fig:NonplanarApp}$(b)$. Figures~\ref{Fig:NonplanarApp}$(c)$ to $(g)$ demonstrate attempts to create a layout for this graph with all possibilities to represent the $C_5$ subgraph on vertices  1,2,6,7,8 with the layout from Figure~\ref{Fig:NonplanarApp}$(b)$. Cases $(c)$ and $(d)$ show the only possibility to add the vertices 3 and 4 which leads to a layout where vertex 5 cannot be added with visibility to both 4 and 6. For cases $(e)$ and $(f)$ it is already impossible to add the vertices 3 and 4 such that they build a $C_5$ with vertices 1,2 and 8. The only possible layout is the non-planar Figure~\ref{Fig:NonplanarApp}$(g)$ which, up to \visomorphism{}, is the only unit square grid representation for the graph in Figure~\ref{Fig:NonplanarApp}$(a)$.
\end{proof}

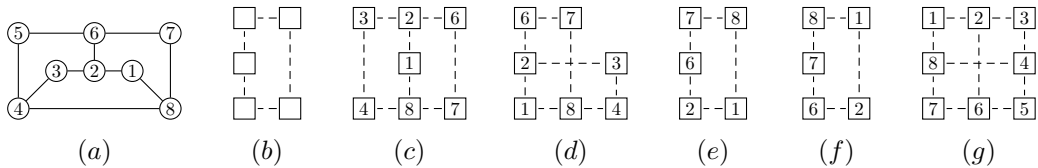
\begin{figure}[h]
\centering
\begin{tabular}{ p{2.5cm} p{1.2cm} p{1.7cm} p{1.7cm} p{1.2cm} p{1.2cm} p{1.7cm}}
\centering
\begin{tikzpicture}[nodes={draw,rectangle},scale=0.5,rotate=270]
\tikzstyle{every node}=[inner sep=0mm,minimum size=4mm,circle,draw]

    \node[scale = 0.7] (1) at (0,0) {$7$};
    \node[scale = 0.7] (2) at (2,0) {$8$};
    \node[scale = 0.7] (3) at (1,-1) {$1$};
    \node[scale = 0.7] (4) at (0,-2) {$6$};
    \node[scale = 0.7] (5) at (1,-2) {$2$};
    \node[scale = 0.7] (6) at (1,-3) {$3$};
    \node[scale = 0.7] (7) at (0,-4) {$5$};
    \node[scale = 0.7] (8) at (2,-4) {$4$};
    
    \draw (1) edge (2);
    \draw (1) edge (4);
    \draw (2) edge (3);
    \draw (2) edge (8);
    \draw (3) edge (5);
    \draw (4) edge (5);
    \draw (4) edge (7);
    \draw (5) edge (6);
    \draw (6) edge (8);
    \draw (7) edge (8);
		
\end{tikzpicture}&
\centering
\begin{tikzpicture}[nodes={draw,rectangle},scale=0.6,rotate=90]
\tikzstyle{every node}=[inner sep=0mm,outer sep=0.5mm,minimum size=4mm,rectangle,draw]
    \node[scale = 0.7] (2) at (0,0) {};
    \node[scale = 0.7] (3) at (1,0) {};
    \node[scale = 0.7] (5) at (2,0) {};
    \node[scale = 0.7] (1) at (0,-1) {};
    \node[scale = 0.7] (4) at (2,-1) {};

    \draw[densely dashed] (1) edge (2);
    \draw[densely dashed] (2) edge (3);
    \draw[densely dashed] (3) edge (5);
    \draw[densely dashed] (4) edge (1);
    \draw[densely dashed] (4) edge (5);

\end{tikzpicture}
&

\centering
\begin{tikzpicture}[nodes={draw,rectangle},scale=0.6,rotate=90]
\tikzstyle{every node}=[inner sep=0mm,outer sep=0.5mm,minimum size=4mm,rectangle,draw]
    \node[scale = 0.7] (1) at (0,-1) {$7$};
    \node[scale = 0.7] (2) at (0,0) {$8$};
    \node[scale = 0.7] (3) at (1,0) {$1$};
    \node[scale = 0.7] (4) at (2,-1) {$6$};
    \node[scale = 0.7] (5) at (2,0) {$2$};
    \node[scale = 0.7] (7) at (0,1)  {$4$};
    \node[scale = 0.7] (8) at (2,1) {$3$};  
    
    \draw[densely dashed] (1) edge (2);
    \draw[densely dashed] (2) edge (3);
    \draw[densely dashed] (3) edge (5);
    \draw[densely dashed] (4) edge (1);
    \draw[densely dashed] (4) edge (5);
    \draw[densely dashed] (5) edge (8);
    \draw[densely dashed] (8) edge (7);
    \draw[densely dashed] (2) edge (7);

\end{tikzpicture}
&
\centering
\begin{tikzpicture}[nodes={draw,rectangle},scale=0.6,rotate=90]
\tikzstyle{every node}=[inner sep=0mm,outer sep=0.5mm,minimum size=4mm,rectangle,draw]
    \node[scale = 0.7] (3) at (0,0) {$1$};
    \node[scale = 0.7] (5) at (1,0) {$2$};
    \node[scale = 0.7] (4) at (2,0) {$6$};
    \node[scale = 0.7] (2) at (0,-1) {$8$};
    \node[scale = 0.7] (1) at (2,-1) {$7$};
    \node[scale = 0.7] (7) at  (0,-2)  {$4$};
    \node[scale = 0.7] (8) at (1,-2) {$3$};    
    
    \draw[densely dashed] (1) edge (2);
    \draw[densely dashed] (2) edge (3);
    \draw[densely dashed] (3) edge (5);
    \draw[densely dashed] (4) edge (1);
    \draw[densely dashed] (4) edge (5);
    \draw[densely dashed] (5) edge (8);
    \draw[densely dashed] (8) edge (7);
    \draw[densely dashed] (2) edge (7);

\end{tikzpicture}

&
\centering
\begin{tikzpicture}[nodes={draw,rectangle},scale=0.6,rotate=90]
\tikzstyle{every node}=[inner sep=0mm,outer sep=0.5mm,minimum size=4mm,rectangle,draw]
    \node[scale = 0.7] (5) at (0,0) {$2$};
    \node[scale = 0.7] (4) at (1,0) {$6$};
    \node[scale = 0.7] (1) at (2,0) {$7$};
    \node[scale = 0.7] (3) at (0,-1) {$1$};
    \node[scale = 0.7] (2) at (2,-1) {$8$};
    
    \draw[densely dashed] (1) edge (2);
    \draw[densely dashed] (2) edge (3);
    \draw[densely dashed] (3) edge (5);
    \draw[densely dashed] (4) edge (1);
    \draw[densely dashed] (4) edge (5);

\end{tikzpicture}
&
\centering
\begin{tikzpicture}[nodes={draw,rectangle},scale=0.6,rotate=90]
\tikzstyle{every node}=[inner sep=0mm,outer sep=0.5mm,minimum size=4mm,rectangle,draw]
    \node[scale = 0.7] (4) at (0,0) {$6$};
    \node[scale = 0.7] (1) at (1,0) {$7$};
    \node[scale = 0.7] (2) at (2,0) {$8$};
    \node[scale = 0.7] (5) at (0,-1) {$2$};
    \node[scale = 0.7] (3) at (2,-1) {$1$};

    \draw[densely dashed] (1) edge (2);
    \draw[densely dashed] (2) edge (3);
    \draw[densely dashed] (3) edge (5);
    \draw[densely dashed] (4) edge (1);
    \draw[densely dashed] (4) edge (5);

\end{tikzpicture}
&
\centering
\begin{tikzpicture}[nodes={draw,rectangle},scale=0.6,rotate=90]
\tikzstyle{every node}=[inner sep=0mm,outer sep=0.5mm,minimum size=4mm,rectangle,draw]
    \node[scale = 0.7] (1) at (0,0) {$7$};
    \node[scale = 0.7] (2) at (1,0) {$8$};
    \node[scale = 0.7] (3) at (2,0) {$1$};
    \node[scale = 0.7] (4) at (0,-1) {$6$};
    \node[scale = 0.7] (5) at (2,-1) {$2$};
    \node[scale = 0.7] (6) at (0,-2) {$5$};
    \node[scale = 0.7] (7) at (1,-2) {$4$}; 
    \node[scale = 0.7] (8) at (2,-2) {$3$};

    \draw[densely dashed] (1) edge (2);
    \draw[densely dashed] (2) edge (3);
    \draw[densely dashed] (3) edge (5);
    \draw[densely dashed] (5) edge (8);
    \draw[densely dashed] (8) edge (7);
    \draw[densely dashed] (7) edge (6);
    \draw[densely dashed] (6) edge (4);
    \draw[densely dashed] (4) edge (1);
    \draw[densely dashed] (2) edge (7);
    \draw[densely dashed] (4) edge (5);
		
\end{tikzpicture}
\cr\\[-0.75cm]
\centering$(a)$&\centering $(b)$&\centering $(c)$&\centering $(d)$&\centering $(e)$&\centering $(f)$&\centering $(g)$
\end{tabular}
\caption{Illustrations for the proof of Proposition~\ref{nonPlanarLayoutProposition}.}
\label{Fig:NonplanarApp} 
\end{figure}

\subsection*{A Unit Square Grid Layout for a Subdivision of $K_5$}

%\begin{figure}
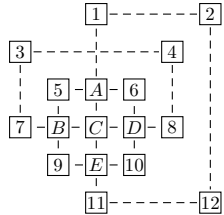
\begin{center}
\begin{tikzpicture}[nodes={draw,rectangle},scale=0.5]
\tikzstyle{every node}=[inner sep=0mm,outer sep=0.5mm,minimum size=4mm,rectangle,draw]
    \node[scale = 0.7] (C) at (0,0) {$C$};
    \node[scale = 0.7] (B) at (-1,0) {$B$};
    \node[scale = 0.7] (A) at (0,1) {$A$};
    \node[scale = 0.7] (D) at (1,0) {$D$};
    \node[scale = 0.7] (E) at (0,-1) {$E$};
    \node[scale = 0.7] (1) at (0,3) {$1$};
    \node[scale = 0.7] (2) at (3,3) {$2$};
    \node[scale = 0.7] (3) at (-2,2) {$3$};
    \node[scale = 0.7] (4) at (2,2) {$4$};
    \node[scale = 0.7] (5) at (-1,1) {$5$};
    \node[scale = 0.7] (6) at (1,1) {$6$};
    \node[scale = 0.7] (7) at (-2,0) {$7$};
    \node[scale = 0.7] (8) at (2,0) {$8$};
    \node[scale = 0.7] (9) at (-1,-1) {$9$};
    \node[scale = 0.7] (10) at (1,-1) {$10$};
    \node[scale = 0.7] (11) at (0,-2) {$11$};
    \node[scale = 0.7] (12) at (3,-2) {$12$};

    \draw[densely dashed] (A) edge (1);
    \draw[densely dashed] (A) edge (5);
    \draw[densely dashed] (A) edge (6);
    \draw[densely dashed] (A) edge (C);
    \draw[densely dashed] (B) edge (5);
    \draw[densely dashed] (B) edge (7);
    \draw[densely dashed] (B) edge (9);
	\draw[densely dashed] (B) edge (C);
    \draw[densely dashed] (C) edge (D);
    \draw[densely dashed] (C) edge (E);
    \draw[densely dashed] (D) edge (6);
    \draw[densely dashed] (D) edge (8); 
    \draw[densely dashed] (D) edge (10);
    \draw[densely dashed] (E) edge (9);
    \draw[densely dashed] (E) edge (10);
    \draw[densely dashed] (E) edge (11);
    \draw[densely dashed] (1) edge (2);
    \draw[densely dashed] (3) edge (4);
    \draw[densely dashed] (3) edge (7);
    \draw[densely dashed] (4) edge (8);
    \draw[densely dashed] (11) edge (12);
    \draw[densely dashed] (2) edge (12);
        
\end{tikzpicture}

\end{center}

\subsection*{Proof of Theorem~\ref{noCharacterisationGridTheorem}}

%\begin{theorem}
%\noindent\emph{$\unitSquareGridGraphs$ does not admit a characterization by a finite number of forbidden induced subgraphs.}
%\end{theorem}

\begin{proof}
Consider the family of graphs $\{G_n \mid n \geq 3\}$, where $G_n=(V_n,E_n)$ with 
\begin{equation*}
V_n=\{u_1,\dots,u_n\}\cup\{v_2,\dots,v_n\}\cup\{w\}\,,
\end{equation*}
and 
\begin{align*}
E_n =\: &\{\{u_i, u_{i+1}\}, \{v_i, v_{i+1}\}, \{u_i, v_i\} \mid 2 \leq i \leq n-1\} \cup\\ &\{\{u_1, u_2\}, \{u_n, v_n\}, \{u_1, w\}, \{v_n, w\}\}\,.
\end{align*}
We note that, for every $n \geq 3$, a grid layout for $G_n - w$ (the graph created from $G_n$ by deleting the vertex $w$ and its incident edges) can be constructed by placing the unit squares for the vertices $u_i$, $1 \leq i \leq n$, on a horizontal line in this order and the unit squares for the vertices $v_i$, $2 \leq i \leq n$, on a parallel horizontal line in this order, such that, for every $i$, $2 \leq i \leq n$, the horizontal components of the coordinates of the unit squares for $u_i$ and $v_i$ correspond. Furthermore, every grid layout for $G_n - w$ has either this structure or places the unit squares analogously on two parallel vertical lines. This consideration not only shows that $G_n - w \in \unitSquareGridGraphs$, but also demonstrates that $G_n \notin \unitSquareGridGraphs$, since it is impossible for a unit square to see both the unit squares for $u_1$ and $v_n$. In the following, we observe that, for every $x \in V_n \setminus \{w\}$, $G_n - x \in \unitSquareGridGraphs$. For $x \in \{u_1, v_2, v_n, u_n\}$, this property can be easily verified. For $x = u_i$, $2 \leq i \leq n-1$, we can construct a grid layout by rotating the part representing vertices $\{u_1, \ldots, u_{i-1}, v_2, \ldots, v_{i-1}\}$ by ninety degrees, and an analogous construction applies in the case $x = v_i$, $3 \leq i \leq n-1$.\par
By Lemma~\ref{removeVerticesEdgesLemma}, it follows that, for every $n \geq 3$, every proper subgraph of $G_n$ is in $\unitSquareGridGraphs$, while $G_n \notin \unitSquareGridGraphs$. Consequently, it is not possible to characterise $\unitSquareGridGraphs$ by a finite number of forbidden induced subgraphs.

\end{proof}

\subsection*{Proof of Theorem~\ref{thm-sgvg-tree}}

%\begin{theorem}
%\emph{A tree $T$ is in $\unitSquareGridGraphs$ if and only if the maximum degree of $T$ is at most four.} 
%\end{theorem}

\begin{proof}
The \emph{only if} direction follows from Lemma~\ref{simpleObsLemma}. 
To prove the \emph{if} direction, let $T \in \unitSquareGridGraphs$ be a tree with a vertex $v$ of degree at most $3$. In order to append a new vertex to $v$, we can place a new unit square $R$ within visibility of $R_v$, the unit square for $v$, without destroying any visibilities. Possible new visibilities between $R$ and other unit squares can be removed due to Lemma~\ref{removeVerticesEdgesLemma}. The statement of the lemma follows by induction.
\end{proof}

% \subsection{Proof of Theorem~\ref{USGVAreaMinimisationTheorem}}

% \begin{proof}
% Let $G = (V, E)$ be a graph with an $\LRDU$-restriction $\sigma$ (with respect to some $\vec{E}$). We first check whether $\sigma$ is valid for $G$, which can be easily done in polynomial-time\todo{MS: More precise complexity?}. \par 
% For $u, v \in V$ with $(u, v) \in \vec{E}$, we set $u \sim'_x v$ if $\sigma((u, v)) \in \{\Dlabel, \Ulabel\}$ and $u \sim'_y v$ if $\sigma((u, v)) \in \{\Llabel, \Rlabel\}$. Let $\sim_x$ and $\sim_y$ be the reflexive-transitive closures (i.\,e., equivalence relations) of $\sim'_x$ and $\sim'_y$, respectively. Moreover, let $\ell_x$ and $\ell_y$ be the number of equivalence classes of $\sim_x$ and $\sim_y$, respectively. \par
% Obviously, in every unit square grid visibility layout for $G$ that satisfies the $\LRDU$-constraint, two vertices $u$ and $v$ are represented by unit squares with the same $x$-coordinate if and only if $u \sim_x y$; thus, any possible layout has a width of at least $\ell_x$ and, analogously, it follows that any possible layout has a height of at least $\ell_y$.\par
% The proof is concluded by observing that the polynomial-time algorithm from~\cite{ManPPT1011} for constructing (if possible) a rectilinear drawing of a given $\LRDU$-restricted graph actually provides a rectilinear drawing of size $\ell_x \times \ell_y$, in which no two non-adjacent vertices have the same $x$- or $y$-component; thus, it can be regarded as a unit square grid visibility layout. 
% \end{proof}

\subsection*{Proof of Lemma~\ref{areaMinimisationLemma}}

%\todo[inline]{MS: This needs to be checked carefully!}

First, we recall the reduction from page~\pageref{areaMinimisationReductionPageRef}. Let $B \in \mathbb{N}$ and $A = \{a_1, a_2, \ldots, a_{3m}\} \subseteq \mathbb{N}$ be a $\threepartition$ instance. Moreover, by simple scaling, we can assume that $a_i > 2$, $1 \leq i \leq 3m$. We construct a \emph{frame graph} $G_f = (V_f, E_f)$ with:
\begin{align*}
V_f = &\:\{u_{i, j}, v_{i, j}, w_{i, 1}, w_{i, 2} \mid 1 \leq i \leq m, 0 \leq j \leq B\} \cup \{u_{m + 1, 0}, v_{m + 1, 0}, w_{m+1, 1}, w_{m + 1, 2}\},\\
E_f = &\:\left\{\{u_{i,j}, u_{i,j+1}\}, \{v_{i,j}, v_{i,j+1}\} \mid 1 \leq i \leq m, 0 \leq j \leq B-1\right\} \cup{} \\
&\left\{\{u_{i, B}, u_{i+1, 0}\}, \{v_{i, B}, v_{i+1, 0}\} \mid 1 \leq i \leq m\right\} \cup{}\\ 
&\left\{\{u_{i,j}, v_{i,j}\}\mid 1 \leq i \leq m, 1 \leq j \leq B\right\} \cup{}\\
&\left\{\{u_{i,0}, v_{i,0}\}, \{v_{i,0}, w_{i,1}\}, \{w_{i,1}, w_{i,2}\} \mid 1 \leq i \leq m+1\right\}\,.
\end{align*}
%\label{areaMinimisationReductionPageRef}
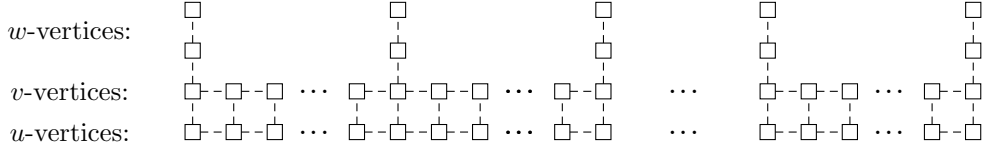
\begin{figure}

\centering
\begin{tikzpicture}[nodes={draw,rectangle},scale=0.6]
\tikzstyle{every node}=[inner sep=0mm,outer sep=0.0mm,minimum size=2mm,rectangle,draw]
     
     \coordinate (xshift) at (0.9,0);
     \coordinate (yshift) at (0,0.9);

	%First segment
     
    \node (u10) at (0,0){};
    \node (u11) at ($(u10) + (xshift)$){};
    \node (u12) at ($(u11) + (xshift)$){};
    \node (u1B) at ($(u12) + 2*(xshift)$){};
       
    \node (v10) at ($(0,0) + (yshift)$){};
    \node (v11) at ($(v10) + (xshift)$){};
    \node (v12) at ($(v11) + (xshift)$){};
    \node (v1B) at ($(v12) + 2*(xshift)$){};
    
    \node (w11) at ($(v10) + (yshift)$){};
    \node (w12) at ($(w11) + (yshift)$){};
    
    \node[draw=none, scale=0.8] (dots10) at ($(u12) + 1*(xshift)$){\bf \dots};
    \node[draw=none, scale=0.8] (dots11) at ($(v12) + 1*(xshift)$){\bf \dots};
    
    \draw[densely dashed] (u10) edge (v10);
    \draw[densely dashed] (u11) edge (v11);
    \draw[densely dashed] (u12) edge (v12);
    \draw[densely dashed] (u1B) edge (v1B);
    
    \draw[densely dashed] (u10) edge (u11);
    \draw[densely dashed] (u11) edge (u12);
    
    \draw[densely dashed] (v10) edge (v11);
    \draw[densely dashed] (v11) edge (v12);
    
    \draw[densely dashed] (v10) edge (w11);
    \draw[densely dashed] (w11) edge (w12);

    %Second segment
     
    \node (u20) at ($(u1B) + (xshift)$){};
    \node (u21) at ($(u20) + (xshift)$){};
    \node (u22) at ($(u21) + (xshift)$){};
    \node (u2B) at ($(u22) + 2*(xshift)$){};
       
    \node (v20) at ($(u20) + (yshift)$){};
    \node (v21) at ($(v20) + (xshift)$){};
    \node (v22) at ($(v21) + (xshift)$){};
    \node (v2B) at ($(v22) + 2*(xshift)$){};
    
    \node (w21) at ($(v20) + (yshift)$){};
    \node (w22) at ($(w21) + (yshift)$){};
    
    \node (u30) at ($(u2B) + (xshift)$){};
    \node (v30) at ($(u30) + (yshift)$){};
    \node (w31) at ($(v30) + (yshift)$){};
    \node (w32) at ($(w31) + (yshift)$){};
    
    \node[draw=none, scale=0.8] (dots20) at ($(u22) + 1*(xshift)$){\bf \dots};
    \node[draw=none, scale=0.8] (dots21) at ($(v22) + 1*(xshift)$){\bf \dots};
    
    \draw[densely dashed] (u1B) edge (u20);
    \draw[densely dashed] (v1B) edge (v20);
    
    \draw[densely dashed] (u20) edge (v20);
    \draw[densely dashed] (u21) edge (v21);
    \draw[densely dashed] (u22) edge (v22);
    \draw[densely dashed] (u2B) edge (v2B);
    
    \draw[densely dashed] (u20) edge (u21);
    \draw[densely dashed] (u21) edge (u22);
    
    \draw[densely dashed] (v20) edge (v21);
    \draw[densely dashed] (v21) edge (v22);
    
    \draw[densely dashed] (v20) edge (w21);
    \draw[densely dashed] (w21) edge (w22);
    
    \draw[densely dashed] (u2B) edge (u30);
    \draw[densely dashed] (v2B) edge (v30);
    
    \draw[densely dashed] (u30) edge (v30);
    \draw[densely dashed] (v30) edge (w31);
    \draw[densely dashed] (w31) edge (w32);
    
    %Dots
    
    \node[draw=none, scale=0.8] (maindots0) at ($(u30) + 2*(xshift)$){\bf \dots};
    \node[draw=none, scale=0.8] (maindots1) at ($(v30) + 2*(xshift)$){\bf \dots};
    
    %Last segment
    
    \node (um0) at ($(maindots0) + 2*(xshift)$){};
    \node (um1) at ($(um0) + (xshift)$){};
    \node (um2) at ($(um1) + (xshift)$){};
    \node (umB) at ($(um2) + 2*(xshift)$){};
       
    \node (vm0) at ($(um0) + (yshift)$){};
    \node (vm1) at ($(vm0) + (xshift)$){};
    \node (vm2) at ($(vm1) + (xshift)$){};
    \node (vmB) at ($(vm2) + 2*(xshift)$){};
    
    \node (wm1) at ($(vm0) + (yshift)$){};
    \node (wm2) at ($(wm1) + (yshift)$){};
    
    \node[draw=none, scale=0.8] (dots20) at ($(u22) + 1*(xshift)$){\bf \dots};
    \node[draw=none, scale=0.8] (dots21) at ($(v22) + 1*(xshift)$){\bf \dots};
    
%    \draw (u1B) edge (um0);
%    \draw (v1B) edge (vm0);
    
    \draw[densely dashed] (um0) edge (vm0);
    \draw[densely dashed] (um1) edge (vm1);
    \draw[densely dashed] (um2) edge (vm2);
    \draw[densely dashed] (umB) edge (vmB);
    
    \draw[densely dashed] (um0) edge (um1);
    \draw[densely dashed] (um1) edge (um2);
    
    \draw[densely dashed] (vm0) edge (vm1);
    \draw[densely dashed] (vm1) edge (vm2);
    
    \draw[densely dashed] (vm0) edge (wm1);
    \draw[densely dashed] (wm1) edge (wm2);
    
    %Last boundary
    
    \node (umplus1) at ($(umB) + (xshift)$){};
    \node (vmplus1) at ($(vmB) + (xshift)$){};
    
    \node (wmplus11) at ($(vmplus1) + (yshift)$){};
    \node (wmplus12) at ($(wmplus11) + (yshift)$){};
    
    \node[draw=none, scale=0.8] (dotsm0) at ($(um2) + 1*(xshift)$){\bf \dots};
    \node[draw=none, scale=0.8] (dotsm1) at ($(vm2) + 1*(xshift)$){\bf \dots};
    
    \draw[densely dashed] (umB) edge (umplus1);
    \draw[densely dashed] (vmB) edge (vmplus1);
    
    \draw[densely dashed] (umplus1) edge (vmplus1);
    \draw[densely dashed] (vmplus1) edge (wmplus11);
    \draw[densely dashed] (wmplus11) edge (wmplus12);
    
    %Description
    
    \node[draw=none, scale=1] (uvertices) at ($(u10) - 3*(xshift)$){$u$-vertices:};
    \node[draw=none, scale=1] (vvertices) at ($(v10) - 3*(xshift)$){$v$-vertices:};
    \node[draw=none, scale=1] (wvertices) at ($(w11) - 3*(xshift) + 0.5*(yshift)$){$w$-vertices:};

%    \node[scale = 0.7] (10) at ($(9) + (xshift)$){$\ldots$};

\end{tikzpicture}
\caption{Unit square grid layout for the graph $G_f$.}
\label{Fig:frameGraphAppendix}
\end{figure}

Next, we define a graph $G_A = (V_A, E_A)$ with 
\begin{align*}
V_A = &\:\bigcup^{3m}_{i = 1} \{b_{i, j}, c_{i, j} \mid 1 \leq j \leq a_i\}\,,\\
E_A = &\:\{\{b_{i, j}, b_{i, j+1}\}, \{c_{i, j}, c_{i, j+1}\} \mid 1 \leq i \leq 3m, 1 \leq j \leq a_i-1\} \cup{}\\
&\{\{b_{i, j}, c_{i, j}\} \mid 1 \leq i \leq 3m, 1 \leq j \leq a_i\}\,. 
\end{align*}
Finally, we let $G = (V, E)$ with $V = V_f \cup V_A$ and $E = E_f \cup E_A$.\par
We now give the proof of Lemma~\ref{areaMinimisationLemma}.

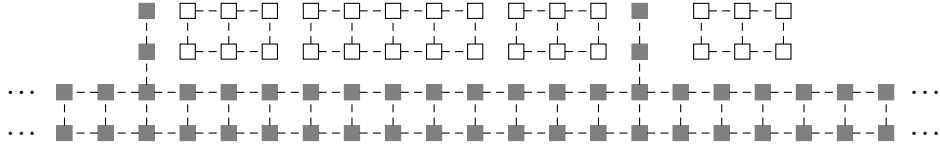
\begin{figure}

\centering
\begin{tikzpicture}[nodes={draw,rectangle},scale=0.6]
\tikzstyle{every node}=[inner sep=0mm,outer sep=0.0mm,minimum size=2mm,rectangle,draw]
     
     \coordinate (xshift) at (0.9,0);
     \coordinate (yshift) at (0,0.9);
     \coordinate (start) at (0,0);
     
	%FRAME
     
    \node[fill, color=gray] (u20) at (start){};
    \node[fill, color=gray] (u21) at ($(u20) + (xshift)$){};
    \node[fill, color=gray] (u22) at ($(u21) + (xshift)$){};
    \node[fill, color=gray] (u23) at ($(u22) + (xshift)$){};
    \node[fill, color=gray] (u24) at ($(u23) + (xshift)$){};
    \node[fill, color=gray] (u25) at ($(u24) + (xshift)$){};
    \node[fill, color=gray] (u26) at ($(u25) + (xshift)$){};
    \node[fill, color=gray] (u27) at ($(u26) + (xshift)$){};
    \node[fill, color=gray] (u28) at ($(u27) + (xshift)$){};
    \node[fill, color=gray] (u29) at ($(u28) + (xshift)$){};
    \node[fill, color=gray] (u30) at ($(u29) + (xshift)$){};
%    \node (u31) at ($(u30) + (xshift)$){};
%    \node (u32) at ($(u31) + (xshift)$){};
    \node[fill, color=gray] (u2B) at ($(u30) + (xshift)$){};

    \node[fill, color=gray] (un) at ($(u2B) + (xshift)$){};
      
    \node[fill, color=gray] (v20) at ($(start) + (yshift)$){};
    \node[fill, color=gray] (v21) at ($(v20) + (xshift)$){};
    \node[fill, color=gray] (v22) at ($(v21) + (xshift)$){};
    \node[fill, color=gray] (v23) at ($(v22) + (xshift)$){};
    \node[fill, color=gray] (v24) at ($(v23) + (xshift)$){};
    \node[fill, color=gray] (v25) at ($(v24) + (xshift)$){};
    \node[fill, color=gray] (v26) at ($(v25) + (xshift)$){};
    \node[fill, color=gray] (v27) at ($(v26) + (xshift)$){};
    \node[fill, color=gray] (v28) at ($(v27) + (xshift)$){};
    \node[fill, color=gray] (v29) at ($(v28) + (xshift)$){};
    \node[fill, color=gray] (v30) at ($(v29) + (xshift)$){};
%    \node (v31) at ($(v30) + (xshift)$){};
%    \node (v32) at ($(v31) + (xshift)$){};
    \node[fill, color=gray] (v2B) at ($(v30) + (xshift)$){};

    \node[fill, color=gray] (vn) at ($(v2B) + (xshift)$){};
    
    \node[fill, color=gray] (w21) at ($(v20) + (yshift)$){};
    \node[fill, color=gray] (w22) at ($(w21) + (yshift)$){};
    
    \node[fill, color=gray] (w31) at ($(vn) + (yshift)$){};
    \node[fill, color=gray] (w32) at ($(w31) + (yshift)$){};

    \draw[densely dashed] (u20) edge (v20);
    \draw[densely dashed] (u21) edge (v21);
    \draw[densely dashed] (u22) edge (v22);
    \draw[densely dashed] (u23) edge (v23);
    \draw[densely dashed] (u24) edge (v24);
    \draw[densely dashed] (u25) edge (v25);
    \draw[densely dashed] (u26) edge (v26);
    \draw[densely dashed] (u27) edge (v27);
    \draw[densely dashed] (u28) edge (v28);
    \draw[densely dashed] (u29) edge (v29);
    \draw[densely dashed] (u30) edge (v30);
%    \draw[densely dashed] (u31) edge (v31);
%    \draw[densely dashed] (u32) edge (v32);
    \draw[densely dashed] (u2B) edge (v2B);
    \draw[densely dashed] (un) edge (vn);

    \draw[densely dashed] (u20) edge (u21);
    \draw[densely dashed] (u21) edge (u22);
    \draw[densely dashed] (u22) edge (u23);
    \draw[densely dashed] (u23) edge (u24);
    \draw[densely dashed] (u24) edge (u25);
    \draw[densely dashed] (u25) edge (u26);
    \draw[densely dashed] (u26) edge (u27);
    \draw[densely dashed] (u27) edge (u28);
    \draw[densely dashed] (u28) edge (u29);
    \draw[densely dashed] (u29) edge (u30);
%    \draw[densely dashed] (u30) edge (u31);
%    \draw[densely dashed] (u31) edge (u32);
%    \draw[densely dashed] (u32) edge (u2B);
    \draw[densely dashed] (u30) edge (u2B);
    \draw[densely dashed] (u2B) edge (un);

    \draw[densely dashed] (v20) edge (v21);
    \draw[densely dashed] (v21) edge (v22);
    \draw[densely dashed] (v22) edge (v23);
    \draw[densely dashed] (v23) edge (v24);
    \draw[densely dashed] (v24) edge (v25);
    \draw[densely dashed] (v25) edge (v26);
    \draw[densely dashed] (v26) edge (v27);
    \draw[densely dashed] (v27) edge (v28);
    \draw[densely dashed] (v28) edge (v29);
    \draw[densely dashed] (v29) edge (v30);
%    \draw[densely dashed] (v30) edge (v31);
%    \draw[densely dashed] (v31) edge (v32);
%    \draw[densely dashed] (v32) edge (v2B);
    \draw[densely dashed] (v30) edge (v2B);
    \draw[densely dashed] (v2B) edge (vn);

    \draw[densely dashed] (v20) edge (w21);
    \draw[densely dashed] (w21) edge (w22);

    \draw[densely dashed] (vn) edge (w31);
    \draw[densely dashed] (w31) edge (w32);

%FRAME EXTENSION

\node[fill, color=gray] (uu0) at ($(un) + (xshift)$){};
\node[fill, color=gray] (uu1) at ($(uu0) + (xshift)$){};
\node[fill, color=gray] (uu2) at ($(uu1) + (xshift)$){};
\node[fill, color=gray] (uu3) at ($(uu2) + (xshift)$){};
\node[fill, color=gray] (uu4) at ($(uu3) + (xshift)$){};
\node[fill, color=gray] (uu5) at ($(uu4) + (xshift)$){};

\node[fill, color=gray] (vv0) at ($(uu0) + (yshift)$){};
\node[fill, color=gray] (vv1) at ($(vv0) + (xshift)$){};
\node[fill, color=gray] (vv2) at ($(vv1) + (xshift)$){};
\node[fill, color=gray] (vv3) at ($(vv2) + (xshift)$){};
\node[fill, color=gray] (vv4) at ($(vv3) + (xshift)$){};
\node[fill, color=gray] (vv5) at ($(vv4) + (xshift)$){};

\draw[densely dashed] (un) edge (uu0);
\draw[densely dashed] (uu0) edge (uu1);
\draw[densely dashed] (uu1) edge (uu2);
\draw[densely dashed] (uu2) edge (uu3);
\draw[densely dashed] (uu3) edge (uu4);
\draw[densely dashed] (uu4) edge (uu5);

\draw[densely dashed] (vn) edge (vv0);
\draw[densely dashed] (vv0) edge (vv1);
\draw[densely dashed] (vv1) edge (vv2);
\draw[densely dashed] (vv2) edge (vv3);
\draw[densely dashed] (vv3) edge (vv4);
\draw[densely dashed] (vv4) edge (vv5);

\draw[densely dashed] (uu0) edge (vv0);
\draw[densely dashed] (uu1) edge (vv1);
\draw[densely dashed] (uu2) edge (vv2);
\draw[densely dashed] (uu3) edge (vv3);
\draw[densely dashed] (uu4) edge (vv4);
\draw[densely dashed] (uu5) edge (vv5);

\node[fill, color=gray] (ul0) at ($(u20) - (xshift)$){};
\node[fill, color=gray] (ul1) at ($(ul0) - (xshift)$){};

\node[fill, color=gray] (vl0) at ($(ul0) + (yshift)$){};
\node[fill, color=gray] (vl1) at ($(vl0) - (xshift)$){};

\draw[densely dashed] (u20) edge (ul0);
\draw[densely dashed] (ul0) edge (ul1);

\draw[densely dashed] (v20) edge (vl0);
\draw[densely dashed] (vl0) edge (vl1);

\draw[densely dashed] (ul0) edge (vl0);
\draw[densely dashed] (ul1) edge (vl1);

	%ELEMENTS

	\node (b11) at ($(v21) + (yshift)$){};
	\node (b12) at ($(b11) + (xshift)$){};
	\node (b13) at ($(b12) + (xshift)$){};
	\node (c11) at ($(b11) + (yshift)$){};
	\node (c12) at ($(b12) + (yshift)$){};
	\node (c13) at ($(b13) + (yshift)$){};

	\node (b21) at ($(b13) + (xshift)$){};
	\node (b22) at ($(b21) + (xshift)$){};
	\node (b23) at ($(b22) + (xshift)$){};
	\node (b24) at ($(b23) + (xshift)$){};
	\node (b25) at ($(b24) + (xshift)$){};
	\node (c21) at ($(b21) + (yshift)$){};
	\node (c22) at ($(b22) + (yshift)$){};
	\node (c23) at ($(b23) + (yshift)$){};
	\node (c24) at ($(b24) + (yshift)$){};
	\node (c25) at ($(b25) + (yshift)$){};

	\node (b31) at ($(b25) + (xshift)$){};
	\node (b32) at ($(b31) + (xshift)$){};
	\node (b33) at ($(b32) + (xshift)$){};
	\node (c31) at ($(b31) + (yshift)$){};
	\node (c32) at ($(b32) + (yshift)$){};
	\node (c33) at ($(b33) + (yshift)$){};

	\draw[densely dashed] (b11) edge (b12);
	\draw[densely dashed] (b12) edge (b13);
	\draw[densely dashed] (c11) edge (c12);
	\draw[densely dashed] (c12) edge (c13);
	\draw[densely dashed] (b11) edge (c11);
	\draw[densely dashed] (b12) edge (c12);
	\draw[densely dashed] (b13) edge (c13);

	\draw[densely dashed] (b21) edge (b22);
	\draw[densely dashed] (b22) edge (b23);
	\draw[densely dashed] (b23) edge (b24);
	\draw[densely dashed] (b24) edge (b25);
	\draw[densely dashed] (c21) edge (c22);
	\draw[densely dashed] (c22) edge (c23);
	\draw[densely dashed] (c23) edge (c24);	
	\draw[densely dashed] (c24) edge (c25);
	\draw[densely dashed] (b21) edge (c21);
	\draw[densely dashed] (b22) edge (c22);
	\draw[densely dashed] (b23) edge (c23);
	\draw[densely dashed] (b24) edge (c24);
	\draw[densely dashed] (b25) edge (c25);

	\draw[densely dashed] (b31) edge (b32);
	\draw[densely dashed] (b32) edge (b33);
	\draw[densely dashed] (c31) edge (c32);
	\draw[densely dashed] (c32) edge (c33);
	\draw[densely dashed] (b31) edge (c31);
	\draw[densely dashed] (b32) edge (c32);
	\draw[densely dashed] (b33) edge (c33);

%MISPLACED ELEMENTS

	\node (bm11) at ($(vv0) + (yshift) + 0.5*(xshift)$){};
	\node (bm12) at ($(bm11) + (xshift)$){};
	\node (bm13) at ($(bm12) + (xshift)$){};
	\node (cm11) at ($(bm11) + (yshift)$){};
	\node (cm12) at ($(bm12) + (yshift)$){};
	\node (cm13) at ($(bm13) + (yshift)$){};

	\draw[densely dashed] (bm11) edge (bm12);
	\draw[densely dashed] (bm12) edge (bm13);
	\draw[densely dashed] (cm11) edge (cm12);
	\draw[densely dashed] (cm12) edge (cm13);
	\draw[densely dashed] (bm11) edge (cm11);
	\draw[densely dashed] (bm12) edge (cm12);
	\draw[densely dashed] (bm13) edge (cm13);

%DOTS

    \node[draw=none, scale=0.8] (dots20) at ($(uu5) + 1*(xshift)$){\bf \dots};
    \node[draw=none, scale=0.8] (dots21) at ($(vv5) + 1*(xshift)$){\bf \dots};

    \node[draw=none, scale=0.8] (dots22) at ($(ul1) - 1*(xshift)$){\bf \dots};
    \node[draw=none, scale=0.8] (dots23) at ($(vl1) - 1*(xshift)$){\bf \dots};

\end{tikzpicture}
\caption{Illustration of how the unit squares for the subgraphs corresponding to $A_i = \{3, 5, 3\}$ can be placed on coordinates $(p_i + 1,\ell), (p_i + 2,\ell), \ldots, (p_i + 2B + 1,\ell)$, $\ell \in \{3, 4, 5, 6\}$, with $p_i = (i-1) \cdot 2(B+1)$ (in the case $B = 11$), and of how the unit squares for a subgraph corresponding to some $a_j = 3$ can be placed ``shifted'' with respect to the frame graph.}
\label{Fig:frameGraphAppendixTwo}
\end{figure}

\begin{proof}
For the sake of convenience, in the following, we denote the vertices $u_{i, j}$, $1 \leq i \leq m, 0 \leq j \leq B$, and $u_{m + 1, 0}$ by \emph{$u$-vertices}, the vertices $v_{i, j}$, $1 \leq i \leq m, 0 \leq j \leq B$, and $v_{m + 1, 0}$ by \emph{$v$-vertices} and the vertices $w_{i, 1}, w_{i, 2}$, $1 \leq i \leq m + 1$, by \emph{$w$-vertices}.\par
We now assume that $A_1, \ldots, A_m$ is a partition of $A$ with $\sum_{a \in A_i} = B$, $1 \leq i \leq m$. We can construct a $(7 \times (2(mB+m+1)-1))$ unit square grid layout for $G_f$, by representing all $u$- and $v$-vertices as a horizontal ``ladder'', as illustrated in Figure~\ref{Fig:frameGraphAppendix}, where vertex $u_{1,0}$ is positioned at coordinate $(0, 0)$. All $w$-vertices can then be placed above their adjacent $v$-vertices (see Figure~\ref{Fig:frameGraphAppendix}). In the thus obtained layout, for every $i$, $1 \leq i \leq m + 1$, the unit squares for $u_{i, 0}$, $v_{i, 0}$, $w_{i, 1}$, $w_{i, 2}$ are positioned at $((i-1) \cdot 2(B+1), 0)$, $((i-1) \cdot 2(B+1), 2)$, $((i-1) \cdot 2(B+1), 4)$, $((i-1) \cdot 2(B+1), 6)$, respectively. Consequently, for every $i$, $1 \leq i \leq m$, and $\ell \in \{3, 4, 5, 6\}$, the coordinates $(p_i + 1,\ell), (p_i + 2,\ell), \ldots, (p_i + 2B + 1,\ell)$, where $p_i = (i-1) \cdot 2(B+1)$, are free and form a $4 \times (2B + 1)$ rectangle (note that, apart from the coordinates ``in between'' the already placed unit squares, these are the only free coordinates). Now let $A_i = \{a_{q_{i, 1}}, a_{q_{i, 2}}, a_{q_{i, 3}}\}$, $1 \leq i \leq m$. Since $a_{q_{i, 1}} + a_{q_{i, 2}} + a_{q_{i, 3}} = B$, the three connected components on vertices $b_{q_{i, r}, s}$ and $c_{q_{i, r}, s}$, $1 \leq r \leq 3$, $1 \leq s \leq a_{q_{i, r}}$, can be placed horizontally on the free coordinates $(p_i + 1,\ell), (p_i + 2,\ell), \ldots, (p_i + 2B + 1,\ell)$ with $p_i = (i-1) \cdot 2(B+1)$ for $\ell = 4$ and $\ell =6$, respectively, as illustrated in Figure~\ref{Fig:frameGraphAppendixTwo}. This constructs a $(7 \times (2(mB+m+1)-1))$ unit square grid layout for $G$.\par 
In order to prove the other direction, we assume that there exists a $(7 \times (2(mB+m+1)-1))$ unit square grid layout for $G$. We first note that, in any such layout, the unit squares for the $u$- and $v$-vertices must be represented as a horizontally or vertically oriented ``ladder'' and the same holds for the subgraphs on vertices $b_{i, j}$ and $c_{i, j}$. Moreover, since the layout has height of only $7$, we can further assume that the orientation for these ladders is indeed horizontal. Due to the fact that the layout has width $2(mB+m+1)-1$, all $mB+m+1$ many $u$ vertices are placed on coordinates $(2x, y_u)$, $0 \leq x \leq mB+m$, all $mB+m+1$ many $v$ vertices are placed on coordinates $(2x, y_v)$, $0 \leq x \leq mB+m$ (otherwise the ladder simply would not fit), and $|y_v - y_u| > 1$. Further, no unit square is placed at $y$-coordinate $y_u-1,y_u+1,y_v-1$ or $y_v+1$, otherwise it would intersect with at least one of the $u$- or $v$-vertices. Without loss of generality, we assume $y_u < y_v$. Since, for every $i$, $1 \leq i \leq m + 1$, the edge $\{v_{i, 0}, w_{i, 1}\}$ must be realised by a visibility of the form $R_{w_{i, 1}} \Vvis R_{v_{i, 0}}$ (note that the other three visibilities of $R_{v_{i, 0}}$ are already used), we conclude that $y_v \leq 4$.  The ladders from $G_A$ require two non-adjacent $y$-coordinates which are not blocked by the $u$- and $v$-vertices, so if  $y_v = 4$, then $y_u=0$ and if $y_v<4$ then $y_v=2$ and $y_u=0$.
%If $y_v = 4$, then $y_u \in \{0, 1, 2\}$. However, if $y_u \in \{1, 2\}$, then it is not possible to fit any other horizontal ladder in the layout. If $y_v = 3$, then $y_u \in \{0, 1\}$, which leads to a contradiction in the same way. Finally, if $y_v = 2$, then $y_u = 0$ obviously holds. Consequently, $y_u = 0$ and $y_v \in \{2, 4\}$. 
If $y_v = 4$, then one side of every horizontal ladder from $G_A$  must be represented by unit squares on coordinates $(2x-1, 6)$, $1 \leq x \leq mB+m$, and in order to represent all such ladders (which represent the subgraphs on vertices $b_{i,j}$ and $c_{i, j}$), $mB$ many of those coordinates are needed. However, for every $i$, $1 \leq i \leq m$, positioning the unit squares for vertices $w_{i, 1}$ and $w_{i, 2}$ blocks at least two of these coordinates (either because one of those unit squares is placed on it or directly next to it). Thus, there are not enough free coordinates to position all these unit squares and therefore we conclude that $y_u = 0$ and $y_v = 2$, which means that the ladder representing the $u$- and $v$-vertices is as illustrated in Figure~\ref{Fig:frameGraphAppendix}. Moreover, for every $i$, $1 \leq i \leq m+1$, we have $R_{w_{i, 1}} \Vvis R_{v_{i, 0}}$ and either $R_{w_{i, 2}} \Vvis R_{w_{i, 1}}$ or $R_{w_{i, 1}} \symHvis R_{w_{i, 2}}$. We assume now that, for every $i$, $1 \leq i \leq m+1$, the former case holds (and consider the latter case later on), which implies that the frame graph is represented as illustrated in Figure~\ref{Fig:frameGraphAppendix}.\par
As mentioned above, for every $i$, $1 \leq i \leq 3m$, the subgraph on vertices $b_{i, j}, c_{i, j}$, $1 \leq j \leq a_i$, is represented by a horizontal ladder, for which the unit squares must be placed on the free coordinates $(p_i + 1,\ell), (p_i + 2,\ell), \ldots, (p_i + 2B + 1,\ell)$, where $p_i = (i-1) \cdot 2(B+1)$ and $\ell \in \{3, 4, 5, 6\}$. This is only possible if the unit squares are aligned with the unit squares of the frame graph (as shown in the middle of the layout of Figure~\ref{Fig:frameGraphAppendixTwo}) or horizontally shifted by one unit (as shown on the left of the layout of Figure~\ref{Fig:frameGraphAppendixTwo}). We now only consider the coordinates $(x, 4)$, $0 \leq x \leq 2(mB+m)$, and note that in total we have to fit $(m + 1) + mB$ unit squares on these coordinates (the unit squares for vertices $w_{i, 1}$ (or an equivalent free space for the visibility $R_{w_{i,1}}\Vvis R{v_{i,0}}$) plus the unit squares for the vertices $b_{i, j}$). We need a gap of one unit in between every two neighbouring unit squares, we need at least $2(m + 1 + mB)-1 = 2(mB + m)+1$ coordinates; observe that if $R_{w_{i,1}}$ is not placed at $y$-coordinate $4$, we still can not place a square for say some $b_{i,j}$ from a $G_A$-ladder at vertical distance one from $R_{w_{i,1}}$, as the square for the corresponding neighbour $c_{i,j}$ would have to be placed right above $R_{b_{i,j}}$ and intersect with $R_{w_{i,1}}$. This implies that for each $1\leq i\leq m$ exactly the coordinates $(2(i-1)(B+1)+2x, 4)$, $1 \leq x \leq B$, are occupied by unit squares, which means that all the horizontal ladders must be aligned with the unit squares of the frame graph. Consequently, partitioning $A$ according to how the ladders are placed in the compartments of the frame graph yields a solution for the $\threepartition$-instance $(B, A)$.\par
The argument for the case that $R_{w_{i, 1}} \symHvis R_{w_{i, 2}}$ is analogous; the only difference is that even more coordinates are already occupied by unit squares. 
\end{proof}

\subsection*{Proof of Proposition~\ref{atMostFourVerticesProposition}}

%\begin{proposition}
%\emph{Every graph with at most $4$ vertices is in $\unitSquareGraphs$.}
%\end{proposition}

\begin{proof}
It is straightforward to construct layouts for graphs with at most $3$ vertices (thus, also for graphs with $4$ vertices that are not connected) and for $P_4$, $C_4$ and $K_{1, 3}$. This only leaves $K_4$, for which a layout is presented in Figure~\ref{visLayoutsCompleteBipartiteGraphs}, and the two graphs represented by the following layouts:
\hspace{3em}
\begin{tikzpicture}[scale=0.7]

\coordinate (width) at (0.4,0);
\coordinate (height) at (0,0.4);

\coordinate (1) at (0,0);
\coordinate (2) at ($(1) + 1.25*(width)$);
\coordinate (3) at ($(1) + 1.25*0.5*(width) + 1.25*(height)$);
\coordinate (4) at ($(1) + 2.5*(width)$); 

\coordinate (temp) at (1);
\draw[black] ($(temp)$) -- ($(temp) + (width)$) -- ($(temp) + (width) + (height)$) -- ($(temp) + (height)$) -- ($(temp)$);
\coordinate (temp) at (2);
\draw[black] ($(temp)$) -- ($(temp) + (width)$) -- ($(temp) + (width) + (height)$) -- ($(temp) + (height)$) -- ($(temp)$);
\coordinate (temp) at (3);
\draw[black] ($(temp)$) -- ($(temp) + (width)$) -- ($(temp) + (width) + (height)$) -- ($(temp) + (height)$) -- ($(temp)$);
\coordinate (temp) at (4);
\draw[black] ($(temp)$) -- ($(temp) + (width)$) -- ($(temp) + (width) + (height)$) -- ($(temp) + (height)$) -- ($(temp)$);

\end{tikzpicture}
\hspace{3em}
\begin{tikzpicture}[scale=0.7]

\coordinate (width) at (0.4,0);
\coordinate (height) at (0,0.4);

\coordinate (1) at (0,0);
\coordinate (2) at ($(1) + 1.25*(width)$);
\coordinate (3) at ($(1) + 1.25*0.5*(width) + 1.25*(height)$);
\coordinate (4) at ($(1) + 2.5*(width) + 0.75*(height)$); 

\coordinate (temp) at (1);
\draw[black] ($(temp)$) -- ($(temp) + (width)$) -- ($(temp) + (width) + (height)$) -- ($(temp) + (height)$) -- ($(temp)$);
\coordinate (temp) at (2);
\draw[black] ($(temp)$) -- ($(temp) + (width)$) -- ($(temp) + (width) + (height)$) -- ($(temp) + (height)$) -- ($(temp)$);
\coordinate (temp) at (3);
\draw[black] ($(temp)$) -- ($(temp) + (width)$) -- ($(temp) + (width) + (height)$) -- ($(temp) + (height)$) -- ($(temp)$);
\coordinate (temp) at (4);
\draw[black] ($(temp)$) -- ($(temp) + (width)$) -- ($(temp) + (width) + (height)$) -- ($(temp) + (height)$) -- ($(temp)$);
\end{tikzpicture}
\end{proof}

\subsection*{Proof of Theorem~\ref{nonGridTreeTheorem}}

%\begin{theorem}
%\emph{Let $T$ be a tree with maximum degree $k$.\\ If $k \leq 5$, then $T \in \unitSquareGraphs$, and if $k \geq 7$, then $T \notin \unitSquareGraphs$.}
%If $T$ is a tree with maximum degree $5$, then $T \in \unitSquareGraphs$.
%\end{theorem}

\begin{proof}
The second statement follows from the fact that for unit square visibility graphs, any vertex of degree at least $7$ lies on a cycle, which has been shown in \cite{DeaEHP2008}. \par
Let $T \in \unitSquareGraphs$ be a tree with a maximum degree of $5$ represented by a layout $\mathcal{R}$. We show that if we append at most $4$ nodes to an arbitrary leaf of $T$, the resulting tree can still be represented by a layout. The first statement of the lemma follows then by induction. Let $v$ be a leaf of $T$ with a parent node $u$ and let $R_v, R_u \in \mathcal{R}$ be the corresponding unit squares. Without loss of generality, we assume that $R_v \Vvis R_u$. Next, we note that that there is no $R \in \mathcal{R}$ with $R \Hvis R_v$, $R_v \Hvis R$ or $R_v \Vvis R$, which, in particular, means that $R_{v}$ can be moved arbitrarily far down without destroying or introducing any visibilities. Consequently, we can assume that the two rectangles of height $0.5$ and infinite width just above and below $R_{v}$ are not intersected by any $R \in \mathcal{R}$. This implies that we can append new vertices $w_i$, $1 \leq i \leq 4$, to $v$ by placing new unit squares $R_{w_i}$, $1 \leq i \leq 4$, as follows:

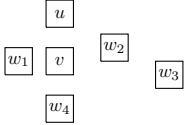
\begin{tikzpicture} [scale=0.9]

\coordinate (width) at (0.4,0);
\coordinate (height) at (0,0.4);  
\coordinate (labelshift) at (0.2, 0.2);

\coordinate (u) at (0,0);
\coordinate (v) at ($(u) - 1.75*(height)$);
\coordinate (w1) at ($(v) - 1.5*(width)$);
\coordinate (w2) at ($(v) + 2*(width) + 0.5*(height)$);
\coordinate (w3) at ($(v) + 4*(width) - 0.5*(height)$);
\coordinate (w4) at ($(v) - 1.75*(height)$);

\coordinate (temp) at (u);
\draw[black] ($(temp)$) -- ($(temp) + (width)$) -- ($(temp) + (width) + (height)$) -- ($(temp) + (height)$) -- ($(temp)$);
\node[scale=0.7] (label) at ($(temp) + (labelshift)$)  {$u$};
\coordinate (temp) at (v);
\draw[black] ($(temp)$) -- ($(temp) + (width)$) -- ($(temp) + (width) + (height)$) -- ($(temp) + (height)$) -- ($(temp)$);
\node[scale=0.7] (label) at ($(temp) + (labelshift)$)  {$v$};
\coordinate (temp) at (w1);
\draw[black] ($(temp)$) -- ($(temp) + (width)$) -- ($(temp) + (width) + (height)$) -- ($(temp) + (height)$) -- ($(temp)$);
\node[scale=0.7] (label) at ($(temp) + (labelshift)$)  {$w_1$};
\coordinate (temp) at (w2);
\draw[black] ($(temp)$) -- ($(temp) + (width)$) -- ($(temp) + (width) + (height)$) -- ($(temp) + (height)$) -- ($(temp)$);
\node[scale=0.7] (label) at ($(temp) + (labelshift)$)  {$w_2$};
\coordinate (temp) at (w3);
\draw[black] ($(temp)$) -- ($(temp) + (width)$) -- ($(temp) + (width) + (height)$) -- ($(temp) + (height)$) -- ($(temp)$);
\node[scale=0.7] (label) at ($(temp) + (labelshift)$)  {$w_3$};
\coordinate (temp) at (w4);
\draw[black] ($(temp)$) -- ($(temp) + (width)$) -- ($(temp) + (width) + (height)$) -- ($(temp) + (height)$) -- ($(temp)$);
\node[scale=0.7] (label) at ($(temp) + (labelshift)$)  {$w_4$};
 
\end{tikzpicture}

\noindent 
Moreover, the only new edges are between the $w_i$, $1 \leq i \leq 4$, and $v$, and no existing edges are destroyed. Consequently, the obtained layout represents the tree $T'$ that is obtained from $T$ by appending $4$ new nodes to the leaf $v$. In a similar way, we can also append less than $4$ new vertices to $v$.
\end{proof}

\subsection*{Proof of Theorem~\ref{NPMembershipTheorem}}

\begin{proof}
Assuming there exists a $\unitSquareGraphs$ layout for a graph $G$ over $n$ vertices, this layout can obviously be considered to use space reasonably, hence with $x$- and $y$-coordinates within range $0$ to $n$. Further, squares do not have to be shifted arbitrarily: Shifting the $x$-coordinate of a rectangle $R$ with respect to the $x$-coordinate of another rectangle $R'$ by more than zero but less than one is only necessary if $R$ needs to see another rectangle to the same side as $R'$. The number of different shifts of distance strictly between zero and one which are necessary for a layout is hence bounded by the maximum degree of the input graph. In general, this means that if $G\in \unitSquareGraphs$, guessing all possibilities to choose coordinates $(x,y)$ with $x,y\in \{\frac an\mid 0\leq a\leq n^2\}$ for each vertex in $G$ yields at least one layout for $G$. Since checking if a set of coordinates yields a feasible layout for a graph $G$ can be done in polynomial time, this kind of guessing  $n$ coordinates from a set of $(n+1)^4$ possibilities yields $\npclass$-membership for $\recognitionProb(\unitSquareGraphs)$. For $\recognitionProb(\unitSquareGridGraphs)$, the similar arguments apply and it is even sufficient to only guess integer coordinates $(x,y)$ with $0\leq x,y\leq 2n-1$. 
\end{proof}

\newpage

\subsection*{Proof of Lemma~\ref{nonGridReductionEasyDirectionLemma}}

\begin{proof}
Let $F$ be a not-all-equal satisfiable 3-SAT formula with clauses $c_1,\dots,c_m$ over variables $v_1,\dots, v_n$ and let $\phi\colon \{x_1\dots,x_n\}\rightarrow \{0,1\}$ be an according not-all-equal satisfying assignment. The following coordinates yield a $\unitSquareGraphs$ drawing for the corresponding graph $G$ (see Figure~\ref{reductionIllustrationFigure} for an illustration):

For  $j\in \{1,\dots,2m\}$, $h\in\{1,2\}$, $i\in \{1,\dots,n\}$, $r\in\{1,2,3\}$, \begin{center}
\begin{tabular}{c||@{\ }l|@{\ }l}
vertex \ &\ $x$-coordinate \ & \ $y$-coordinate\ \\
\hline
&& \\[-2.3ex]
\hline
&& \\[-2.3ex]
$c_j$&$4j$& $0$ \\
&&  \\[-2.3ex]
\hline
&&\\[-2.3ex]
$c_j^h$&$4j+2$& $2-1.3h$ \\
&&  \\[-2.3ex]
\hline
&& \\[-2.3ex]
$x_i$&$8(m+i)$& $0$  \\
&&  \\[-2.3ex]
\hline
&& \\[-2.3ex]
$x_i^h$&$8(m+i)-6$& $2-1.3h$ \\
&&  \\[-2.3ex]
\hline
&& \\[-2.3ex]
$t_i$&$8(m+i)$& $(-1)^{(1-\phi(x_i))}(5i+2)+1-\frac{2m+4}{2m+8}$\\
&&  \\[-2.3ex]
\hline
&& \\[-2.3ex]
$f_i^h$&$8(m+i)+\frac{3}{2}-h$& $(-1)^{\phi(x_i)}(5i+2h)+1-\frac{2m+4}{2m+8}$\\
&&  \\[-2.3ex]
\hline
&& \\[-2.3ex]
$l_j^r$&$4j+\frac12(\frac{r-k}{|r-k|})$ &$(-1)^{(1-\phi(x_i))}(5i+2)+1-\frac{j+2}{2m+8}$\\
& \multicolumn{2}{l}{} \\[-2ex]
&\multicolumn{2}{c}{ for $y_{j,h}=x_i$, $k=\text{argmax}\{|r-k|\ \mid \phi(l_j^r)=\phi(l_j^k)\}$}\\[1ex]
\hline
&& \\[-2.3ex]
$l_j^r$&$4j+\frac12(\frac{r-k}{|r-k|})$& $(-1)^{\phi(x_i)}(5i+2\lceil \frac{j}{m}\rceil)+1-\frac{j+2}{2m+8}$\\& \multicolumn{2}{l}{} \\[-2ex]&\multicolumn{2}{c}{ for $y_{j,h}=\bar x_i$, $k=\text{argmax}\{|r-k| \ \mid \phi(l_j^r)=\phi(l_j^k)\}$}\\[1ex]
\hline
&& \\[-2.3ex]
 $\overset{_{\rightarrow}}{t_i}$&$-9i$&$(-1)^{(1-\phi(x_i)}(5i+2)$\\
&&  \\[-2.3ex]
\hline
&& \\[-2.3ex]
 $\overset{_{\leftarrow}}{t_i}$& $8(n+1+m)+9i$&$(-1)^{(1-\phi(x_i)}(5i+2) +1$\\
&&  \\[-2.3ex]
\hline
&& \\[-2.3ex]
$\overset{_{\rightarrow}}{f_i^h}$&$-9i-h$& $(-1)^{\phi(x_i)}(5i+2h)$\\
&&  \\[-2.3ex]
\hline
&& \\[-2.3ex]
$\overset{_{\leftarrow}}{f_i^h}$&$8(n+1+m)+9i-h$& $(-1)^{\phi(x_i)}(5i+2h)+1$\\
&&  \\[-2.3ex]
\hline
&& \\[-2.3ex]
$h^0_{t_i}$&$8(m+i)-3$& $(-1)^{(1-\phi(x_i)}(5i+2)+1-\frac{2m+3}{2m+8}$\\
&&  \\[-2.3ex]
\hline
&& \\[-2.3ex]
$h^0_{f_i^h}$&$8(m+i)-2h$& $(-1)^{\phi(x_i)}(5i+2h)+1-\frac{2m+3}{2m+8}$\\
&&  \\[-2.3ex]
\hline
&& \\[-2.3ex]
$h^r_{t_i}$&$-9i+3r$&$(-1)^{(1-\phi(x_i)}(5i+2) +1-\frac{r}{2m+8} $ \\
& \multicolumn{2}{c}{} \\[-2ex]&\multicolumn{2}{c}{for $r\in\{1,2\}$}\\[1ex]
\hline 
&& \\[-2.3ex]
$h^r_{f_i^h}$&$-9i+3r-h$& $(-1)^{\phi(x_i)}(5i+2h)+1-\frac{r}{2m+8}$\\
& \multicolumn{2}{c}{} \\[-2ex]&\multicolumn{2}{c}{for $r\in\{1,2\}$}\\[1ex]
\hline
&& \\[-2.3ex]
$h^r_{t_i}$&$8(n+1+m)+9i+3r-15$&$(-1)^{(1-\phi(x_i)}(5i+2) +\frac{5-r}{2m+8} $\\
& \multicolumn{2}{c}{} \\[-2ex]&\multicolumn{2}{c}{ for $r\in\{3,4\}$}\\[1ex]
 \hline
&& \\[-2.3ex]
$h^r_{f_i^h}$&$8(n+1+m)+9i+3r-h-15$& $(-1)^{\phi(x_i)}(5i+2h)+\frac{5-r}{2m+8}$\\
& \multicolumn{2}{c}{} \\[-2ex]&\multicolumn{2}{c}{for  $r\in\{3,4\}$}\\
\end{tabular}
\end{center}

\end{proof}

%\newgeometry{bottom=1cm}
\begin{landscape}
\thispagestyle{empty}
\begin{figure}

\hspace{-1.5cm}\input{reduction_figure.tex}

\caption{An illustration of the visibility layout for the not-all-equal satisfiable formula $\{c_1, c_2, c_3\}$ with $c_1 = \{x_1, \bar x_2, x_3\}, c_2 = \{x_1, x_3, \bar x_4\}, c_3 = \{\bar x_2, x_3, x_4\}$ (note that, since the clauses are copied, we have $2m = 6$ clause gadgets). The vertical visibilities of the unit squares for the clause vertices $c_j$, $1 \leq j \leq 6$, and of the unit squares of the variable vertices $x_i$, $1 \leq i \leq 4$, are highlighted in grey. In addition, the left-to-right horizontal visibilities of the unit squares for the vertices $\overset{_{\rightarrow}}{f_{2}^{2}}$, $\overset{_{\rightarrow}}{t_{2}}$, $\overset{_{\rightarrow}}{t_{3}}$, and the right-to-left horizontal visibilities of the unit squares for the vertices $\overset{_{\leftarrow}}{f_{2}^{2}}$, $\overset{_{\leftarrow}}{t_{2}}$, $\overset{_{\leftarrow}}{t_{3}}$ are highlighted as well (this should also illustrate how the edges between the vertices $\overset{_{\rightarrow}}{t_{i}}, h^1_{t_i}, h^2_{t_i}, l^{q_1}_{j_1}, \ldots, l^{r_q}_{j_q}, h^0_{t_i}, t_i, h^3_{t_i}, h^4_{t_i}, \overset{_{\leftarrow}}{t_{i}}$ (and the corresponding vertex sets including $f^1_{i}$ or $f^2_{i}$) are represented by the layout). The satisfying assignment can be obtained by setting $x_i$ to \emph{true} if and only if $R_{x_i} \Vvis R_{t_i}$, which yields the assignment $x_1 \to 0$, $x_2 \to 1$, $x_3 \to 1$, $x_4 \to 0$.}
\label{reductionIllustrationFigure}
\end{figure}
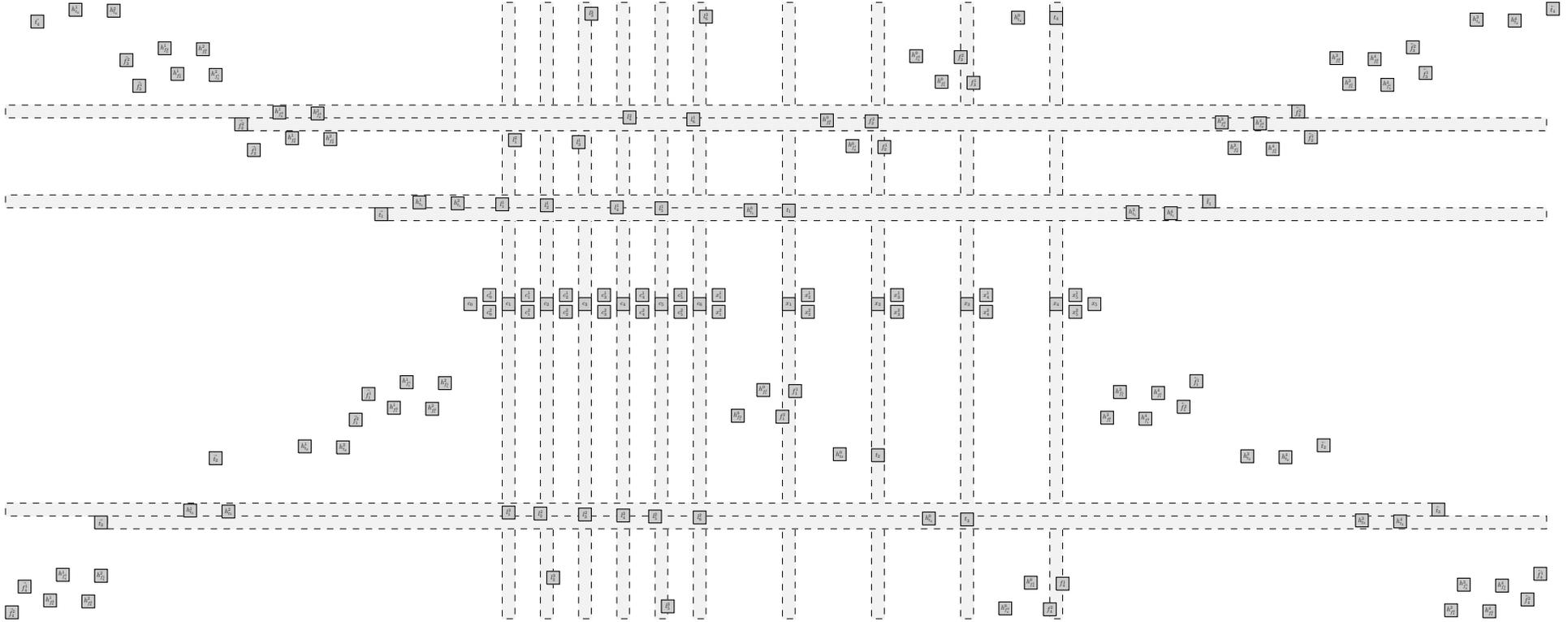
\end{landscape}
%\end{sideways}
%\restoregeometry

\subsection*{Proof of Lemma~\ref{K4Lemma}}

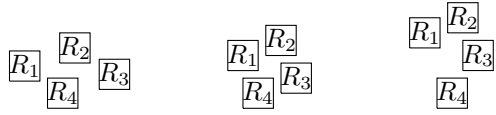
\begin{figure}
\begin{center}

\begin{tikzpicture}

\coordinate (width) at (0.4,0);
\coordinate (height) at (0,0.4);
\coordinate (labelshift) at (0.2, 0.2);

\coordinate (1) at (0,0);
\coordinate (2) at ($(1) + 1.5*(height) + 0.4*(width)$);
\coordinate (3) at ($(1) + 1.25*0.5*(height) + 1.7*(width)$);
\coordinate (4) at ($(1) + 0.9*(height) - 1.25*(width)$);

\coordinate (temp) at (1);
\draw[black] ($(temp)$) -- ($(temp) + (width)$) -- ($(temp) + (width) + (height)$) -- ($(temp) + (height)$) -- ($(temp)$);
\node (label) at ($(temp) + (labelshift)$)  {$R_4$};

\coordinate (temp) at (2);
\draw[black] ($(temp)$) -- ($(temp) + (width)$) -- ($(temp) + (width) + (height)$) -- ($(temp) + (height)$) -- ($(temp)$);
\node (label) at ($(temp) + (labelshift)$)  {$R_2$};

\coordinate (temp) at (3);
\draw[black] ($(temp)$) -- ($(temp) + (width)$) -- ($(temp) + (width) + (height)$) -- ($(temp) + (height)$) -- ($(temp)$);
\node (label) at ($(temp) + (labelshift)$)  {$R_3$};

\coordinate (temp) at (4);
\draw[black] ($(temp)$) -- ($(temp) + (width)$) -- ($(temp) + (width) + (height)$) -- ($(temp) + (height)$) -- ($(temp)$);
\node (label) at ($(temp) + (labelshift)$)  {$R_1$};

\end{tikzpicture}
\hspace{0.8cm}
\begin{tikzpicture}

\coordinate (width) at (0.4,0);
\coordinate (height) at (0,0.4);
\coordinate (labelshift) at (0.2, 0.2);

\coordinate (1) at ($(0,0)$);
\coordinate (2) at ($(1) + 1.25*(width) + 0.5*(height)$);
\coordinate (3) at ($(1) + 0.5*(width) - 1.25*(height)$);
\coordinate (4) at ($(1) + 1.75*(width) - 0.75*(height)$);

\coordinate (temp) at (1);
\draw[black] ($(temp)$) -- ($(temp) + (width)$) -- ($(temp) + (width) + (height)$) -- ($(temp) + (height)$) -- ($(temp)$);
\node (label) at ($(temp) + (labelshift)$)  {$R_1$};

\coordinate (temp) at (2);
\draw[black] ($(temp)$) -- ($(temp) + (width)$) -- ($(temp) + (width) + (height)$) -- ($(temp) + (height)$) -- ($(temp)$);
\node (label) at ($(temp) + (labelshift)$)  {$R_2$};

\coordinate (temp) at (3);
\draw[black] ($(temp)$) -- ($(temp) + (width)$) -- ($(temp) + (width) + (height)$) -- ($(temp) + (height)$) -- ($(temp)$);
\node (label) at ($(temp) + (labelshift)$)  {$R_4$};

\coordinate (temp) at (4);
\draw[black] ($(temp)$) -- ($(temp) + (width)$) -- ($(temp) + (width) + (height)$) -- ($(temp) + (height)$) -- ($(temp)$);
\node (label) at ($(temp) + (labelshift)$)  {$R_3$};

\end{tikzpicture}
\hspace{0.8cm}
\begin{tikzpicture}

\coordinate (width) at (0.4,0);
\coordinate (height) at (0,0.4);
\coordinate (labelshift) at (0.2, 0.2);

\coordinate (1) at ($(0,0)$);
\coordinate (2) at ($(1) + 1.25*(width) + 0.5*(height)$);
\coordinate (3) at ($(1) + 0.9*(width) - 2*(height)$);
\coordinate (4) at ($(1) + 1.75*(width) - 0.75*(height)$);

\coordinate (temp) at (1);
\draw[black] ($(temp)$) -- ($(temp) + (width)$) -- ($(temp) + (width) + (height)$) -- ($(temp) + (height)$) -- ($(temp)$);
\node (label) at ($(temp) + (labelshift)$)  {$R_1$};

\coordinate (temp) at (2);
\draw[black] ($(temp)$) -- ($(temp) + (width)$) -- ($(temp) + (width) + (height)$) -- ($(temp) + (height)$) -- ($(temp)$);
\node (label) at ($(temp) + (labelshift)$)  {$R_2$};

\coordinate (temp) at (3);
\draw[black] ($(temp)$) -- ($(temp) + (width)$) -- ($(temp) + (width) + (height)$) -- ($(temp) + (height)$) -- ($(temp)$);
\node (label) at ($(temp) + (labelshift)$)  {$R_4$};

\coordinate (temp) at (4);
\draw[black] ($(temp)$) -- ($(temp) + (width)$) -- ($(temp) + (width) + (height)$) -- ($(temp) + (height)$) -- ($(temp)$);
\node (label) at ($(temp) + (labelshift)$)  {$R_3$};

\end{tikzpicture}

\caption{The three ways of representing $K_4$ by a layout.}
\label{visLayoutsCompleteBipartiteGraphsFigureApp}
\end{center}
\end{figure}

We note that these three possibilities are uniquely determined by the horizontal and vertical visibilities (up to a renaming of the unit squares), e.\,g., for the first layout of Figure~\ref{visLayoutsCompleteBipartiteGraphsFigureApp}, we have $R_1 \Hvis \{R_2, R_3, R_4\}$, $R_2 \Hvis R_3$, $R_2 \Vvis R_4$, $R_4 \Hvis R_3$. We shall now formally prove that any layout for $K_4$ is \visomorphic{} to one of the three layouts of Figure~\ref{visLayoutsCompleteBipartiteGraphsFigureApp} (note that the three cases of the following lemma correspond to the three layouts of Figure~\ref{visLayoutsCompleteBipartiteGraphsFigureApp}).
To make these cases easier to understand in connection with the previous sentences, we re-worded  Lemma~\ref{K4Lemma} and cast it in the following form.

\begin{lemma}
Every layout for $K_4$ is isomorphic to a layout $\{R_1, R_2, R_3, R_4\}$ that satisfies one of the following cases:
\begin{enumerate}
\item $R_1 \Hvis \{R_2, R_3, R_4\}$, $R_2 \Hvis R_3$, $R_2 \Vvis R_4$, $R_4 \Hvis R_3$,
\item $R_1 \Hvis \{R_2, R_3\}$, $R_1 \Vvis R_4$, $R_2 \Vvis \{R_3, R_4\}$, $R_4 \Hvis R_3$,
\item $R_1 \Hvis \{R_2, R_3\}$, $R_1 \Vvis R_4$, $R_2 \Vvis \{R_3, R_4\}$, $R_3 \Vvis R_4$.
\end{enumerate}
\end{lemma}

\begin{proof}
It can be easily verified that at least one of the edges of $K_4$ must be represented by a visibility of length strictly less than $1$. Hence, we assume that this is true for the visibility between $R_1$ and $R_2$ and, furthermore, we assume that $R_1 \Hvis R_2$ and that for the $y$-components $y_1$ and $y_2$ of the coordinates of $R_1$ and $R_2$, respectively, we have $y_2 \leq y_1$ (i.\,e., $R_1$ is to the left of $R_2$ and $R_2$ is either horizontally aligned with $R_1$ or further down (see also Figure~\ref{K4LemmaFigureApp})). We now investigate all possibilities of how the remaining unit squares $R_3$ and $R_4$ can be placed in the layout in order to represent $K_4$. 
\begin{itemize}
\item $\{R_3, R_4\} \symVvis \{R_1, R_2\}$: This implies that $R_3$ must be placed above and $R_4$ below $R_1$ and $R_2$, or vice versa (see Figure~\ref{K4LemmaFigureApp}$(a)$), which means that we have case $1$.
\item $\{R_3, R_4\} \symHvis \{R_1, R_2\}$: If $R_3$ and $R_4$ are placed on opposite sides of $R_1$ and $R_2$, then they either cannot see each other or one of them cannot see $R_1$ or $R_2$. If they are placed on the same side of $R_1$ and $R_2$, then at most one of them can see both $R_1$ or $R_2$. Thus, this case is not possible.
\item $\{R_3\} \symHvis \{R_1, R_2\}$ and $\{R_4\} \symVvis \{R_1, R_2\}$ or $\{R_3\} \symVvis \{R_1, R_2\}$ and $\{R_4\} \symHvis \{R_1, R_2\}$: We only consider case $\{R_3\} \symHvis \{R_1, R_2\}$ and $\{R_4\} \symVvis \{R_1, R_2\}$, since the other case is symmetric. If $R_3 \Hvis \{R_1, R_2\}$, then $\{R_1, R_2\} \Vvis R_4$ and $R_3\Hvis R_4$, which means that we have case $3$. Analogously, if $\{R_1, R_2\} \Hvis R_3$, then $R_4 \Vvis \{R_1, R_2\}$ and $R_4\Hvis R_3$, which again means that we have case $3$.
\end{itemize}
Hence, from now on, we can assume that at least one of $R_3$ and $R_4$ is placed such that it sees one of $R_1$ and $R_2$ horizontally and the other one vertically. Without loss of generality, we assume that this is the case for $R_3$, which means that either $R_3 \symHvis R_1$ and $R_3 \symVvis R_2$ or $R_3 \symVvis R_1$ and $R_3 \symHvis R_2$. Moreover, due to the relative positions of $R_1$ and $R_2$, this is only possible if $R_1 \Hvis R_3$ and $R_3 \Vvis R_2$ or $R_1 \Vvis R_3$ and $R_3 \Hvis R_2$. We assume the former situation (as illustrated in Figure~\ref{K4LemmaFigureApp}$(b) - (e)$) and now check all possibilities of how $R_4$ can be placed in the layout in order to represent $K_4$. 
\begin{itemize}
\item $R_4 \symVvis \{R_1, R_2\}$: Since $R_4$ does not vertically fit between $R_1$ and $R_2$, this means that either $\{R_1,R_2\} \Vvis R_4$ or $R_4 \Vvis \{R_1,R_2\}$. The case $\{R_1,R_2\} \Vvis R_4$ implies $R_3\Vvis R_4$ (see Figure~\ref{K4LemmaFigureApp}$(b)$) and thus, we have case $3$.  $R_4 \Vvis \{R_1,R_2\}$ requires that for the $x$-coordinates $x_i$ of $R_i$ we have $x_1<x_4<x_2<x_3$ and hence either $R_4\Vvis R_3$ which also yields case $3$, or $R_4\Hvis R_3$ which yields case $2$.
\item $R_4 \symHvis \{R_1, R_2\}$: Since $R_4$ must see $R_3$, this means that $R_1 \Hvis R_4$ and $R_2 \Hvis R_4$ which means that either either $R_3\Hvis R_4$ (see Figure~\ref{K4LemmaFigureApp}$(c)$) which yields  case 1, or $R_3\Vvis R_4$ (see Figure~\ref{K4LemmaFigureApp}$(d)$) which yields case 3.
\item $R_4 \symHvis R_1$, $R_4 \symVvis R_2$: We note that if $R_4 \Hvis R_1$, then $R_4$ cannot see $R_2$ vertically, which implies $R_1 \Hvis R_4$. In particular, this also implies $R_4 \Vvis R_2$ (see Figure~\ref{K4LemmaFigureApp}$(d)$; $R_3$ and $R_4$ can switch). Consequently, we have case $3$.
\item $R_4 \symHvis R_2$, $R_4 \symVvis R_1$: Similarly to the previous case, if $R_4 \Vvis R_1$, then $R_4$ cannot see $R_2$ horizontally; thus, $R_1 \Vvis R_4$, which, in particular, implies $R_4 \Hvis R_2$ (see Figure~\ref{K4LemmaFigureApp}$(e)$). Consequently, we have case $2$.
\end{itemize}
The case where $R_1 \Vvis R_3$ and $R_3 \Hvis R_2$ is symmetric to the case $R_1 \Hvis R_3$ and $R_3 \Vvis R_2$ considered above. Furthermore, the cases that $y_1 \leq y_2$ or that the visibility between $R_1$ and $R_2$ is vertical can be handled analogously. This completes the proof.

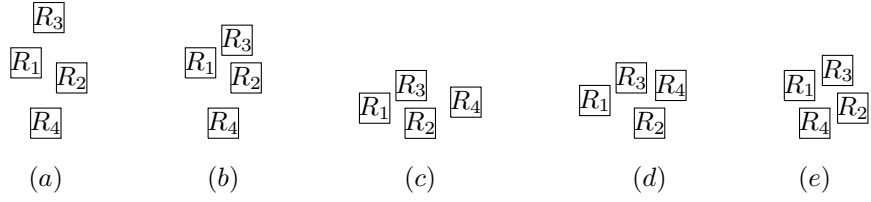
\begin{figure}
\begin{center}

\begin{tikzpicture}

\coordinate (width) at (0.4,0);
\coordinate (height) at (0,0.4);
\coordinate (labelshift) at (0.2, 0.2);

\coordinate (1) at ($(0,0)$);
\coordinate (2) at ($(1) + 1.5*(width) - 0.5*(height)$);
\coordinate (3) at ($(1) + 0.75*(width) + 1.5*(height)$);
\coordinate (4) at ($(1) + 0.65*(width) - 2*(height)$);
\coordinate (5) at ($(4) - (0,0.75)$);

\node (label) at ($(5) + (labelshift)$)  {$(a)$};

\coordinate (temp) at (1);
\draw[black] ($(temp)$) -- ($(temp) + (width)$) -- ($(temp) + (width) + (height)$) -- ($(temp) + (height)$) -- ($(temp)$);
\node (label) at ($(temp) + (labelshift)$)  {$R_1$};

\coordinate (temp) at (2);
\draw[black] ($(temp)$) -- ($(temp) + (width)$) -- ($(temp) + (width) + (height)$) -- ($(temp) + (height)$) -- ($(temp)$);
\node (label) at ($(temp) + (labelshift)$)  {$R_2$};

\coordinate (temp) at (3);
\draw[black] ($(temp)$) -- ($(temp) + (width)$) -- ($(temp) + (width) + (height)$) -- ($(temp) + (height)$) -- ($(temp)$);
\node (label) at ($(temp) + (labelshift)$)  {$R_3$};

\coordinate (temp) at (4);
\draw[black] ($(temp)$) -- ($(temp) + (width)$) -- ($(temp) + (width) + (height)$) -- ($(temp) + (height)$) -- ($(temp)$);
\node (label) at ($(temp) + (labelshift)$)  {$R_4$};

\end{tikzpicture}
\hspace{0.8cm}
\begin{tikzpicture}

\coordinate (width) at (0.4,0);
\coordinate (height) at (0,0.4);
\coordinate (labelshift) at (0.2, 0.2);

\coordinate (1) at ($(0,0)$);
\coordinate (2) at ($(1) + 1.5*(width) - 0.5*(height)$);
\coordinate (3) at ($(1) + 1.2*(width) + 0.75*(height)$);
\coordinate (4) at ($(1) + 0.75*(width) - 2*(height)$);
\coordinate (5) at ($(4) - (0,0.75)$);

\node (label) at ($(5) + (labelshift)$)  {$(b)$};

\coordinate (temp) at (1);
\draw[black] ($(temp)$) -- ($(temp) + (width)$) -- ($(temp) + (width) + (height)$) -- ($(temp) + (height)$) -- ($(temp)$);
\node (label) at ($(temp) + (labelshift)$)  {$R_1$};

\coordinate (temp) at (2);
\draw[black] ($(temp)$) -- ($(temp) + (width)$) -- ($(temp) + (width) + (height)$) -- ($(temp) + (height)$) -- ($(temp)$);
\node (label) at ($(temp) + (labelshift)$)  {$R_2$};

\coordinate (temp) at (3);
\draw[black] ($(temp)$) -- ($(temp) + (width)$) -- ($(temp) + (width) + (height)$) -- ($(temp) + (height)$) -- ($(temp)$);
\node (label) at ($(temp) + (labelshift)$)  {$R_3$};

\coordinate (temp) at (4);
\draw[black] ($(temp)$) -- ($(temp) + (width)$) -- ($(temp) + (width) + (height)$) -- ($(temp) + (height)$) -- ($(temp)$);
\node (label) at ($(temp) + (labelshift)$)  {$R_4$};

\end{tikzpicture}
\hspace{0.8cm}
\begin{tikzpicture}

\coordinate (width) at (0.4,0);
\coordinate (height) at (0,0.4);
\coordinate (labelshift) at (0.2, 0.2);

\coordinate (1) at ($(0,0)$);
\coordinate (2) at ($(1) + 1.5*(width) - 0.5*(height)$);
\coordinate (3) at ($(1) + 1.2*(width) + 0.75*(height)$);
\coordinate (4) at ($(1) + 3*(width) + 0.2*(height)$);
\coordinate (5) at ($(2) - (0,0.75)$);

\node (label) at ($(5) + (labelshift)$)  {$(c)$};

\coordinate (temp) at (1);
\draw[black] ($(temp)$) -- ($(temp) + (width)$) -- ($(temp) + (width) + (height)$) -- ($(temp) + (height)$) -- ($(temp)$);
\node (label) at ($(temp) + (labelshift)$)  {$R_1$};

\coordinate (temp) at (2);
\draw[black] ($(temp)$) -- ($(temp) + (width)$) -- ($(temp) + (width) + (height)$) -- ($(temp) + (height)$) -- ($(temp)$);
\node (label) at ($(temp) + (labelshift)$)  {$R_2$};

\coordinate (temp) at (3);
\draw[black] ($(temp)$) -- ($(temp) + (width)$) -- ($(temp) + (width) + (height)$) -- ($(temp) + (height)$) -- ($(temp)$);
\node (label) at ($(temp) + (labelshift)$)  {$R_3$};

\coordinate (temp) at (4);
\draw[black] ($(temp)$) -- ($(temp) + (width)$) -- ($(temp) + (width) + (height)$) -- ($(temp) + (height)$) -- ($(temp)$);
\node (label) at ($(temp) + (labelshift)$)  {$R_4$};

\end{tikzpicture}
\hspace{0.8cm}
\begin{tikzpicture}

\coordinate (width) at (0.4,0);
\coordinate (height) at (0,0.4);
\coordinate (labelshift) at (0.2, 0.2);

\coordinate (1) at ($(0,0)$);
\coordinate (2) at ($(1) + 1.8*(width) - 0.75*(height)$);
\coordinate (3) at ($(1) + 1.2*(width) + 0.75*(height)$);
\coordinate (4) at ($(1) + 2.5*(width) + 0.5*(height)$);
\coordinate (5) at ($(2) - (0,0.75)$);

\node (label) at ($(5) + (labelshift)$)  {$(d)$};

\coordinate (temp) at (1);
\draw[black] ($(temp)$) -- ($(temp) + (width)$) -- ($(temp) + (width) + (height)$) -- ($(temp) + (height)$) -- ($(temp)$);
\node (label) at ($(temp) + (labelshift)$)  {$R_1$};

\coordinate (temp) at (2);
\draw[black] ($(temp)$) -- ($(temp) + (width)$) -- ($(temp) + (width) + (height)$) -- ($(temp) + (height)$) -- ($(temp)$);
\node (label) at ($(temp) + (labelshift)$)  {$R_2$};

\coordinate (temp) at (3);
\draw[black] ($(temp)$) -- ($(temp) + (width)$) -- ($(temp) + (width) + (height)$) -- ($(temp) + (height)$) -- ($(temp)$);
\node (label) at ($(temp) + (labelshift)$)  {$R_3$};

\coordinate (temp) at (4);
\draw[black] ($(temp)$) -- ($(temp) + (width)$) -- ($(temp) + (width) + (height)$) -- ($(temp) + (height)$) -- ($(temp)$);
\node (label) at ($(temp) + (labelshift)$)  {$R_4$};

\end{tikzpicture}
\hspace{0.8cm}
\begin{tikzpicture}

\coordinate (width) at (0.4,0);
\coordinate (height) at (0,0.4);
\coordinate (labelshift) at (0.2, 0.2);

\coordinate (1) at ($(0,0)$);
\coordinate (2) at ($(1) + 1.75*(width) - 0.75*(height)$);
\coordinate (3) at ($(1) + 1.25*(width) + 0.5*(height)$);
\coordinate (4) at ($(1) + 0.5*(width) - 1.25*(height)$);
\coordinate (5) at ($(4) - (0,0.75)$);

\node (label) at ($(5) + (labelshift)$)  {$(e)$};

\coordinate (temp) at (1);
\draw[black] ($(temp)$) -- ($(temp) + (width)$) -- ($(temp) + (width) + (height)$) -- ($(temp) + (height)$) -- ($(temp)$);
\node (label) at ($(temp) + (labelshift)$)  {$R_1$};

\coordinate (temp) at (2);
\draw[black] ($(temp)$) -- ($(temp) + (width)$) -- ($(temp) + (width) + (height)$) -- ($(temp) + (height)$) -- ($(temp)$);
\node (label) at ($(temp) + (labelshift)$)  {$R_2$};

\coordinate (temp) at (3);
\draw[black] ($(temp)$) -- ($(temp) + (width)$) -- ($(temp) + (width) + (height)$) -- ($(temp) + (height)$) -- ($(temp)$);
\node (label) at ($(temp) + (labelshift)$)  {$R_3$};

\coordinate (temp) at (4);
\draw[black] ($(temp)$) -- ($(temp) + (width)$) -- ($(temp) + (width) + (height)$) -- ($(temp) + (height)$) -- ($(temp)$);
\node (label) at ($(temp) + (labelshift)$)  {$R_4$};

\end{tikzpicture}

\caption{Illustrations for the proof of Lemma~\ref{K4Lemma}.}
\label{K4LemmaFigureApp}
\end{center}
\end{figure}
\end{proof}

\subsection*{Proof of Lemma~\ref{between_all}}

Before we can prove Lemma~\ref{between_all}, we have to make a few more assumptions about the structure of the formula. Without loss of generality (if necessary with additional satisfiable clauses over new variables), we can assume that the clauses $c_1,\dots,c_m$ are ordered such that for each $i$, the indices $prev(y_{i,h}):=\max\{-6,\sup\{j<i\mid c_j \textrm{ contains } y_{i,h} \textrm{ as literal}\}\}$ and $succ(y_{i,h}):=\min\{2m+6,\inf\{j>i\mid c_j \textrm{ contains } y_{i,h} \textrm{ as literal }\}\}$ for $h=1,2,3$ differ from $i$ and from each other for different values of $h$ by at least six.\par
Next, we first have to prove the following lemma:

\begin{lemma}\label{aligned_path} 
For all $1\leq i\leq 2m$ and $r\in\{1,2,3\}$ and  $z\in N(c_i)$, there is no path between $l_i^r$ and $z$ which does not include $c_i$ such that the associated unit squares are vertically aligned and such that there exists no other unit square which is strictly between two unit squares of this path with strictly smaller $x$-coordinate. 
\end{lemma}
\begin{proof}
Assume that there is a path $P$ between $l_i^r$ and $z$ which does not include $c_i$ such that the associated unit squares are vertically aligned and such that there exists no other unit square which is strictly between two unit squares of this path with strictly smaller $x$-coordinate. Regardless whether $z = l_i^t$ for some $t\in \{1,2,3\}\setminus\{r\}$ or $z\in C_{i} \setminus \{c_i\}$, $P$ has to start with some neighbour of  $l_i^r$. Let in the following, without loss of generality, $y_{i,r}=x_j$ and let $\{h_{t_j}^1,h_{t_j}^2,l^{r_1}_{j_1}, l^{r_2}_{j_2}, \ldots, l^{r_{q}}_{j_{q}},h_{t_j}^0,t_j,h_{t_j}^3,h_{t_j}^4\}$ with $j_1 < j_2 < \ldots < j_q$  be the neighbourhood of $\overset{_{\leftarrow}}{t_j}$ which, by definition, builds a path in this order.

%%%%%%%%%%%%%%%%%%%%%%%%%%%%%%%%%%%%%%%%%%%%%%%%%%%%%%%%%%%%%%%%%%%%
First, observe the following:\\
{\it Claim:} No square for a vertex in  $S\setminus\{\overset{_{\leftarrow}}{t_j},\overset{_{\rightarrow}}{t_j},h_{t_j}^0,\dots,h_{t_j}^4\}$ can be vertically aligned with one of its neighbours from $S$ such that there is no unit square with smaller $x$-coordinate strictly between them. To see this, observe generally that for two vertically aligned squares $R_{u},
R_{v}$  there are only the two layouts from Figure~\ref{previously-wrong-claim} to place squares for two common neighbours $s_1,s_2$ of $u$ and $w$. In Figure~\ref{previously-wrong-claim}$(a)$ there is no possibility to avoid placing either $R_{s_1}$ or $R_{s_2}$ strictly between  $R_{u}$ and $R_{v}$  with smaller $x$-coordinate. In Figure~\ref{previously-wrong-claim}$(b)$, there is no possibility to place another square $R_{s_3}$ which sees both  $R_{s_1}$ and $R_{s_2}$. With some vertex from $S\setminus\{\overset{_{\leftarrow}}{t_j},\overset{_{\rightarrow}}{t_j},h_{t_j}^0,\dots,h_{t_j}^4\}$  in the role of $u$, and $w\in N(u)\cap S$, there are always vertices with the properties of $s_1,s_2,s_3$. If $w$ is in $S\setminus \{\overset{_{\leftarrow}}{t_j},\overset{_{\rightarrow}}{t_j}\}$ the vertices $\overset{_{\leftarrow}}{t_j}$ and $ \overset{_{\rightarrow}}{t_j}$ are the common neighbours $s_1$ and $s_2$ while $h_{t_j}^1$ or $h_{t_j}^4$ can be considered as their common neighbour $s_3$. If $w\in \{\overset{_{\leftarrow}}{t_j},\overset{_{\rightarrow}}{t_j}\}$, the two vertices in $N(u)\setminus \{\overset{_{\leftarrow}}{t_j},\overset{_{\rightarrow}}{t_j}\}$ (observe that with the additional vertices $h_{t_j}^0,\dots,h_{t_j}^4$, $u$ has exactly two neighbours in $S$ other than $\overset{_{\leftarrow}}{t_j}$ and $\overset{_{\rightarrow}}{t_j}$) are the common neighbours $s_1$ and $s_2$ of $u$ and $w$, while the vertex in $\{\overset{_{\leftarrow}}{t_j},\overset{_{\rightarrow}}{t_j}\}\setminus \{w\}$ is their common neighbour $s_3$.

\begin{figure}
\begin{tabular}{ p{3cm} p{3cm}}
\centering
\begin{tikzpicture}[scale=0.7,transform shape]
\coordinate (width) at (0.8,0);
\coordinate (height) at (0,0.8);
\coordinate (labelshift) at (0.4, 0.4);

\coordinate (1) at ($(0,0)$);
\coordinate (2) at ($(1) + 2.6*(height)$);
\coordinate (3) at ($(1) + 0.8*(width) + 1.2*(height)$);
\coordinate (4) at ($(1) - 0.8*(width) + 1.4*(height)$);

\coordinate (temp) at (1);
\draw[black] ($(temp)$) -- ($(temp) + (width)$) -- ($(temp) + (width) + (height)$) -- ($(temp) + (height)$) -- ($(temp)$);
\node (label) at ($(temp) + (labelshift)$)  {$u$};

\coordinate (temp) at (2);
\draw[black] ($(temp)$) -- ($(temp) + (width)$) -- ($(temp) + (width) + (height)$) -- ($(temp) + (height)$) -- ($(temp)$);
\node (label) at ($(temp) + (labelshift)$)  {$w$};

\coordinate (temp) at (3);
\draw[black] ($(temp)$) -- ($(temp) + (width)$) -- ($(temp) + (width) + (height)$) -- ($(temp) + (height)$) -- ($(temp)$);
\node (label) at ($(temp) + (labelshift)$)  {${s_1}$};

\coordinate (temp) at (4);
\draw[black] ($(temp)$) -- ($(temp) + (width)$) -- ($(temp) + (width) + (height)$) -- ($(temp) + (height)$) -- ($(temp)$);
\node (label) at ($(temp) + (labelshift)$)  {${s_2}$};
\end{tikzpicture}&\centering
\begin{tikzpicture}[scale=0.7,transform shape]
\coordinate (width) at (0.8,0);
\coordinate (height) at (0,0.8);
\coordinate (labelshift) at (0.4, 0.4);

\coordinate (1) at ($(0,0)$);
\coordinate (2) at ($(1) + 1.2*(height)$);
\coordinate (3) at ($(1) - 1.3*(width) + 0.5*(height)$);
\coordinate (4) at ($(1) +  1.3*(width) + 0.5*(height)$);

\coordinate (temp) at (1);
\draw[black] ($(temp)$) -- ($(temp) + (width)$) -- ($(temp) + (width) + (height)$) -- ($(temp) + (height)$) -- ($(temp)$);
\node (label) at ($(temp) + (labelshift)$)  {$u$};

\coordinate (temp) at (2);
\draw[black] ($(temp)$) -- ($(temp) + (width)$) -- ($(temp) + (width) + (height)$) -- ($(temp) + (height)$) -- ($(temp)$);
\node (label) at ($(temp) + (labelshift)$)  {$w$};

\coordinate (temp) at (3);
\draw[black] ($(temp)$) -- ($(temp) + (width)$) -- ($(temp) + (width) + (height)$) -- ($(temp) + (height)$) -- ($(temp)$);
\node (label) at ($(temp) + (labelshift)$)  {$s_1$};

\coordinate (temp) at (4);
\draw[black] ($(temp)$) -- ($(temp) + (width)$) -- ($(temp) + (width) + (height)$) -- ($(temp) + (height)$) -- ($(temp)$);
\node (label) at ($(temp) + (labelshift)$)  {$s_2$};
\end{tikzpicture}\cr
\centering $(a)$&\centering$(b)$\cr
\end{tabular}\caption{Illustrations for the claim that $l^{r_x}_{j_x},l^{r_y}_{j_y}\in S$ cannot be aligned}\label{previously-wrong-claim}
\end{figure}
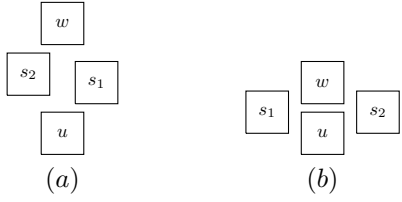

Since $y_{i,r}=x_j$, we can conclude that $l_{j_q}^{r_q}=l_i^r$, for some $p\in \{1,\dots,q\}$. In order to reach $z\in N(c_i)$, $P$ has to contain at least one vertex from  $S\setminus\{\overset{_{\leftarrow}}{t_j},\overset{_{\rightarrow}}{t_j},h_{t_j}^0,\dots,h_{t_j}^4\}$. By the claim above, it follows that $P$ has to start with either $l_{j_{p-1}}^{r_{p-1}}$ which then has to be followed by $c_{j_{p-1}}$ (if $p>1$) or $l_{j_{p+1}}^{r_{p+1}}$ which then has to be followed by $c_{j_{p+1}}$ (if $p<q$). Consider, without loss of generality, that $p<q$ and that $P$  contains $l_{j_{p+1}}^{r_{p+1}}$ followed by $c_{j_{p+1}}$. With the above definition of $succ$, we know that $j_{p+1}=succ(y_{i,r})$. By the previously assumed properties of the input-formula, we know that $j_{p+1}$ differs from $i$  and also from the smallest index $k>i$ for which $c_i$ shares a literal other than $y_{i,r}$ with $c_k$ by at least six; especially $c_i$ and  $c_{j_{p+1}}$ share no common literal other than $x_j = y_{i, r}$. This means that $P$  has to continue from  $c_{j_{p+1}}$ with at least five vertices from $V_c\setminus\{l_s^t\mid 1\leq s\leq 2m, 1\leq t\leq 3\}$; observe that all other paths to $z$ contain at least two vertices from some  path $S'$ with at least one of them from $S'\setminus\{\overset{_{\leftarrow}}{t'_j},\overset{_{\rightarrow}}{t'_j},h_{t'_j}^0,\dots,h_{t'_j}^4\}$  which is excluded by the claim above. 
%%%%%%%%%%%%%%%%%%%%%%%%%%%%%%%%%%%%%%%%%%%%%%%%%%%%%%%%%%%%%%%%%%%%

Let $s=j_{p+1}$. Assume that the vertex following $c_s$ on $P$ is some $v\in C^l_s \setminus \{c_s\}$ (the case $v\in C^r_s \setminus \{c_s\}$ is analogous and these are the only non-literal neighbour-vertices of~$c_s$). The vertices $C^l_{s}$ build a $K_4$ and, by Lemma~\ref{K4Lemma}, the only possibility for a layout of this $K_4$ such that there is no unit square with smaller $x$-coordinate strictly between the vertically aligned $R_{v}$ and $R_{c_s}$, is case 1 from Figure~\ref{visLayoutsCompleteBipartiteGraphsFigureApp} with $R_{v}$ and $R_{c_s}$ taking the role of $R_2$ and~$R_4$. Observe that in case 1 there is no possibility to add a unit square which sees both $R_1$ and $R_3$ without destroying any of the $K_4$ edges. The only possibility for $v$ is hence $c_{s-1}$, since $c_{s}$ and $c_{s-1}^t$  have the common neighbour $c_{s}^t$ for $t=1,2$ which could not be placed otherwise. By the same argument, the whole part of at least five vertices from $V_c\setminus\{l_s^t\mid 1\leq s\leq 2m, 1\leq t\leq 3\}$ in $P$ are in fact vertices in $\{c_j\mid 1\leq j\leq 2m\}$ and especially contain the sequence $c_s,c_{s-1},c_{s-2},c_{s-3},c_{s-4}$ which has to be arranged as illustrated below.

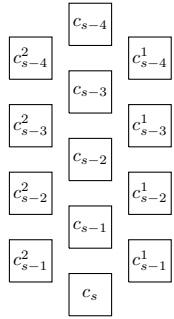
\begin{center}\begin{tikzpicture}[scale=0.7, transform shape]
\coordinate (height) at (0.8,0);
\coordinate (width) at (0,0.8);
\coordinate (labelshift) at (0.4, 0.4);

\coordinate (1) at ($(0,0)$);
\coordinate (2) at ($(1) + 1.6*(width)$);
\coordinate (3) at ($(1) + 1.4*(height) + 0.8*(width)$);
\coordinate (4) at ($(1) - 1.4*(height) + 0.8*(width)$);
\coordinate (5) at ($(2)+ 1.6*(width)$);
\coordinate (6) at ($(5) + 1.6*(width)$);
\coordinate (7) at ($(5) + 1.4*(height) + 0.8*(width)$);
\coordinate (8) at ($(5) - 1.4*(height) + 0.8*(width)$);
\coordinate (9) at ($(6) + 1.6*(width)$);
\coordinate (10) at ($(6) + 1.4*(height) + 0.8*(width)$);
\coordinate (11) at ($(6) - 1.4*(height) + 0.8*(width)$);
\coordinate (12) at ($(2) + 1.4*(height) + 0.8*(width)$);
\coordinate (13) at ($(2) - 1.4*(height) + 0.8*(width)$);

\coordinate (temp) at (1);
\draw[black] ($(temp)$) -- ($(temp) + (width)$) -- ($(temp) + (width) + (height)$) -- ($(temp) + (height)$) -- ($(temp)$);
\node (label) at ($(temp) + (labelshift)$)  {$c_{s}$};

\coordinate (temp) at (2);
\draw[black] ($(temp)$) -- ($(temp) + (width)$) -- ($(temp) + (width) + (height)$) -- ($(temp) + (height)$) -- ($(temp)$);
\node (label) at ($(temp) + (labelshift)$)  {$c_{s-1}$};

\coordinate (temp) at (3);
\draw[black] ($(temp)$) -- ($(temp) + (width)$) -- ($(temp) + (width) + (height)$) -- ($(temp) + (height)$) -- ($(temp)$);
\node (label) at ($(temp) + (labelshift)$)  {$c_{s-1}^1$};

\coordinate (temp) at (4);
\draw[black] ($(temp)$) -- ($(temp) + (width)$) -- ($(temp) + (width) + (height)$) -- ($(temp) + (height)$) -- ($(temp)$);
\node (label) at ($(temp) + (labelshift)$)  {$c_{s-1}^2$};

\coordinate (temp) at (12);
\draw[black] ($(temp)$) -- ($(temp) + (width)$) -- ($(temp) + (width) + (height)$) -- ($(temp) + (height)$) -- ($(temp)$);
\node (label) at ($(temp) + (labelshift)$)  {$c_{s-2}^1$};

\coordinate (temp) at (13);
\draw[black] ($(temp)$) -- ($(temp) + (width)$) -- ($(temp) + (width) + (height)$) -- ($(temp) + (height)$) -- ($(temp)$);
\node (label) at ($(temp) + (labelshift)$)  {$c_{s-2}^2$};

\coordinate (temp) at (5);
\draw[black] ($(temp)$) -- ($(temp) + (width)$) -- ($(temp) + (width) + (height)$) -- ($(temp) + (height)$) -- ($(temp)$);
\node (label) at ($(temp) + (labelshift)$)  {$c_{s-2}$};

\coordinate (temp) at (6);
\draw[black] ($(temp)$) -- ($(temp) + (width)$) -- ($(temp) + (width) + (height)$) -- ($(temp) + (height)$) -- ($(temp)$);
\node (label) at ($(temp) + (labelshift)$)  {$c_{s-3}$};

\coordinate (temp) at (7);
\draw[black] ($(temp)$) -- ($(temp) + (width)$) -- ($(temp) + (width) + (height)$) -- ($(temp) + (height)$) -- ($(temp)$);
\node (label) at ($(temp) + (labelshift)$)  {$c_{s-3}^1$};

\coordinate (temp) at (8);
\draw[black] ($(temp)$) -- ($(temp) + (width)$) -- ($(temp) + (width) + (height)$) -- ($(temp) + (height)$) -- ($(temp)$);
\node (label) at ($(temp) + (labelshift)$)  {$c_{s-3}^2$};

\coordinate (temp) at (9);
\draw[black] ($(temp)$) -- ($(temp) + (width)$) -- ($(temp) + (width) + (height)$) -- ($(temp) + (height)$) -- ($(temp)$);
\node (label) at ($(temp) + (labelshift)$)  {$c_{s-4}$};

\coordinate (temp) at (10);
\draw[black] ($(temp)$) -- ($(temp) + (width)$) -- ($(temp) + (width) + (height)$) -- ($(temp) + (height)$) -- ($(temp)$);
\node (label) at ($(temp) + (labelshift)$)  {$c_{s-4}^1$};

\coordinate (temp) at (11);
\draw[black] ($(temp)$) -- ($(temp) + (width)$) -- ($(temp) + (width) + (height)$) -- ($(temp) + (height)$) -- ($(temp)$);
\node (label) at ($(temp) + (labelshift)$)  {$c_{s-4}^2$};

\end{tikzpicture}
\end{center}

There is no possibility to enable visibility representing the edge $\{c_{s-2},l_{s-2}^1\}$ such that $R_{l_{s-2}^1}$ sees none of the unit squares $R_{c_{s-2}^1},R_{c_{s-2}^2},R_{c_{s-3}^1},R_{c_{s-3}^2}$, for these unit squares however, all neighbours have been placed, so there is no possibility to block this unwanted visibility which overall yields a contradiction to the existence of $P$. 

\end{proof}

We are now ready to give a proof for Lemma~\ref{between_all}. However, in order to make this Lemma easier applicable in the following proofs, we restate it in a way that the three different cases are labeled.

\begin{lemma}
For all $i$, $1\leq i\leq 2m$, $r\in\{1,2,3\}$ and every $R_z$ with $z\in N(c_i)\setminus\{l_i^r\}$, there exists no non-degenerate axis-parallel rectangle $S$ which is not intersected by $R_{c_i}$ such that one side of $S$ is in $R_{l_i^r}$ and the opposite side is in $R_z$. In particular, this implies the following properties:\begin{enumerate}
\item $R_z$ is not strictly between $R_{c_i}$ and $R_{l_i^r}$.
\item $R_{l_i^r}$ is not strictly between $R_{c_i}$ and $R_z$ .
\item If $R_{c_i}$ is strictly between $R_{l_i^r}$ and $R_z$ then  $R_{c_i}$ blocks the view between $R_{l_i^r}$ and $R_z$.
\end{enumerate}
\end{lemma}

\begin{proof} 
Let $z$, $i$ and $r$ be as in the statement of the lemma. We assume that there exists a non-degenerate axis-parallel rectangle $S$ which is not intersected by $R_{c_i}$ such that one side of $S$ is in $R_{l_i^r}$ and the opposite side is in $R_z$. Without loss of generality, we assume that, after removing all unit squares except $R_{l_i^r}$ and $R_{z}$ from the layout, we have $R_{l_i^r} \Vvis R_{z}$, and, furthermore, that the $x$-coordinate of $R_{l_i^r}$ is not smaller than the $x$-coordinate of $R_z$ (so we have the situation shown in Figure~\ref{betweenlemma}). Moreover, since $z$ is not adjacent to $l_i^r$, some further unit square(s) have to block the visibility (indicated by the rectangle $S$) between $R_{l_i^r}$ and $R_z$, while $R_{c_i}$ has to see both  $R_z$ and $R_{l_i^r}$. There are only the following possibilities for this situation:
\begin{enumerate}
\item $R_{c_i}\Vvis \{R_{l_i^r},R_z\}$ or $\{R_{l_i^r},R_z\}\Vvis R_{c_i}$ (see Figure~\ref{betweenlemma}$(a)$): We only consider the case $R_{c_i}\Vvis \{R_{l_i^r},R_z\}$, since $\{R_{l_i^r},R_z\}\Vvis R_{c_i}$ can be dealt with analogously. Let $R_{h_1},\dots,R_{h_s}$ be the unit squares strictly between $R_{l_i^r}$ and $R_z$ (to intersect $S$  in order to block the view) of minimum $x$-coordinate sorted by $y$-coordinate. Observe that each $R_{h_t}$ has a larger $x$-coordinate than $R_{z}$ since otherwise there is no visibility between $R_{c_i}$ and $R_z$. If there is a unit square $R$ strictly between some $R_{h_t}$ and $R_{h_{t+1}}$ with smaller $x$-coordinate, then this either contradicts the definition of the $R_{h_t}$ (i.\,e., if $R$ is strictly between $R_{l_i^r}$ and $R_z$) or, again, the visibility between $c_i$ and $z$ would be blocked. Consequently, there is no such unit square strictly between some $R_{h_t}$ and $R_{h_{t+1}}$ with smaller $x$-coordinate. The vertices $h_1,\dots,h_s$ corresponding to $R_{h_1},\dots,R_{h_s}$ hence describe a path which does not include $c_i$ and for which $h_1$ is adjacent to $z\in N(c_i)\setminus\{l_i^r\}$ and $h_s$ is adjacent to $l_i^r$. The unit squares $R_{h_1},\dots,R_{h_s}$ are, by definition, aligned and no unit square with smaller $x$-coordinate is strictly between any $R_{h_t}$ and $R_{h_{t+1}}$ which is a contradiction to Lemma~\ref{aligned_path}.
\item ($R_z\Hvis R_{c_i}$ and $R_{l_i^r}\Vvis R_{c_i}$) or ($R_{c_i}\Vvis R_z$ and $R_{c_i}\Hvis R_{l_i^r}$) (see Figure~\ref{betweenlemma}$(b)$): We only consider the case ($R_z\Hvis R_{c_i}$ and $R_{l_i^r}\Vvis R_{c_i}$), since ($R_{c_i}\Vvis R_z$ and $R_{c_i}\Hvis R_{l_i^r}$) can be dealt with analogously. To preserve the visibility $R_{l_i^r}\Vvis R_{c_i}$, all unit squares which intersect $S$ have to be strictly left of $R_{l_i^r}$. Let again $R_{h_1},\dots,R_{h_s}$ be the unit squares strictly between $R_{l_i^r}$ and $R_z$ (to intersect $S$ in order to block the view) of maximum $x$-coordinate. These unit squares have the same properties as for case $1$, which yields a contradiction to Lemma~\ref{aligned_path}.
\item $R_{c_i}\Vvis R_z$ and $R_{l_i^r}\Vvis R_{c_i}$ (see Figure~\ref{betweenlemma}$(c)$ and $(d)$): Note that $S$ is not intersected by $R_{c_i}$ and, without loss of generality, we assume that $S$ lies to the left of $R_{c_i}$. If there exists a unit square that intersects $S$ and has a larger $x$-coordinate than $R_z$ or $R_{l_i^r}$, then there is no visibility between $R_{c_i}$ and $R_z$ or $R_{l_i^r}$ and $R_{c_j}$. Consequently, all unit squares which intersect $S$ have to be strictly to the left of $R_z$ and $R_{l_i^r}$. Among all unit squares which intersect $S$, let $R_{h_1},\dots,R_{h_s}$ be the ones of maximum $x$-coordinate, sorted by $y$-coordinate. We consider two different cases according two whether one of the unit squares $R_{h_1},\dots,R_{h_s}$ sees $R_{c_i}$ vertically or not:
\begin{enumerate}
\item %Case 2.1: 
We assume that there is no $j$, $1\leq j\leq s$, such that $R_{h_j}$ sees $R_{c_i}$ vertically (see Figure~\ref{betweenlemma}(c)). This implies that the vertices associated to $R_{h_1},\dots,R_{h_s}$ build a path from $z$ to $l_i^r$ and $c_i$ is not included in this path (observe that although one or even two of the vertices $h_1,\dots,h_s$ could be neighbours of $c_j$, vertex $c_j$ is not among the vertices $h_1,\dots,h_s$). Further, since the $x$-coordinate of the unit squares $R_{h_1},\dots,R_{h_s}$ is assumed to be maximum among the unit squares which intersect $S$, there is no other unit square which lies strictly between some $R_{h_t}$ and $R_{h_{t+1}}$ and has a larger $x$-coordinate. Since $R_{h_1},\dots,R_{h_s}$ are vertically aligned by definition, this is a contradiction to Lemma~\ref{aligned_path}.
\item %Case 2.2: 
We assume that, for some $j$, $1\leq j\leq s$, the unit square $R_{h_j}$ sees $R_{c_i}$ vertically (see Figure~\ref{betweenlemma}(d)). This is only possible if this $R_{h_j}$ is between $R_{c_i}$ and either $R_{l_i^r}$ or $R_z$. However, if $R_{h_j}$ is between $R_{c_i}$ and $R_{l_i^r}$, then, with $R_{h_j}$ playing the role of $z$ (note that $h_j \in N(c_i)\setminus\{l_i^r\}$), we obtain case $1$ again. Thus, we can assume that $R_{h_j}$ is between $R_{c_i}$ and $R_z$. Furthermore, if $j < s$, then $R_{h_{j+1}}$ is between $R_{c_i}$ and $R_{l_i^r}$ and we obtain case $1$ as before. Consequently, $j=s$ and $h_s\in N(c_i)$. Since $l_i^r$ is not adjacent to any neighbour of $c_i$, $l_i^r$ and $h_s$ are not adjacent; thus, the visibility between $R_{h_s}$ and $R_{l_i^r}$ must be blocked. Let $R_{k_1},\dots, R_{k_t}$ be the unit squares of maximum $x$-coordinate which intersect every visibility rectangle between $R_{h_s}$ and $R_{l_i^r}$, sorted by $y$-coordinate. We note that if they are strictly between $R_{l^r_i}$ and $R_{c_i}$, then we obtain case $1$ again, with $R_{k_1}$ playing the role of $z$. Moreover, if they have an $x$-coordinate that is more than one less than the $x$-coordinate of $R_{c_i}$, then they would not block the visibility between $R_{h_s}$ and $R_{l_i^r}$. Consequently, these unit squares all have the same $x$-coordinate which is exactly one less than the $x$-coordinate of $c_i$ (as shown in Figure~\ref{betweenlemma}(d)). The vertices associated to  $R_{k_1},\dots, R_{k_t}$ again build a path between $l_i^r$ and some neighbour of $c_i$ and do not include $c_i$, and, as explained above, there is no unit square strictly between some $R_{k_t}$ and $R_{k_{t+1}}$ of larger $x$-coordinate. Hence, the path $R_{k_1},\dots, R_{k_t}$ is also a contradiction to Lemma~\ref{aligned_path}.
\end{enumerate}
\end{enumerate}
\begin{center}
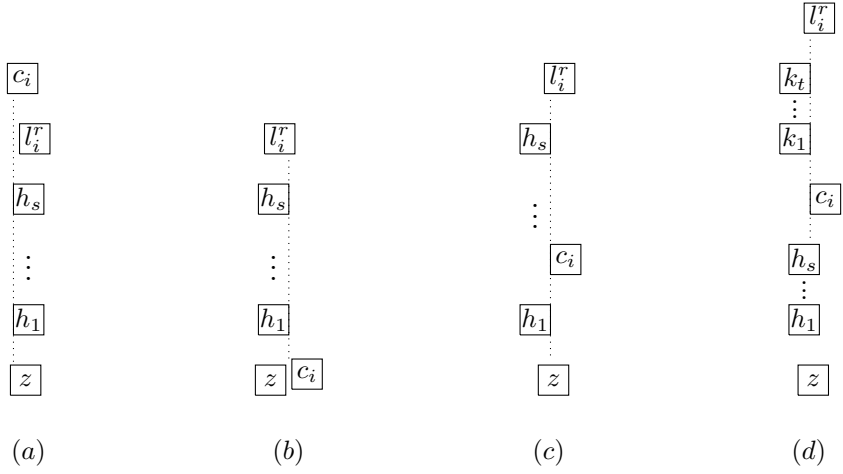
\begin{figure}
\begin{tabular}{ p{3cm} p{3cm} p{3cm} p{3cm}}
\centering \begin{tikzpicture}
\coordinate (width) at (0.4,0);
\coordinate (height) at (0,0.4);
\coordinate (labelshift) at (0.2, 0.2);

\coordinate (1) at ($(0,0)$);
\coordinate (2) at ($(1) + 0.3*(width) + 8*(height)$);
\coordinate (3) at ($(1) - 0.1*(width) + 10*(height)$);
\coordinate (4) at ($(1) + 0.1*(width) + 2*(height)$);
\coordinate (5) at ($(1) + 0.1*(width) + 4*(height)$);
\coordinate (6) at ($(1) + 0.1*(width) + 6*(height)$);
\coordinate (7) at ($(4) + 0.5*(width) + 2.5*(height)$);
\coordinate (8) at ($(4) + 8.1*(height) $);
\coordinate (9) at ($(4) - 1.2*(height) $);

\coordinate (temp) at (1);
\draw[black] ($(temp)$) -- ($(temp) + (width)$) -- ($(temp) + (width) + (height)$) -- ($(temp) + (height)$) -- ($(temp)$);
\node (label) at ($(temp) + (labelshift)$)  {$z$};

\coordinate (temp) at (2);
\draw[black] ($(temp)$) -- ($(temp) + (width)$) -- ($(temp) + (width) + (height)$) -- ($(temp) + (height)$) -- ($(temp)$);
\node (label) at ($(temp) + (labelshift)$)  {$l^r_i$};

\coordinate (temp) at (3);
\draw[black] ($(temp)$) -- ($(temp) + (width)$) -- ($(temp) + (width) + (height)$) -- ($(temp) + (height)$) -- ($(temp)$);
\node (label) at ($(temp) + (labelshift)$)  {$c_i$};

\coordinate (temp) at (4);
\draw[black] ($(temp)$) -- ($(temp) + (width)$) -- ($(temp) + (width) + (height)$) -- ($(temp) + (height)$) -- ($(temp)$);
\node (label) at ($(temp) + (labelshift)$)  {$h_1$};

\coordinate (temp) at (6);
\draw[black] ($(temp)$) -- ($(temp) + (width)$) -- ($(temp) + (width) + (height)$) -- ($(temp) + (height)$) -- ($(temp)$);
\node (label) at ($(temp) + (labelshift)$)  {$h_s$};

\node[draw=none]  at (7) {$\large\vdots$};
\node[draw=none] (u) at (8) {};
\node[draw=none] (l) at (9) {};
\draw[dotted] (u)--(l);
\end{tikzpicture}& \centering \begin{tikzpicture}
\coordinate (width) at (0.4,0);
\coordinate (height) at (0,0.4);
\coordinate (labelshift) at (0.2, 0.2);

\coordinate (1) at ($(0,0)$);
\coordinate (2) at ($(1) + 0.3*(width) + 8*(height)$);
\coordinate (3) at ($(1) + 1.2*(width) + 0.2*(height)$);
\coordinate (4) at ($(1) + 0.1*(width) + 2*(height)$);
\coordinate (5) at ($(1) + 0.1*(width) + 4*(height)$);
\coordinate (6) at ($(1) + 0.1*(width) + 6*(height)$);
\coordinate (7) at ($(4) + 0.5*(width) + 2.5*(height)$);
\coordinate (8) at ($(4) + 6.1*(height)+(width) $);
\coordinate (9) at ($(4) - 1.2*(height)+(width) $);

\coordinate (temp) at (1);
\draw[black] ($(temp)$) -- ($(temp) + (width)$) -- ($(temp) + (width) + (height)$) -- ($(temp) + (height)$) -- ($(temp)$);
\node (label) at ($(temp) + (labelshift)$)  {$z$};

\coordinate (temp) at (2);
\draw[black] ($(temp)$) -- ($(temp) + (width)$) -- ($(temp) + (width) + (height)$) -- ($(temp) + (height)$) -- ($(temp)$);
\node (label) at ($(temp) + (labelshift)$)  {$l^r_i$};

\coordinate (temp) at (3);
\draw[black] ($(temp)$) -- ($(temp) + (width)$) -- ($(temp) + (width) + (height)$) -- ($(temp) + (height)$) -- ($(temp)$);
\node (label) at ($(temp) + (labelshift)$)  {$c_i$};

\coordinate (temp) at (4);
\draw[black] ($(temp)$) -- ($(temp) + (width)$) -- ($(temp) + (width) + (height)$) -- ($(temp) + (height)$) -- ($(temp)$);
\node (label) at ($(temp) + (labelshift)$)  {$h_1$};

\coordinate (temp) at (6);
\draw[black] ($(temp)$) -- ($(temp) + (width)$) -- ($(temp) + (width) + (height)$) -- ($(temp) + (height)$) -- ($(temp)$);
\node (label) at ($(temp) + (labelshift)$)  {$h_s$};

\node[draw=none]  at (7) {$\large\vdots$};
\node[draw=none] (u) at (8) {};
\node[draw=none] (l) at (9) {};
\draw[dotted] (u)--(l);
\end{tikzpicture}&\centering \begin{tikzpicture}
\coordinate (width) at (0.4,0);
\coordinate (height) at (0,0.4);
\coordinate (labelshift) at (0.2, 0.2);

\coordinate (1) at ($(0,0)$);
\coordinate (2) at ($(1) + 0.2*(width) + 10*(height)$);
\coordinate (3) at ($(1) + 0.4*(width) + 4*(height)$);
\coordinate (4) at ($(1) - 0.6*(width) + 2*(height)$);
\coordinate (5) at ($(1) - 0.6*(width) + 6*(height)$);
\coordinate (6) at ($(1) - 0.6*(width) + 8*(height)$);
\coordinate (7) at ($(5) + 0.5*(width) + 0.2*(height)$);
\coordinate (8) at ($(3) + 6.1*(height)$);
\coordinate (9) at ($(3) - 3.1*(height)$);

\coordinate (temp) at (1);
\draw[black] ($(temp)$) -- ($(temp) + (width)$) -- ($(temp) + (width) + (height)$) -- ($(temp) + (height)$) -- ($(temp)$);
\node (label) at ($(temp) + (labelshift)$)  {$z$};

\coordinate (temp) at (2);
\draw[black] ($(temp)$) -- ($(temp) + (width)$) -- ($(temp) + (width) + (height)$) -- ($(temp) + (height)$) -- ($(temp)$);
\node (label) at ($(temp) + (labelshift)$)  {$l^r_i$};

\coordinate (temp) at (3);
\draw[black] ($(temp)$) -- ($(temp) + (width)$) -- ($(temp) + (width) + (height)$) -- ($(temp) + (height)$) -- ($(temp)$);
\node (label) at ($(temp) + (labelshift)$)  {$c_i$};

\coordinate (temp) at (4);
\draw[black] ($(temp)$) -- ($(temp) + (width)$) -- ($(temp) + (width) + (height)$) -- ($(temp) + (height)$) -- ($(temp)$);
\node (label) at ($(temp) + (labelshift)$)  {$h_1$};

\coordinate (temp) at (6);
\draw[black] ($(temp)$) -- ($(temp) + (width)$) -- ($(temp) + (width) + (height)$) -- ($(temp) + (height)$) -- ($(temp)$);
\node (label) at ($(temp) + (labelshift)$)  {$h_s$};

\node[draw=none]  at (7) {$\large\vdots$};
\node[draw=none] (u) at (8) {};
\node[draw=none] (l) at (9) {};
\draw[dotted] (u)--(l);
\end{tikzpicture}
& \centering 
\begin{tikzpicture}

\coordinate (width) at (0.4,0);
\coordinate (height) at (0,0.4);
\coordinate (labelshift) at (0.2, 0.2);

\coordinate (1) at ($(0,0)$);
\coordinate (2) at ($(1) + 0.2*(width) + 10*(height)$);
\coordinate (3) at ($(1) + 0.4*(width) + 4*(height)$);
\coordinate (4) at ($(1) - 0.3*(width) + 2*(height)$);
\coordinate (5) at ($(1) - 0.6*(width) + 6*(height)$);
\coordinate (6) at ($(1) - 0.6*(width) + 8*(height)$);
\coordinate (7) at ($(6) + 0.5*(width) - 0.4*(height)$);
\coordinate (8) at ($(3) + 6.1*(height)$);
\coordinate (9) at ($(3) - 1.1*(height)$);
\coordinate (10) at ($(1) - 2*(height)$);
\coordinate (11) at ($(1) - 0.3*(width) $);
\coordinate (12) at ($(4) + 0.5*(width) - 0.4*(height)$);

\coordinate (temp) at (10);
\draw[black] ($(temp)$) -- ($(temp) + (width)$) -- ($(temp) + (width) + (height)$) -- ($(temp) + (height)$) -- ($(temp)$);
\node (label) at ($(temp) + (labelshift)$)  {$z$};

\coordinate (temp) at (2);
\draw[black] ($(temp)$) -- ($(temp) + (width)$) -- ($(temp) + (width) + (height)$) -- ($(temp) + (height)$) -- ($(temp)$);
\node (label) at ($(temp) + (labelshift)$)  {$l^r_i$};

\coordinate (temp) at (3);
\draw[black] ($(temp)$) -- ($(temp) + (width)$) -- ($(temp) + (width) + (height)$) -- ($(temp) + (height)$) -- ($(temp)$);
\node (label) at ($(temp) + (labelshift)$)  {$c_i$};

\coordinate (temp) at (4);
\draw[black] ($(temp)$) -- ($(temp) + (width)$) -- ($(temp) + (width) + (height)$) -- ($(temp) + (height)$) -- ($(temp)$);
\node (label) at ($(temp) + (labelshift)$)  {$h_s$};
\coordinate (temp) at (11);
\draw[black] ($(temp)$) -- ($(temp) + (width)$) -- ($(temp) + (width) + (height)$) -- ($(temp) + (height)$) -- ($(temp)$);
\node (label) at ($(temp) + (labelshift)$)  {$h_1$};

\coordinate (temp) at (5);
\draw[black] ($(temp)$) -- ($(temp) + (width)$) -- ($(temp) + (width) + (height)$) -- ($(temp) + (height)$) -- ($(temp)$);
\node (label) at ($(temp) + (labelshift)$)  {$k_1$};

\coordinate (temp) at (6);
\draw[black] ($(temp)$) -- ($(temp) + (width)$) -- ($(temp) + (width) + (height)$) -- ($(temp) + (height)$) -- ($(temp)$);
\node (label) at ($(temp) + (labelshift)$)  {$k_t$};

\node[draw=none]  at ($(7)-(0,0.15)$) {.};
\node[draw=none]  at ($(7)-(0,0.05)$) {.};
\node[draw=none]  at ($(7)+(0,0.05)$) {.};
\node[draw=none]  at ($(12)-(0,0.15)$) {.};
\node[draw=none]  at ($(12)-(0,0.05)$) {.};
\node[draw=none]  at ($(12)+(0,0.05)$) {.};
\node[draw=none] (u) at (8) {};
\node[draw=none] (l) at (9) {};
\draw[dotted] (u)--(l);

\end{tikzpicture}
\cr
&&&\cr
\centering$(a)$&\centering $(b)$&\centering $(c)$&\centering $(d)$

\end{tabular}
\caption{Illustrations for the proof of Lemma~\ref{between_all}.}
\label{betweenlemma}
\end{figure}
\end{center}
Since Lemma~\ref{aligned_path} holds equivalently, the same argumentation yields this result for  $t_j$ or $f_j^1$ or $f_j^2$ instead of  $l_i^r$ and $x_j$ instead of $c_i$ for all $1\leq j\leq n$.
\end{proof}

\subsection*{Proof of Lemma~\ref{notAllOneSideLemma}}

%\begin{lemma}
%For every $j$, $1 \leq j \leq 2m - 1$, and every $y \in C_j \setminus \{c_j\}$, $R_{c_j} \Hvis R_{C_j \setminus \{y, c_j\}}$ is not possible.
%\end{lemma}

\begin{proof}
We first note that, independent from the choice of $y$, either $C^l_j$ or $C^r_j$ is completely contained in $R_{C_j \setminus \{y\}}$. We assume that the former applies, which means that the $K_4$ on vertices $C^l_j$ must satisfy case $1$ of Lemma~\ref{K4Lemma} with $c_j$ playing the role of $R_1$, or it satisfies case $3$ of Lemma~\ref{K4Lemma} with $c_j$ playing the role of $R_4$. Next, we note that $R_y \Hvis R_{c_j}$ is not possible, since then the $K_4$ on vertices $C^r_j$ contains a unit square, namely $R_{c_j}$, which horizontally sees all other vertices, but not in the same direction and this is, according to Lemma~\ref{K4Lemma}, not possible; thus, we have either case $R_y \symVvis R_{c_j}$ or case $R_{c_j} \Hvis R_y$, which we shall now consider separately. \par
\begin{enumerate}
\item Case $R_y \Vvis R_{c_j}$ (the case $R_{c_j} \Vvis R_y $ can be handled analogously): For the $K_4$ on vertices $C^r_{j}$, we have that $R_y \Vvis R_{c_j}$, while all unit squares $R_{C^r_{j} \setminus \{c_j, y\}}$ see $R_{c_j}$ horizontally to the same side. This means that the $K_4$ on vertices $C^r_{j}$ satisfies case $2$ or case $3$ of Lemma~\ref{K4Lemma}. If it satisfies case $2$, then we have the situation illustrated in Figure~\ref{notAllOneSideLemmaFigure}$(a)$ (where the four vertices on the left are the vertices from $C^r_{j}$). On the other hand, if it satisfies case $3$, then we claim that the only unit square that can play the role of the unit square $R_4$ (i.\,e., the one that sees all the others by the same kind of visibility) is $R_{y}$. In order to verify this claim, we first observe that $R_{c_j}$ cannot play the role of $R_4$, since it sees $R_{y}$ vertically and the two other unit squares in $R_{C^r_{j} \setminus \{c_j, y\}}$ horizontally. If, for a $z \in C^r_j \setminus \{y, c_j\}$, $R_z$ plays the role of $R_4$, then, since $R_{c_j} \Hvis R_z$, we must have $R_{C^r_j \setminus \{z\}} \Hvis R_{z}$; in particular, $R_{z}$ plays the role of $R_4$, $R_{c_j}$ plays the role of $R_2$ and $R_{y}$ plays the role of $R_1$. Now it is not possible to add another unit square $R$ with $R_{c_j} \Hvis R$, since in order to see $R_{c_j}$, $R$ must be placed to the left of $R_{z}$, which means that it necessarily blocks the visibility between $R_{z}$ and one of $C^r_j \setminus \{c_j, z\}$.
This is a contradiction, since such a unit square must exist in order to represent the edges of the $K_4$ on vertices $C^l_{j}$. Consequently, if the $K_4$ on vertices $C^r_{j}$ satisfies case $3$, then we have the situation illustrated in Figure~\ref{notAllOneSideLemmaFigure}$(b)$.\par
We now turn to the $K_4$ on vertices $R_{C^l_{j}}$. As mentioned before, $R_{c_j} \Hvis R_{C^l_{j} \setminus \{c_j\}}$ implies that the $K_4$ on vertices $C^l_{j}$ must satisfy case $1$ of Lemma~\ref{K4Lemma} with $R_{c_{j}}$ playing the role of $R_1$ or case $3$ of Lemma~\ref{K4Lemma} with $R_{c_j}$ playing the role of $R_4$. (Note that in Figure~\ref{notAllOneSideLemmaFigure}$(a)$ and $(b)$, we only illustrate case $1$ for the $K_4$ on vertices $C^l_{j}$, which are the three vertices to the right together with vertex $c_j$; the following arguments can be carried out analogously for the case that the $K_4$ on vertices $C^l_{j}$ satisfies case $3$ of Lemma~\ref{K4Lemma}.) However, for both situations depicted in Figure~\ref{notAllOneSideLemmaFigure}$(a)$ and $(b)$, it is not possible that a unit square in $R_{C^l_{j} \setminus \{c_j\}}$ sees $R_{c_j}$ horizontally and at the same time (by any kind of visibility) also $R_{y}$. Thus, $y$ must be a vertex that is not connected to any vertex from $R_{C^l_{j} \setminus \{c_j\}}$, which implies $y = c_{j + 1}$. More precisely, since $R_{c^1_{j-1}} \in R_{C^l_{j} \setminus \{c_j\}}$ must see both $R_{c_j}$ and $R_{c^1_{j}}$, and $R_{c^2_{j-1}} \in R_{C^l_{j} \setminus \{c_j\}}$ must see both $R_{c_j}$ and $R_{c^2_{j}}$, we can conclude that $y \notin \{c^1_{j}, c^2_{j}\}$, which implies $y = c_{j + 1}$.
\begin{enumerate}
\item Situation illustrated in Figure~\ref{notAllOneSideLemmaFigure}$(b)$: We assume that $R_{c^1_{j}} \Vvis R_{c^2_{j}}$ (note that, due to the edges $\{c^1_{j}, c^1_{j-1}\}$ and $\{c^2_{j}, c^2_{j-1}\}$, this uniquely defines all the unlabelled unit squares of Figure~\ref{notAllOneSideLemmaFigure}$(b)$); the case $R_{c^2_{j}} \Vvis R_{c^1_{j}}$ can be handled analogously. We now consider the vertices $c^1_{j+1}$ and $c^2_{j+1}$ from the $K_4$ on vertices $C^l_{j + 1}$ (which we have not considered so far and which are not present in Figure~\ref{notAllOneSideLemmaFigure}$(b)$). This vertex $c^1_{j+1}$ is connected to both $c_{j + 1}$ and $c^1_{j}$ and, likewise, the vertex $c^2_{j+1}$ is connected to both $c_{j + 1}$ and $c^2_{j}$. To see both $R_{c_{j + 1}}$ and $R_{c^2_{j}}$, $R_{c^2_{j+1}}$ must be placed such that $R_{c_{j + 1} } \Vvis R_{c^2_{j+1}}$ and $R_{c^2_{j+1}}\Hvis R_{c_{j}^2}$  or $R_{c_{j}^2}\Vvis R_{c_{j+1}^2}$. For $R_{c^1_{j+1}}$, there are several possibilities to see $R_{c_{j + 1}}$ and $R_{c^1_{j}}$, but all of them would either block one of the necessary visibilities of the $K_4$ on vertices $C^r_{j}$, or would be such that $R_{c^1_{j+1}}$ and $R_{c^2_{j+1}}$ cannot see each other. More precisely, $R_{\{c_{j + 1}, c^1_{j}\}} \Vvis R_{c^1_{j + 1}}$ is clearly not possible, while $R_{c^1_{j + 1}} \Vvis R_{\{c_{j + 1}, c^1_{j}\}}$ or $R_{c^1_{j + 1}} \symHvis R_{\{c_{j + 1}, c^1_{j}\}}$ means that $R_{c^1_{j+1}}$ and $R_{c^2_{j+1}}$ cannot see each other. Moreover, also $R_{c_{j + 1}} \Hvis R_{c^1_{j + 1}}$ and $R_{c^1_{j + 1}} \Vvis R_{c^1_{j}}$ implies that $R_{c^1_{j+1}}$ and $R_{c^2_{j+1}}$ cannot see each other, while $R_{c_{j + 1}} \Vvis R_{c^1_{j + 1}}$ and $R_{c^1_{j + 1}} \Hvis R_{c^1_{j}}$ means that either the visibility between $R_{c_{j+1}}$ and $R_{c_j}$ is blocked by $R_{c^1_{j+1}}$ or again $R_{c^1_{j+1}}$ and $R_{c^2_{j+1}}$ cannot see each other. This means that the situation illustrated in Figure~\ref{notAllOneSideLemmaFigure}$(b)$ is not possible. 
\item Situation illustrated in Figure~\ref{notAllOneSideLemmaFigure}$(a)$: Similar as in the previous case, we assume that $R_{c^1_{j}} \Vvis R_{c^2_{j}}$ (again, this uniquely defines all the unlabelled unit squares of Figure~\ref{notAllOneSideLemmaFigure}$(a)$) and note that the case $R_{c^2_{j}} \Vvis R_{c^1_{j}}$ can be handled analogously. We now consider all possibilities of how the unit squares $R_{L_{j + 1}}$ for the literal vertices $l^1_{j + 1}, l^2_{j + 1}, l^3_{j + 1}$ can be placed such that they see $R_{c_{j + 1}}$. \par
If, for some $r$, $1 \leq r \leq 3$, $R_{c_{j+1}} \Vvis R_{l^r_{j+1}}$, then there is at least one unit square $R_{z}$, with $z \in N(c_{j + 1})$, strictly between $R_{c_{j+1}}$ and $R_{l^r_{j+1}}$, which according to case $1$ of Lemma~\ref{between_all}, is not possible. If, for some $r$, $1 \leq r \leq 3$, $R_{l^r_{j+1}} \Vvis R_{c_{j+1}}$, such that $R_{l^r_{j+1}}$ and $R_{c_{j+1}}$ are not aligned, then there is at least one unit square $R_{z}$, with $z \in N(c_{j + 1})$, such that $R_{c_{j + 1}}$ does not block the view between $R_{l_{j+1}^r}$ and $R_z$, which according to case $3$ of Lemma~\ref{between_all}, is not possible.
Consequently, there is at most one $r$, $1 \leq r \leq 3$, such that $R_{l^r_{j+1}} \Vvis R_{c_{j+1}}$ and, furthermore, $R_{l^r_{j+1}}$ and $R_{c_{j+1}}$ must be aligned. If, for some $r$, $1 \leq r \leq 3$, $R_{c_{j+1}} \Hvis R_{l^r_{j+1}}$, then, due to case $1$ of Lemma~\ref{between_all} and the position of $R_{c^1_j}$, we can conclude that $R_{l^r_{j+1}}$ is not aligned with $R_{c_{j+1}}$, but shifted upwards. Furthermore, again due to case $1$ of Lemma~\ref{between_all}, there is at most one such $R_{l^r_{j+1}}$ with $R_{c_{j+1}} \Hvis R_{l^r_{j+1}}$. If, for some $r$, $1 \leq r \leq 3$, $R_{l^r_{j+1}} \Hvis R_{c_{j+1}}$, such that $R_{l^r_{j+1}}$ and $R_{c_{j+1}}$ are not aligned, but $R_{l^r_{j+1}}$ is shifted downwards, $R_{c_{j + 1}}$ does not block the view between $R_{c_j^1}$ and $R_{l_{j+1}^r}$ which according to case $3$ of Lemma~\ref{between_all}, is not possible. In particular, due to case $1$ of Lemma~\ref{between_all}, this means that there is at most one $R_{l^r_{j+1}}$ with $R_{l^r_{j+1}} \Hvis R_{c_{j+1}}$, which is either aligned with $R_{c_{j+1}}$ or shifted upwards. However, we may assume that there is an $r$, $1 \leq r \leq 3$, with $R_{c_{j+1}} \Hvis R_{l^r_{j+1}}$ (which, as explained above, is shifted upwards), since otherwise not all three unit squares in $R_{L_{j + 1}}$ can be placed. Consequently, by applying case $3$ Lemma~\ref{between_all}, if there is an $R_{l^r_{j+1}}$ with $R_{l^r_{j+1}} \Hvis R_{c_{j+1}}$, then $R_{l^r_{j+1}}$ is aligned with $R_{c_{j+1}}$. We conclude that the unit squares in $R_{L_{j + 1}}$ must be placed as illustrated in Figure~\ref{notAllOneSideLemmaFigureTwo}$(a)$ (obviously, the positions of the unit squares in $R_{L_{j + 1}}$ can be switched). \par
Next, as already done in the previous case, we again consider the unit squares $R_{c^1_{j + 1}}$ and $R_{c^2_{j + 1}}$ from $C^l_{j + 1}$, which both must see $R_{c_{j + 1}}$. Due to the positions of $R_{l^1_{j + 1}}$ and $R_{l^2_{j + 1}}$, and due to cases $1$ and $2$ of Lemma~\ref{between_all}, for every $z \in \{c^1_{j + 1}, c^2_{j + 1}\}$, neither $R_{z} \Hvis R_{c_{j + 1}}$ nor $R_{z} \Vvis R_{c_{j + 1}}$ is possible. If $R_{c_{j + 1}} \Hvis R_{c^2_{j + 1}}$, then, in order to also see $R_{c^2_{j}}$, $R_{c^2_{j + 1}}$ must be placed such that $R_{c^2_{j + 1}} \Vvis R_{c^2_{j}}$; as there is no space between $R_{c^2_{j + 1}}$ and $R_{c^1_{j}}$ to add another unit square, this implies that $R_{c^2_{j + 1}} \Vvis R_{c^1_{j}}$, which is a contradiction. Consequently, $R_{c_{j + 1}} \Vvis R_{c^2_{j + 1}}$. Clearly, $R_{c_{j + 1}} \Vvis R_{c^1_{j + 1}}$ is not possible, since then visibility between $R_{c^1_{j + 1}}$ and $R_{c^1_{j}}$ is not possible. Hence, $R_{c_{j + 1}} \Hvis R_{c^1_{j + 1}}$ (recall that above we have excluded all other direction). However, now there is no visibility between $R_{c^1_{j + 1}}$ and $R_{c^2_{j + 1}}$, which is a contradiction. Consequently, the situation illustrated in Figure~\ref{notAllOneSideLemmaFigure}$(a)$ is not possible.
\end{enumerate}
\item Case $R_{c_j} \Hvis R_y$: We first note that this yields the situation illustrated in Figure~\ref{notAllOneSideLemmaFigure}$(c)$. Again, we only consider the situation where the $K_4$ on $C^l_j$ satisfies case $1$ of Lemma~\ref{K4Lemma} as the further argument does not differ for case $3$ of Lemma~\ref{K4Lemma}. In the same way as for case $1$ from above, we can conclude that $y = c_{j+1}$. We again assume $R_{c^1_{j}} \Vvis R_{c^2_{j}}$ (since $R_{c^2_{j}} \Vvis R_{c^1_{j}}$ can be handled analogously), we note that this uniquely defines all unit squares (as illustrated in Figure~\ref{notAllOneSideLemmaFigureTwo}$(b)$), and again we consider how the literal-vertices for $c_{j+1}$ can be placed in order to see $R_{c_{j + 1}}$. Note that from case $1$ of Lemma~\ref{between_all}, it follows that if, for some $r$, $1 \leq r \leq 3$, $R_{c_{j+1}} \symHvis R_{l^r_{j+1}}$, then $R_{c_{j+1}}$ and $R_{l^r_{j+1}}$ is not aligned, but shifted upwards. Moreover, by case $3$ of Lemma~\ref{between_all}, there is at most one $r$, $1 \leq r \leq 3$, with $R_{c_{j+1}} \symHvis R_{l^r_{j+1}}$. From case $3$ of Lemma~\ref{between_all} it also follows that if, for some $r$, $1 \leq r \leq 3$, $R_{l^r_{j+1}} \Vvis R_{c_{j+1}}$, then $R_{l^r_{j+1}}$ is aligned with $R_{c_{j+1}}$ or shifted to the left.
From case $1$ of Lemma~\ref{between_all} it follows that if, for some $r$, $1 \leq r \leq 3$, $R_{c_{j+1}} \Vvis R_{l^r_{j+1}}$, then $R_{l^r_{j+1}}$ is not aligned with $R_{c_{j+1}}$, but shifted to the left. Since, according to case $3$ of Lemma~\ref{between_all}, it is not possible that, for some $r, r'$, $1 \leq r < r' \leq 3$, $R_{l^r_{j+1}} \Vvis R_{c_{j+1}}$ and $R_{c_{j+1}} \Vvis R_{l^{r'}_{j+1}}$ in such a way that both $R_{l^{r}_{j+1}}$ and $R_{l^{r'}_{j+1}}$ are not aligned with $R_{c_{j+1}}$, but shifted to the left, we can conclude that we either have the situation illustrated in Figure~\ref{notAllOneSideLemmaFigureTwo}$(b)$, or a similar situation with the only difference that $R_{c_{j + 1}} \Hvis R_{l^1_{j + 1}}$ instead of $R_{l^1_{j + 1}} \Hvis R_{c_{j + 1}}$, which can be handled analogously. Again, we now consider the unit squares $R_{c^1_{j + 1}}$ and $R_{c^2_{j + 1}}$. We first note that, due to cases $1$ and $2$ of Lemma~\ref{between_all}, $R_{c^1_{j + 1}} \Vvis R_{c_{j + 1}}$ is not possible. If $R_{c^1_{j + 1}} \Hvis R_{c_{j + 1}}$, then the position of $R_{l^1_{j + 1}}$ and cases $1$ and $2$ of Lemma~\ref{between_all} imply that $R_{c^1_{j + 1}}$ cannot be aligned with $R_{c_{j + 1}}$, but must be shifted down. However, then either $R_{c^1_{j + 1}}$ cannot see $c^1_{j}$, or it blocks the view between $R_{c_j}$ and one of $R_{C^r_{j}\setminus\{c_j\}}$. If $R_{c_{j + 1}} \Vvis R_{c^1_{j + 1}}$, then, due to the position of $R_{l^3_{j+1}}$ and cases $1$ and $2$ of Lemma~\ref{between_all}, $R_{c^1_{j + 1}}$ cannot be aligned with $R_{c_{j + 1}}$, but must be shifted to the right. However, then $R_{c^1_{j + 1}}$ cannot see $R_{c^1_{j}}$. Consequently, we can conclude that $R_{c_{j + 1}} \Hvis R_{c^1_{j + 1}}$. In the same way, we can exclude $R_{c^2_{j + 1}} \Vvis R_{c_{j + 1}}$ and $R_{c^2_{j + 1}} \Hvis R_{c_{j + 1}}$. If $R_{c_{j + 1}} \Vvis R_{c^2_{j + 1}}$, then, due to $R_{c^2_j}$, $R_{c^2_{j + 1}}$ cannot see $R^1_{c_{j + 1}}$. Consequently, we must have $R_{c_{j + 1}} \Hvis R_{c^2_{j + 1}}$. However, there is no way for $R_{c^2_{j+1}}$ to see both $R_{c_{j + 1}}$ and $R_{c^2_j}$ without seeing $R_{c^1_j}$ (i.\,e., $R_{c^2_{j+1}}$ would have been placed with less than one unit distance to $R_{c^1_j}$). Consequently, the situation illustrated in Figure~\ref{notAllOneSideLemmaFigure}$(b)$ is not possible.
\end{enumerate}
\end{proof}

\begin{figure}
\begin{center}

\begin{tikzpicture}

\coordinate (width) at (0.5,0);
\coordinate (height) at (0,0.5);
\coordinate (labelshift) at (0.25, 0.25);

\coordinate (cj) at ($(0,0)$);
\coordinate (cj+1) at ($(cj) + 0.75*(width) + 1.25*(height)$);
\coordinate (cj1) at ($(cj) + 2*(width) + 0.9*(height) $);
\coordinate (cj2) at ($(cj) + 1.25*(width) - 0.75*(height)$);
\coordinate (cj-11) at ($(cj) + 3.5*(width) + 0.75*(height)$);
\coordinate (cj-12) at ($(cj) + 3.25*(width) - 0.5*(height)$);
\coordinate (cj-1) at ($(cj) + 4.75*(width) + 0.25*(height)$);

\coordinate (5) at ($(cj2) - (0,0.75)$);

\node (label) at ($(5) + (labelshift)$)  {$(a)$};

\coordinate (temp) at (cj);
\draw[black] ($(temp)$) -- ($(temp) + (width)$) -- ($(temp) + (width) + (height)$) -- ($(temp) + (height)$) -- ($(temp)$);
\node (label) at ($(temp) + (labelshift)$)  {$c_j$};

\coordinate (temp) at (cj+1);
\draw[black] ($(temp)$) -- ($(temp) + (width)$) -- ($(temp) + (width) + (height)$) -- ($(temp) + (height)$) -- ($(temp)$);
\node (label) at ($(temp) + (labelshift)$)  {$y$};

\coordinate (temp) at (cj1);
\draw[black] ($(temp)$) -- ($(temp) + (width)$) -- ($(temp) + (width) + (height)$) -- ($(temp) + (height)$) -- ($(temp)$);

\coordinate (temp) at (cj2);
\draw[black] ($(temp)$) -- ($(temp) + (width)$) -- ($(temp) + (width) + (height)$) -- ($(temp) + (height)$) -- ($(temp)$);

\coordinate (temp) at (cj-11);
\draw[black] ($(temp)$) -- ($(temp) + (width)$) -- ($(temp) + (width) + (height)$) -- ($(temp) + (height)$) -- ($(temp)$);

\coordinate (temp) at (cj-12);
\draw[black] ($(temp)$) -- ($(temp) + (width)$) -- ($(temp) + (width) + (height)$) -- ($(temp) + (height)$) -- ($(temp)$);

\coordinate (temp) at (cj-1);
\draw[black] ($(temp)$) -- ($(temp) + (width)$) -- ($(temp) + (width) + (height)$) -- ($(temp) + (height)$) -- ($(temp)$);

\end{tikzpicture}
\hspace{1cm}
\begin{tikzpicture}

\coordinate (width) at (0.5,0);
\coordinate (height) at (0,0.5);
\coordinate (labelshift) at (0.25, 0.25);

\coordinate (cj) at ($(0,0)$);
\coordinate (cj+1) at ($(cj) + 0.9*(width) + 2*(height)$);
\coordinate (cj1) at ($(cj) + 1.8*(width) + 0.9*(height) $);
\coordinate (cj2) at ($(cj) + 1.25*(width) - 0.75*(height)$);
\coordinate (cj-11) at ($(cj) + 3.5*(width) + 0.75*(height)$);
\coordinate (cj-12) at ($(cj) + 3.25*(width) - 0.5*(height)$);
\coordinate (cj-1) at ($(cj) + 4.75*(width) + 0.25*(height)$);

\coordinate (5) at ($(cj2) - (0,0.75)$);

\node (label) at ($(5) + (labelshift)$)  {$(b)$};

\coordinate (temp) at (cj);
\draw[black] ($(temp)$) -- ($(temp) + (width)$) -- ($(temp) + (width) + (height)$) -- ($(temp) + (height)$) -- ($(temp)$);
\node (label) at ($(temp) + (labelshift)$)  {$c_j$};

\coordinate (temp) at (cj+1);
\draw[black] ($(temp)$) -- ($(temp) + (width)$) -- ($(temp) + (width) + (height)$) -- ($(temp) + (height)$) -- ($(temp)$);
\node (label) at ($(temp) + (labelshift)$)  {$y$};

\coordinate (temp) at (cj1);
\draw[black] ($(temp)$) -- ($(temp) + (width)$) -- ($(temp) + (width) + (height)$) -- ($(temp) + (height)$) -- ($(temp)$);

\coordinate (temp) at (cj2);
\draw[black] ($(temp)$) -- ($(temp) + (width)$) -- ($(temp) + (width) + (height)$) -- ($(temp) + (height)$) -- ($(temp)$);

\coordinate (temp) at (cj-11);
\draw[black] ($(temp)$) -- ($(temp) + (width)$) -- ($(temp) + (width) + (height)$) -- ($(temp) + (height)$) -- ($(temp)$);

\coordinate (temp) at (cj-12);
\draw[black] ($(temp)$) -- ($(temp) + (width)$) -- ($(temp) + (width) + (height)$) -- ($(temp) + (height)$) -- ($(temp)$);

\coordinate (temp) at (cj-1);
\draw[black] ($(temp)$) -- ($(temp) + (width)$) -- ($(temp) + (width) + (height)$) -- ($(temp) + (height)$) -- ($(temp)$);

\end{tikzpicture}
\hspace{1cm}

\begin{tikzpicture}

\coordinate (width) at (0.5,0);
\coordinate (height) at (0,0.5);
\coordinate (labelshift) at (0.25, 0.25);

\coordinate (cj) at ($(0,0)$);
\coordinate (cj+1) at ($(cj) + 1.25*(width) + 0.85*(height)$);
\coordinate (cj1) at ($(cj) + 2.5*(width) + 0.7*(height) $);
\coordinate (cj2) at ($(cj) + 1.6*(width) - 0.85*(height)$);
\coordinate (cj-11) at ($(cj) + 4*(width) + 0.55*(height)$);
\coordinate (cj-12) at ($(cj) + 3.75*(width) - 0.7*(height)$);
\coordinate (cj-1) at ($(cj) + 5.25*(width) + 0.05*(height)$);

\coordinate (5) at ($(cj2) - (0,0.75)$);

\node (label) at ($(5) + (labelshift)$)  {$(c)$};

\coordinate (temp) at (cj);
\draw[black] ($(temp)$) -- ($(temp) + (width)$) -- ($(temp) + (width) + (height)$) -- ($(temp) + (height)$) -- ($(temp)$);
\node (label) at ($(temp) + (labelshift)$)  {$c_j$};

\coordinate (temp) at (cj+1);
\draw[black] ($(temp)$) -- ($(temp) + (width)$) -- ($(temp) + (width) + (height)$) -- ($(temp) + (height)$) -- ($(temp)$);
\node (label) at ($(temp) + (labelshift)$)  {$y$};

\coordinate (temp) at (cj1);
\draw[black] ($(temp)$) -- ($(temp) + (width)$) -- ($(temp) + (width) + (height)$) -- ($(temp) + (height)$) -- ($(temp)$);

\coordinate (temp) at (cj2);
\draw[black] ($(temp)$) -- ($(temp) + (width)$) -- ($(temp) + (width) + (height)$) -- ($(temp) + (height)$) -- ($(temp)$);

\coordinate (temp) at (cj-11);
\draw[black] ($(temp)$) -- ($(temp) + (width)$) -- ($(temp) + (width) + (height)$) -- ($(temp) + (height)$) -- ($(temp)$);

\coordinate (temp) at (cj-12);
\draw[black] ($(temp)$) -- ($(temp) + (width)$) -- ($(temp) + (width) + (height)$) -- ($(temp) + (height)$) -- ($(temp)$);

\coordinate (temp) at (cj-1);
\draw[black] ($(temp)$) -- ($(temp) + (width)$) -- ($(temp) + (width) + (height)$) -- ($(temp) + (height)$) -- ($(temp)$);

\end{tikzpicture}

\caption{Illustrations for the proof of Lemma~\ref{notAllOneSideLemma}}
\label{notAllOneSideLemmaFigure}
\end{center}
\end{figure}

\begin{figure}
\begin{center}

\begin{tikzpicture}

\coordinate (width) at (0.6,0);
\coordinate (height) at (0,0.6);
\coordinate (labelshift) at (0.325, 0.3);

\coordinate (cj) at ($(0,0)$);
\coordinate (cj+1) at ($(cj) + 0.75*(width) + 1.25*(height)$);
\coordinate (cj1) at ($(cj) + 2*(width) + 0.9*(height) $);
\coordinate (cj2) at ($(cj) + 1.25*(width) - 0.75*(height)$);
\coordinate (cj-11) at ($(cj) + 3.5*(width) + 0.75*(height)$);
\coordinate (cj-12) at ($(cj) + 3.25*(width) - 0.5*(height)$);
\coordinate (cj-1) at ($(cj) + 4.75*(width) + 0.25*(height)$);

\coordinate (lj+11) at ($(cj+1) - 2.5*(width)$);
\coordinate (lj+12) at ($(cj+1) + 2.5*(height)$);
\coordinate (lj+13) at ($(cj+1) + 5.25*(width) + 0.9*(height)$);

\coordinate (5) at ($(cj2) - (0,0.75)$);

\node (label) at ($(5) + (labelshift)$)  {$(a)$};

\coordinate (temp) at (cj);
\draw[black] ($(temp)$) -- ($(temp) + (width)$) -- ($(temp) + (width) + (height)$) -- ($(temp) + (height)$) -- ($(temp)$);
\node (label) at ($(temp) + (labelshift)$)  {$c_j$};

\coordinate (temp) at (cj+1);
\draw[black] ($(temp)$) -- ($(temp) + (width)$) -- ($(temp) + (width) + (height)$) -- ($(temp) + (height)$) -- ($(temp)$);
\node (label) at ($(temp) + (labelshift)$)  {$c_{j + 1}$};

\coordinate (temp) at (cj1);
\draw[black] ($(temp)$) -- ($(temp) + (width)$) -- ($(temp) + (width) + (height)$) -- ($(temp) + (height)$) -- ($(temp)$);
\node (label) at ($(temp) + (labelshift)$)  {$c^1_j$};

\coordinate (temp) at (cj2);
\draw[black] ($(temp)$) -- ($(temp) + (width)$) -- ($(temp) + (width) + (height)$) -- ($(temp) + (height)$) -- ($(temp)$);
\node (label) at ($(temp) + (labelshift)$)  {$c^2_j$};

\coordinate (temp) at (cj-11);
\draw[black] ($(temp)$) -- ($(temp) + (width)$) -- ($(temp) + (width) + (height)$) -- ($(temp) + (height)$) -- ($(temp)$);
\node (label) at ($(temp) + (labelshift)$)  {$c^1_{j-1}$};

\coordinate (temp) at (cj-12);
\draw[black] ($(temp)$) -- ($(temp) + (width)$) -- ($(temp) + (width) + (height)$) -- ($(temp) + (height)$) -- ($(temp)$);
\node (label) at ($(temp) + (labelshift)$)  {$c^2_{j-1}$};

\coordinate (temp) at (cj-1);
\draw[black] ($(temp)$) -- ($(temp) + (width)$) -- ($(temp) + (width) + (height)$) -- ($(temp) + (height)$) -- ($(temp)$);
\node (label) at ($(temp) + (labelshift)$)  {$c_{j-1}$};

\coordinate (temp) at (lj+11);
\draw[black] ($(temp)$) -- ($(temp) + (width)$) -- ($(temp) + (width) + (height)$) -- ($(temp) + (height)$) -- ($(temp)$);
\node (label) at ($(temp) + (labelshift)$)  {$l^1_{j + 1}$};

\coordinate (temp) at (lj+12);
\draw[black] ($(temp)$) -- ($(temp) + (width)$) -- ($(temp) + (width) + (height)$) -- ($(temp) + (height)$) -- ($(temp)$);
\node (label) at ($(temp) + (labelshift)$)  {$l^2_{j + 1}$};

\coordinate (temp) at (lj+13);
\draw[black] ($(temp)$) -- ($(temp) + (width)$) -- ($(temp) + (width) + (height)$) -- ($(temp) + (height)$) -- ($(temp)$);
\node (label) at ($(temp) + (labelshift)$)  {$l^3_{j + 1}$};
\end{tikzpicture}
\hspace{1cm}
\begin{tikzpicture}

\coordinate (width) at (0.6,0);
\coordinate (height) at (0,0.6);
\coordinate (labelshift) at (0.325, 0.3);

\coordinate (cj) at ($(0,0)$);
\coordinate (cj+1) at ($(cj) + 1.75*(width) + 0.85*(height)$);
\coordinate (cj1) at ($(cj) + 3*(width) + 0.7*(height) $);
\coordinate (cj2) at ($(cj) + 2.5*(width) - 0.85*(height)$);
\coordinate (cj-11) at ($(cj) + 4.5*(width) + 0.55*(height)$);
\coordinate (cj-12) at ($(cj) + 4.25*(width) - 0.7*(height)$);
\coordinate (cj-1) at ($(cj) + 5.75*(width) + 0.05*(height)$);

\coordinate (lj+11) at ($(cj+1) - 3*(width) + 0.4*(height)$);
\coordinate (lj+12) at ($(cj+1) + 2.5*(height)$);
\coordinate (lj+13) at ($(cj+1) - 0.4*(width) - 3*(height)$);

\coordinate (5) at ($(lj+13) - (0,0.75)$);

\node (label) at ($(5) + (labelshift)$)  {$(b)$};

\coordinate (temp) at (cj);
\draw[black] ($(temp)$) -- ($(temp) + (width)$) -- ($(temp) + (width) + (height)$) -- ($(temp) + (height)$) -- ($(temp)$);
\node (label) at ($(temp) + (labelshift)$)  {$c_j$};

\coordinate (temp) at (cj+1);
\draw[black] ($(temp)$) -- ($(temp) + (width)$) -- ($(temp) + (width) + (height)$) -- ($(temp) + (height)$) -- ($(temp)$);
\node (label) at ($(temp) + (labelshift)$)  {$c_{j + 1}$};

\coordinate (temp) at (cj1);
\draw[black] ($(temp)$) -- ($(temp) + (width)$) -- ($(temp) + (width) + (height)$) -- ($(temp) + (height)$) -- ($(temp)$);
\node (label) at ($(temp) + (labelshift)$)  {$c^1_j$};

\coordinate (temp) at (cj2);
\draw[black] ($(temp)$) -- ($(temp) + (width)$) -- ($(temp) + (width) + (height)$) -- ($(temp) + (height)$) -- ($(temp)$);
\node (label) at ($(temp) + (labelshift)$)  {$c^2_j$};

\coordinate (temp) at (cj-11);
\draw[black] ($(temp)$) -- ($(temp) + (width)$) -- ($(temp) + (width) + (height)$) -- ($(temp) + (height)$) -- ($(temp)$);
\node (label) at ($(temp) + (labelshift)$)  {$c^1_{j-1}$};

\coordinate (temp) at (cj-12);
\draw[black] ($(temp)$) -- ($(temp) + (width)$) -- ($(temp) + (width) + (height)$) -- ($(temp) + (height)$) -- ($(temp)$);
\node (label) at ($(temp) + (labelshift)$)  {$c^2_{j-1}$};

\coordinate (temp) at (cj-1);
\draw[black] ($(temp)$) -- ($(temp) + (width)$) -- ($(temp) + (width) + (height)$) -- ($(temp) + (height)$) -- ($(temp)$);
\node (label) at ($(temp) + (labelshift)$)  {$c_{j-1}$};

\coordinate (temp) at (lj+11);
\draw[black] ($(temp)$) -- ($(temp) + (width)$) -- ($(temp) + (width) + (height)$) -- ($(temp) + (height)$) -- ($(temp)$);
\node (label) at ($(temp) + (labelshift)$)  {$l^1_{j + 1}$};

\coordinate (temp) at (lj+12);
\draw[black] ($(temp)$) -- ($(temp) + (width)$) -- ($(temp) + (width) + (height)$) -- ($(temp) + (height)$) -- ($(temp)$);
\node (label) at ($(temp) + (labelshift)$)  {$l^2_{j + 1}$};

\coordinate (temp) at (lj+13);
\draw[black] ($(temp)$) -- ($(temp) + (width)$) -- ($(temp) + (width) + (height)$) -- ($(temp) + (height)$) -- ($(temp)$);
\node (label) at ($(temp) + (labelshift)$)  {$l^3_{j + 1}$};
\end{tikzpicture}

\caption{Illustrations for the proof of Lemma~\ref{notAllOneSideLemma}}
\label{notAllOneSideLemmaFigureTwo}
\end{center}
\end{figure}

\subsection*{Proof of Lemma~\ref{literalsLemma}}

%\begin{lemma}
%For every $j$, $1 \leq j \leq m$, either $R_{c_j} \symHvis R_{L_j}$ or $R_{c_j} \symVvis R_{L_j}$.
%\end{lemma}

\begin{proof}
We first observe that it is not possible that $R_{c_j} \Hvis R_{L_j}$, $R_{L_j} \Hvis R_{c_j}$, $R_{c_j} \Vvis R_{L_j}$ or $R_{L_j} \Vvis R_{c_j}$. More precisely, all these cases mean that there are $R, R' \in R_{L_j}$, such that $R'$ is strictly between $R$ and $R_{c_j}$, which is a contradiction to statement $1$ of Lemma~\ref{between_all}. Hence, in order to prove the lemma, we only have to rule out the following cases (for the sake of convenience, we set $R_{L_j} = \{x, y, z\}$):
\begin{enumerate}
\item $x \Vvis R_{c_j}$, $R_{c_j} \Vvis y$ and $R_{c_j} \symHvis z$,
\item $x \Hvis R_{c_j}$, $R_{c_j} \Hvis y$ and $R_{c_j} \symVvis z$,
\item $\{x, y\} \Vvis R_{c_j}$ and $R_{c_j} \symHvis z$,
\item $R_{c_j} \Vvis \{x, y\}$ and $R_{c_j} \symHvis z$,
\item $\{x, y\} \Hvis R_{c_j}$ and $R_{c_j} \symVvis z$,
\item $R_{c_j} \Hvis \{x, y\}$ and $R_{c_j} \symVvis z$.
\end{enumerate}
Since cases $1$ and $2$ are symmetric, as well as cases $3$, $4$, $5$ and $6$, we only consider cases $1$ and $3$.
\begin{enumerate}
\item Case $1$ (see Figure~\ref{clausePathLemmaFigure}$(a)$): We assume that $R_{c_j} \Hvis z$; the case $z \Hvis R_{c_j}$ can be handled analogously. Due to statement $3$ of Lemma~\ref{between_all}, we can assume that $R_{c_j}$ blocks the view between $x$ and $y$ (which, in particular, means that if both $x$ and $y$ are not aligned with $R_{c_j}$, then they cannot both be shifted to the same side). We now consider the $K_4$ on vertices $C^l_{j}$. First, we assume that the unit squares $R_{C^l_{j}}$ are placed such that for some $S, S' \in R_{C^l_{j}}$ with $S \neq S'$, $R_{c_{j}} \Hvis \{S, S'\}$. By consulting Lemma~\ref{K4Lemma}, we observe that this means that there are $T, T' \in R_{C^l_{j}}$, such that  $R_{c_{j}} \Hvis \{T, T'\}$, neither $T$ nor $T'$ are aligned with $R_{c_{j}}$, $T$ is shifted upwards and $T'$ is shifted downwards. However, this necessarily means that there is a unit square $R \in R_{N(c_j)}$, such that $R$ is strictly between $R_{c_j}$ and $z$, or $z$ is strictly between $R_{c_j}$ and $R$, which is a contradiction to statement $1$ or $2$, respectively, of Lemma~\ref{between_all}. The same argument applies to the situations that two unit squares of $R_{C^l_{j}}$ are placed within vertical visibility both above or both below $R_{c_j}$. \par
If there are $R, R' \in R_{C^l_{j}}$ with $R \Vvis R_{c_j}$ and $R_{c_j} \Vvis R'$, they have to be shifted to the same side in order to see each other. However, since $x$ and $y$ are shifted to opposite directions, this means that $R$ or $R'$ is strictly between $R_{c_j}$ and $x$ or $y$, or $x$ or $y$ is strictly between $R$ or $R'$, which is a contradiction to case $1$ or $2$, respectively, of Lemma~\ref{between_all}. Consequently, there is at most one $R \in R_{C^l_{j}}$ with $R \symVvis R_{c_j}$ and, in the following, we assume that $R \Vvis R_{c_j}$ holds (the case $R_{c_j} \Vvis R$ can be handled analogously). In particular, this means $y$ is aligned with $R_{c_j}$. Now let $R'$ and $R''$ be the two remaining unit squares from $R_{C^l_{j}}$, i.\,e., $\{R', R''\} = R_{C^l_{j} \setminus \{c_j\}} \setminus \{R\}$. According to what we observed above, either $R' \Hvis R_{c_j}$ and $R_{c_j} \Hvis R''$, or $\{R', R''\} \Hvis R_{c_j}$. We first assume the former case and note that, according to Lemma~\ref{K4Lemma}, this means that the $K_4$ on vertices $C^l_{j}$ satisfies case $1$ of Lemma~\ref{K4Lemma} with $R_{c_j}$ playing the role of $R_4$ (Figure~\ref{clausePathLemmaFigure2}(a)), or case $3$ of Lemma~\ref{K4Lemma} with $R_{c_j}$ playing the role of $R_3$ (Figure~\ref{clausePathLemmaFigure2}(b)). In both these cases $z$ has to be shifted downwards to avoid a contradiction to case $1$ or $2$ of Lemma~\ref{between_all} with $R''$. On the other hand, if $\{R', R''\} \Hvis R_{c_j}$, then the $K_4$ on $C^l_{j}$ either satisfies case $2$ (Figure~\ref{clausePathLemmaFigure2}(c)) or case $3$ with $R$ in the role of $R_4$  (Figure~\ref{clausePathLemmaFigure2}(d)) of Lemma~\ref{K4Lemma}. Consequently, under the assumption that, for some $R \in R_{C^l_{j} \setminus \{c_j\}}$, $R \Vvis R_{c_j}$, the unit squares for the $K_4$ on vertices $C^l_{j}$ satisfy one of the cases illustrated in Figure~\ref{clausePathLemmaFigure2}$(a)$ to $(d)$, and, since the arguments from above apply in the same way, the same holds for the unit squares for the $K_4$ on vertices $C^r_{j}$. 
\begin{enumerate}
\item Both $K_4$ on vertices $C^l_{j}$ and $C^r_{j}$ satisfy case $(a)$, $(b)$ or $(c)$ from Figure~\ref{clausePathLemmaFigure2}: We note that this implies that there are $R \in R_{C^l_{j} \setminus \{c_j\}}$ and $S \in R_{C^r_{j} \setminus \{c_j\}}$ with $\{R, S\} \Vvis R_{c_j}$ with a distance of less than one unit from $R_{c_j}$. This is only possible if $R$ and $S$ are place (horizontally) next to each other, which means that one of them is strictly between $x$ and $R_{c_j}$, which yields a contradiction with case $1$ Lemma~\ref{between_all}.
\item The $K_4$ on vertices $C^l_{j}$ or the the $K_4$ on vertices $C^r_{j}$ satisfies case $(d)$ of Figure~\ref{clausePathLemmaFigure2}: We assume that the $K_4$ on vertices $C^l_{j}$ satisfies case $(d)$ of Figure~\ref{clausePathLemmaFigure2} (the other case is analogous). We note that this implies that there are $R', R'' \in R_{C^l_{j} \setminus \{c_j\}}$ with $\{R', R''\} \Hvis R_{c_j}$, such that both $R'$ and $R''$ have a horizontal distance of less than one to $R_{c_j}$. This means that there is no $S \in R_{C^r_{j} \setminus \{c_j\}}$ with $S \Hvis R_{c_j}$ that also has a distance of less that one unit from $R_{c_j}$. Consequently, the $K_4$ on vertices $C^r_{j}$ can only satisfy case $(a)$ of Figure~\ref{clausePathLemmaFigure2}. However, in this case $z$ cannot be aligned with $R_{c_j}$, which contradicts the fact that the $K_4$ on vertices $C^l_{j}$ satisfies case $(d)$ of Figure~\ref{clausePathLemmaFigure2} (which requires $z$ to be aligned with $R_{c_j}$).
\end{enumerate}
Consequently, we can assume that there is no $R \in R_{C^l_j \setminus \{c_j\}}$ with $R \symVvis R_{c_j}$ (or that this holds for the $K_4$ on vertices $C^r_{j}$, which can be handled analogously). Consequently, $R_{C^l_{j} \setminus \{c_j\}} \symHvis R_{c_j}$, which, by Lemma~\ref{K4Lemma}, implies that either $R_{C^l_{j} \setminus \{c_j\}} \Hvis R_{c_j}$ or $R_{c_j} \Hvis R_{C^l_{j} \setminus \{c_j\}}$. Since, as explained above, the latter leads to a contradiction, we can conclude that $R_{C^l_{j} \setminus \{c_j\}} \Hvis R_{c_j}$. Now if the $K_4$ on vertices $C^r_{j}$ satisfies case $(c)$ or $(d)$ of Figure~\ref{clausePathLemmaFigure2}, or if this $K_4$ is also realised exclusively by horizontal visibilities, then we obtain a contradiction to Lemma~\ref{notAllOneSideLemma}. Thus, we assume that the $K_4$ on vertices $C^r_{j}$ satisfies case $(a)$ or $(b)$, which means that $z$ is not aligned with $R_{c_j}$. This is a contradiction, since, due to Lemmas~\ref{K4Lemma}~and~\ref{between_all}, $R_{C^l_{j} \setminus \{c_j\}} \Hvis R_{c_j}$ implies that $z$ must be aligned with $R_{c_j}$. 
\item Case $3$ (Figure~\ref{clausePathLemmaFigure}$(b)$): Due to cases $1$ and $2$ of Lemma~\ref{between_all}, we know that neither $x$ nor $y$ is aligned with $R_{c_j}$ and, furthermore, there is no $R \in R_{C_{j}}$ with $R \Vvis R_{c_j}$. Next, we assume that there is also no $R \in R_{C_{j}}$ with $R_{c_j} \Vvis R$, which implies $R_{C_j \setminus \{c_j\}} \symHvis R_{c_j}$. By Lemma~\ref{K4Lemma}, this means that either $R_{C^l_j \setminus \{c_j\}} \Hvis R_{c_j}$ or $R_{c_j} \Hvis R_{C^l_j \setminus \{c_j\}}$ and either $R_{C^r_j \setminus \{c_j\}} \Hvis R_{c_j}$ or $R_{c_j} \Hvis R_{C^r_j \setminus \{c_j\}}$. However, $R_{c_j} \Hvis R_{C^l_j \setminus \{c_j\}}$ or $R_{c_j} \Hvis R_{C^r_j \setminus \{c_j\}}$ yields a contradiction with Lemma~\ref{between_all}, which implies that $R_{C_{j} \setminus \{c_j\}} \Hvis R_{c_j}$. This is a contradiction to Lemma~\ref{notAllOneSideLemma}. Consequently, there is at least one $R \in R_{C_{j}}$ with $R_{c_j} \Vvis R$ and, due to case $3$ of Lemma~\ref{between_all}, we can conclude that there is exactly one such unit square that is aligned with $R_{c_j}$. Moreover, without loss of generality, let $R \in R_{C^l_{j}}$. This means that the $K_4$ on vertices $C^l_j$ satisfies case $1$ of Lemma~\ref{K4Lemma} with $R_{c_j}$ playing the role of $R_2$. In particular, this implies that $z$ cannot be aligned with $R_{c_j}$, since this would lead to a contradiction with case $1$ or $2$ of Lemma~\ref{between_all}. However, due to the fact $R_{C^r_{j} \setminus \{c_j\}} \symHvis R_{c_j}$, we obtain a contradiction to one of the cases of Lemma~\ref{between_all}.

\begin{figure}
\begin{center}
\begin{tabular}{p{3.5cm} p{3.5cm} p{3cm} p{3cm}}
\centering\begin{tikzpicture}

\coordinate (width) at (0.4,0);
\coordinate (height) at (0,0.4);
\coordinate (labelshift) at (0.2, 0.2);

\coordinate (1) at ($(0,0)$);
\coordinate (2) at ($(1) + 0.6*(width) + 3*(height)$);
\coordinate (3) at ($(1) - 2.3*(height)$);
\coordinate (4) at ($(1) + 4*(width) - 0.6*(height)$);
\coordinate (5) at ($(1) -0.5*(width)+1.3*(height)$);
\coordinate (7) at ($(1) + 2.3*(width)+0.5*(height)$);
\coordinate (6) at ($(1) -1.7*(width)+0.5*(height)$);

\coordinate (temp) at (1);
\draw[black] ($(temp)$) -- ($(temp) + (width)$) -- ($(temp) + (width) + (height)$) -- ($(temp) + (height)$) -- ($(temp)$);
\node (label) at ($(temp) + (labelshift)$)  {$c_j$};

\coordinate (temp) at (2);
\draw[black] ($(temp)$) -- ($(temp) + (width)$) -- ($(temp) + (width) + (height)$) -- ($(temp) + (height)$) -- ($(temp)$);
\node (label) at ($(temp) + (labelshift)$)  {$x$};

\coordinate (temp) at (3);
\draw[black] ($(temp)$) -- ($(temp) + (width)$) -- ($(temp) + (width) + (height)$) -- ($(temp) + (height)$) -- ($(temp)$);
\node (label) at ($(temp) + (labelshift)$)  {$y$};

\coordinate (temp) at (4);
\draw[black] ($(temp)$) -- ($(temp) + (width)$) -- ($(temp) + (width) + (height)$) -- ($(temp) + (height)$) -- ($(temp)$);
\node (label) at ($(temp) + (labelshift)$)  {$z$};

\coordinate (temp) at (5);
\draw[black] ($(temp)$) -- ($(temp) + (width)$) -- ($(temp) + (width) + (height)$) -- ($(temp) + (height)$) -- ($(temp)$);
\node (label) at ($(temp) + (labelshift)$)  {$R$};

\coordinate (temp) at (6);
\draw[black] ($(temp)$) -- ($(temp) + (width)$) -- ($(temp) + (width) + (height)$) -- ($(temp) + (height)$) -- ($(temp)$);
\node (label) at ($(temp) + (labelshift)$)  {$R'$};

\coordinate (temp) at (7);
\draw[black] ($(temp)$) -- ($(temp) + (width)$) -- ($(temp) + (width) + (height)$) -- ($(temp) + (height)$) -- ($(temp)$);
\node (label) at ($(temp) + (labelshift)$)  {$R''$};

\end{tikzpicture}
&
\centering \begin{tikzpicture}

\coordinate (width) at (0.4,0);
\coordinate (height) at (0,0.4);
\coordinate (labelshift) at (0.2, 0.2);

\coordinate (1) at ($(0,0)$);
\coordinate (2) at ($(1) + 0.6*(width) + 3*(height)$);
\coordinate (3) at ($(1) - 2.3*(height)$);
\coordinate (4) at ($(1) + 4*(width) - 0.6*(height)$);
\coordinate (5) at ($(1) -0.5*(width)+1.5*(height)$);
\coordinate (7) at ($(1) +2.3*(width)+0.6*(height)$);
\coordinate (6) at ($(1) -1.3*(width)+0.3*(height)$);

\coordinate (temp) at (1);
\draw[black] ($(temp)$) -- ($(temp) + (width)$) -- ($(temp) + (width) + (height)$) -- ($(temp) + (height)$) -- ($(temp)$);
\node (label) at ($(temp) + (labelshift)$)  {$c_j$};

\coordinate (temp) at (2);
\draw[black] ($(temp)$) -- ($(temp) + (width)$) -- ($(temp) + (width) + (height)$) -- ($(temp) + (height)$) -- ($(temp)$);
\node (label) at ($(temp) + (labelshift)$)  {$x$};

\coordinate (temp) at (3);
\draw[black] ($(temp)$) -- ($(temp) + (width)$) -- ($(temp) + (width) + (height)$) -- ($(temp) + (height)$) -- ($(temp)$);
\node (label) at ($(temp) + (labelshift)$)  {$y$};

\coordinate (temp) at (4);
\draw[black] ($(temp)$) -- ($(temp) + (width)$) -- ($(temp) + (width) + (height)$) -- ($(temp) + (height)$) -- ($(temp)$);
\node (label) at ($(temp) + (labelshift)$)  {$z$};

\coordinate (temp) at (5);
\draw[black] ($(temp)$) -- ($(temp) + (width)$) -- ($(temp) + (width) + (height)$) -- ($(temp) + (height)$) -- ($(temp)$);
\node (label) at ($(temp) + (labelshift)$)  {$R$};

\coordinate (temp) at (6);
\draw[black] ($(temp)$) -- ($(temp) + (width)$) -- ($(temp) + (width) + (height)$) -- ($(temp) + (height)$) -- ($(temp)$);
\node (label) at ($(temp) + (labelshift)$)  {$R'$};

\coordinate (temp) at (7);
\draw[black] ($(temp)$) -- ($(temp) + (width)$) -- ($(temp) + (width) + (height)$) -- ($(temp) + (height)$) -- ($(temp)$);
\node (label) at ($(temp) + (labelshift)$)  {$R''$};

\end{tikzpicture}&\centering \begin{tikzpicture}

\coordinate (width) at (0.4,0);
\coordinate (height) at (0,0.4);
\coordinate (labelshift) at (0.2, 0.2);

\coordinate (1) at ($(0,0)$);
\coordinate (2) at ($(1) + 0.6*(width) + 3*(height)$);
\coordinate (3) at ($(1) - 2.3*(height)$);
\coordinate (4) at ($(1) + 2*(width)$);
\coordinate (5) at ($(1) -0.5*(width)+1.3*(height)$);
\coordinate (7) at ($(1) -1.2*(width)-0.6*(height)$);
\coordinate (6) at ($(1) -1.7*(width)+0.8*(height)$);

\coordinate (temp) at (1);
\draw[black] ($(temp)$) -- ($(temp) + (width)$) -- ($(temp) + (width) + (height)$) -- ($(temp) + (height)$) -- ($(temp)$);
\node (label) at ($(temp) + (labelshift)$)  {$c_j$};

\coordinate (temp) at (2);
\draw[black] ($(temp)$) -- ($(temp) + (width)$) -- ($(temp) + (width) + (height)$) -- ($(temp) + (height)$) -- ($(temp)$);
\node (label) at ($(temp) + (labelshift)$)  {$x$};

\coordinate (temp) at (3);
\draw[black] ($(temp)$) -- ($(temp) + (width)$) -- ($(temp) + (width) + (height)$) -- ($(temp) + (height)$) -- ($(temp)$);
\node (label) at ($(temp) + (labelshift)$)  {$y$};

\coordinate (temp) at (4);
\draw[black] ($(temp)$) -- ($(temp) + (width)$) -- ($(temp) + (width) + (height)$) -- ($(temp) + (height)$) -- ($(temp)$);
\node (label) at ($(temp) + (labelshift)$)  {$z$};

\coordinate (temp) at (5);
\draw[black] ($(temp)$) -- ($(temp) + (width)$) -- ($(temp) + (width) + (height)$) -- ($(temp) + (height)$) -- ($(temp)$);
\node (label) at ($(temp) + (labelshift)$)  {$R$};

\coordinate (temp) at (6);
\draw[black] ($(temp)$) -- ($(temp) + (width)$) -- ($(temp) + (width) + (height)$) -- ($(temp) + (height)$) -- ($(temp)$);
\node (label) at ($(temp) + (labelshift)$)  {$R'$};

\coordinate (temp) at (7);
\draw[black] ($(temp)$) -- ($(temp) + (width)$) -- ($(temp) + (width) + (height)$) -- ($(temp) + (height)$) -- ($(temp)$);
\node (label) at ($(temp) + (labelshift)$)  {$R''$};

\end{tikzpicture}
&
\centering\begin{tikzpicture}

\coordinate (width) at (0.4,0);
\coordinate (height) at (0,0.4);
\coordinate (labelshift) at (0.2, 0.2);

\coordinate (1) at ($(0,0)$);
\coordinate (2) at ($(1) + 0.7*(width) + 3*(height)$);
\coordinate (3) at ($(1) - 2.3*(height)$);
\coordinate (4) at ($(1) + 2*(width) $);
\coordinate (5) at ($(1) -0.7*(width)+2*(height)$);
\coordinate (7) at ($(1) -1.2*(width)-0.7*(height)$);
\coordinate (6) at ($(1) -1.5*(width)+0.7*(height)$);

\coordinate (temp) at (1);
\draw[black] ($(temp)$) -- ($(temp) + (width)$) -- ($(temp) + (width) + (height)$) -- ($(temp) + (height)$) -- ($(temp)$);
\node (label) at ($(temp) + (labelshift)$)  {$c_j$};

\coordinate (temp) at (2);
\draw[black] ($(temp)$) -- ($(temp) + (width)$) -- ($(temp) + (width) + (height)$) -- ($(temp) + (height)$) -- ($(temp)$);
\node (label) at ($(temp) + (labelshift)$)  {$x$};

\coordinate (temp) at (3);
\draw[black] ($(temp)$) -- ($(temp) + (width)$) -- ($(temp) + (width) + (height)$) -- ($(temp) + (height)$) -- ($(temp)$);
\node (label) at ($(temp) + (labelshift)$)  {$y$};

\coordinate (temp) at (4);
\draw[black] ($(temp)$) -- ($(temp) + (width)$) -- ($(temp) + (width) + (height)$) -- ($(temp) + (height)$) -- ($(temp)$);
\node (label) at ($(temp) + (labelshift)$)  {$z$};

\coordinate (temp) at (5);
\draw[black] ($(temp)$) -- ($(temp) + (width)$) -- ($(temp) + (width) + (height)$) -- ($(temp) + (height)$) -- ($(temp)$);
\node (label) at ($(temp) + (labelshift)$)  {$R$};

\coordinate (temp) at (6);
\draw[black] ($(temp)$) -- ($(temp) + (width)$) -- ($(temp) + (width) + (height)$) -- ($(temp) + (height)$) -- ($(temp)$);
\node (label) at ($(temp) + (labelshift)$)  {$R'$};

\coordinate (temp) at (7);
\draw[black] ($(temp)$) -- ($(temp) + (width)$) -- ($(temp) + (width) + (height)$) -- ($(temp) + (height)$) -- ($(temp)$);
\node (label) at ($(temp) + (labelshift)$)  {$R''$};

\end{tikzpicture}

\cr & \cr
\centering$(a)$&\centering $(b)$&\centering$(c)$&\centering $(d)$

\end{tabular}
\caption{Illustrations for the proof of Lemma~\ref{literalsLemma}.} 
\label{clausePathLemmaFigure2}
\end{center}
\end{figure}

\begin{figure}
\begin{center} 

\begin{tikzpicture}

\coordinate (width) at (0.4,0);
\coordinate (height) at (0,0.4);
\coordinate (labelshift) at (0.2, 0.2);

\coordinate (1) at ($(0,0)$);
\coordinate (2) at ($(1) - 0.6*(width) + 2.3*(height)$);
\coordinate (3) at ($(1) + 0.6*(width) - 2.3*(height)$);
\coordinate (4) at ($(1) + 2.3*(width) + 0.3*(height)$);
\coordinate (5) at ($(3) - (0,0.75)$);

\node (label) at ($(5) + (labelshift)$)  {$(a)$};

\coordinate (temp) at (1);
\draw[black] ($(temp)$) -- ($(temp) + (width)$) -- ($(temp) + (width) + (height)$) -- ($(temp) + (height)$) -- ($(temp)$);
\node (label) at ($(temp) + (labelshift)$)  {$c_j$};

\coordinate (temp) at (2);
\draw[black] ($(temp)$) -- ($(temp) + (width)$) -- ($(temp) + (width) + (height)$) -- ($(temp) + (height)$) -- ($(temp)$);
\node (label) at ($(temp) + (labelshift)$)  {$x$};

\coordinate (temp) at (3);
\draw[black] ($(temp)$) -- ($(temp) + (width)$) -- ($(temp) + (width) + (height)$) -- ($(temp) + (height)$) -- ($(temp)$);
\node (label) at ($(temp) + (labelshift)$)  {$y$};

\coordinate (temp) at (4);
\draw[black] ($(temp)$) -- ($(temp) + (width)$) -- ($(temp) + (width) + (height)$) -- ($(temp) + (height)$) -- ($(temp)$);
\node (label) at ($(temp) + (labelshift)$)  {$z$};

\end{tikzpicture}
\hspace{1.5cm}
\begin{tikzpicture}

\coordinate (width) at (0.4,0);
\coordinate (height) at (0,0.4);
\coordinate (labelshift) at (0.2, 0.2);

\coordinate (1) at ($(0,0)$);
\coordinate (2) at ($(1) - 0.8*(width) + 3.2*(height)$);
\coordinate (3) at ($(1) + 0.4*(width) + 2.0*(height)$);
\coordinate (4) at ($(1) + 3*(width) + 0.5*(height)$);
\coordinate (5) at ($(1) - (0,0.75)$);

\node (label) at ($(5) + (labelshift)$)  {$(b)$};

\coordinate (temp) at (1);
\draw[black] ($(temp)$) -- ($(temp) + (width)$) -- ($(temp) + (width) + (height)$) -- ($(temp) + (height)$) -- ($(temp)$);
\node (label) at ($(temp) + (labelshift)$)  {$c_j$};

\coordinate (temp) at (2);
\draw[black] ($(temp)$) -- ($(temp) + (width)$) -- ($(temp) + (width) + (height)$) -- ($(temp) + (height)$) -- ($(temp)$);
\node (label) at ($(temp) + (labelshift)$)  {$x$};

\coordinate (temp) at (3);
\draw[black] ($(temp)$) -- ($(temp) + (width)$) -- ($(temp) + (width) + (height)$) -- ($(temp) + (height)$) -- ($(temp)$);
\node (label) at ($(temp) + (labelshift)$)  {$y$};

\coordinate (temp) at (4);
\draw[black] ($(temp)$) -- ($(temp) + (width)$) -- ($(temp) + (width) + (height)$) -- ($(temp) + (height)$) -- ($(temp)$);
\node (label) at ($(temp) + (labelshift)$)  {$z$};

\end{tikzpicture}
\hspace{1.5cm}

\caption{Illustrations for the proof of Lemma~\ref{literalsLemma}.}
\label{clausePathLemmaFigure}
\end{center}
\end{figure}
\end{enumerate}

\end{proof}

\subsection*{Proof of Lemma~\ref{final}}

Before we prove the lemma, we observe that Lemmas~\ref{between_all},~\ref{notAllOneSideLemma}~and~\ref{literalsLemma} also hold for the part of the graph consisiting of the vertices $\{x_i, x^1_i, x^2_i \mid 1 \leq i \leq n + 1\} \cup \{t_i, f^1_i, f^2_i,  \mid 1 \leq i \leq n\}$. More precisely, let $X_0=\{c_{2m-1},c_{2m-1}^1,c_{2m-1}^2,c_{2m},x_1^1,x_1^2,x_1\}$, $X_1=\{c_{2m},x_1^1,x_1^2,x_1,x_2^1,x_2^2,x_2\}$ and $X_i=\{x_{i-1,}x_i^1,x_i^2,x_i,x_{i+1}^1,x_{i+1}^2,x_{i+1}\}$, for all $i$, $2\leq i\leq n$, and let $A_i=\{t_i,f_i^1,f_i^2\}$, for all $i$, $1 \leq i \leq n$. Then Lemma~\ref{between_all} also holds for the version where $c_i$ is replaced by $x_i$ and $l^r_i$ is replaced by $t_i$, $f^1_i$ or $f^2_i$, Lemma~\ref{notAllOneSideLemma} also holds for the version where $C_i$ is replaced by $X_i$ and $c_i$ is replaced by $x_i$ (or $c_{2m}$ in case of $X_0$), and Lemma~\ref{literalsLemma} also holds for the version where $c_i$ is replaced by $x_i$ and $L_i$ is replaced by $A_i$. This is due to the identical structure of these parts of the graph. In the following, we shall refer to these more general versions of the lemmas.

%\begin{lemma}
%If $G \in \unitSquareGraphs$, then $F$ is not-all-equal satisfiable.
%\end{lemma}

\begin{proof}
We assume that, for some $j$, $1 \leq j \leq 2m-1$, $R_{c_j} \symVvis L_{j}$. If, for some $x \in C_{j} \setminus \{c_j\}$, $R_x \symVvis R_{c_j}$, then we obtain a contradiction with case $1$ or $2$ of Lemma~\ref{between_all}; thus, $R_{c_j} \symHvis R_{C_j \setminus \{c_j\}}$. Consequently, for every $j$, $1\leq j\leq 2m-1$, $R_{c_j} \symHvis R_{C_j \setminus \{c_j\}}$ or $R_{c_j} \symVvis R_{C_j \setminus \{c_j\}}$, $R_{c_{2m}} \symHvis R_{X_0 \setminus \{c_{2m}\}}$ or $R_{c_{2m}} \symVvis R_{X_0 \setminus \{c_{2m}\}}$ and, for every $i$, $1\leq i\leq n$, $R_{x_i} \symHvis R_{X_i \setminus \{x_i\}}$ or $R_{x_i} \symVvis R_{X_i \setminus \{x_i\}}$.\par
We assume, without loss of generality, that $R_{c_1} \symHvis R_{C_1 \setminus \{c_1\}}$. By Lemma~\ref{K4Lemma}, this implies that either $R_{c_1} \Hvis R_{C^l_1 \setminus \{c_1\}}$ or $R_{C^l_1 \setminus \{c_1\}} \Hvis R_{c_1}$ and that either $R_{c_1} \Hvis R_{C^r_1 \setminus \{c_1\}}$ or $R_{C^r_1 \setminus \{c_1\}} \Hvis R_{c_1}$. Moreover, Lemma~\ref{notAllOneSideLemma} yields that $R_{C^l_1 \setminus \{c_1\}} \Hvis R_{c_1}$ if and only if  $R_{c_1} \Hvis R_{C^r_1 \setminus \{c_1\}}$. Obviously, this argument applies to every $c_j$, $1 \leq i \leq 2m$, and every $i$, $1 \leq i \leq n$. We now assume, without loss of generality, that $R_{C^l_1 \setminus \{c_1\}} \Hvis R_{c_1}$, which implies $R_{c_1} \Hvis R_{C^r_1 \setminus \{c_1\}}$ and, in particular, $R_{C^l_2 \setminus \{c_2\}} \Hvis R_{c_2}$. Repeating this argument inductively on all $C_{j}$, $1\leq j\leq 2m-1$, and on all $X_i$, $0 \leq i \leq n$, implies that the part of the graph consisting of vertices $\{c_i,c_i^1,c_i^2\mid 0\leq j\leq 2m-1\}\cup\{c_{2m}\}\cup\{x_i,x_i^1,x_i^2\mid 1\leq i\leq n+1\}$, which we shall call \emph{backbone} in the following, is represented by a layout that is \visomorphic{} to the one in Figure~\ref{backbone}, except for the $K_4$ on vertices $R_{C_1^l}$ and the $K_4$ on vertices $x_n, x^1_{n+1}, x^2_{n+1}, , x_{n+1}$, which could also satisfy case $3$ of Lemma~\ref{K4Lemma} (note that all the other $K_4$ must satisfy case $1$ of Lemma~\ref{K4Lemma}, since all their visibilities are horizontal). Moreover, as explained above, this also implies that, for every $j$, $1 \leq j \leq 2m$, $R_{L_j} \symVvis R_{c_j}$ and, for every $i$, $1 \leq i \leq n$, $R_{A_i} \symVvis R_{x_i}$.\par

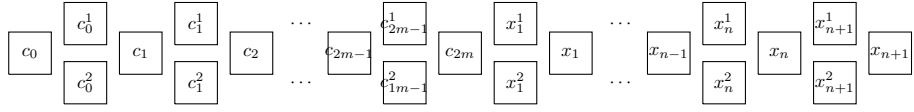
\begin{figure}
\begin{center}
\begin{tikzpicture}[scale=0.7, transform shape]
\coordinate (width) at (0.8,0);
\coordinate (height) at (0,0.8);
\coordinate (labelshift) at (0.4, 0.4);

\coordinate (1) at ($(0,0)$);
\coordinate (2) at ($(1) + 2.6*(width)$);
\coordinate (3) at ($(1) + 0.7*(height) + 1.3*(width)$);
\coordinate (4) at ($(1) - 0.7*(height) + 1.3*(width)$);
\coordinate (5) at ($(2) + 2.6*(width)$);
\coordinate (6) at ($(2) + 0.7*(height) + 1.3*(width)$);
\coordinate (7) at ($(2) - 0.7*(height) + 1.3*(width)$);

\coordinate (8) at ($(6) + 3*(width) + 0.5*(height)$);
\coordinate (9) at ($(7) + 3*(width) + 0.5*(height)$);

\coordinate (10) at ($7.5*(width)$);
\coordinate (11) at ($(10) + 2.6*(width)$);
\coordinate (12) at ($(10) + 0.7*(height) + 1.3*(width)$);
\coordinate (13) at ($(10) - 0.7*(height) + 1.3*(width)$);
\coordinate (14) at ($(11) + 2.6*(width)$);
\coordinate (15) at ($(11) + 0.7*(height) + 1.3*(width)$);
\coordinate (16) at ($(11) - 0.7*(height) + 1.3*(width)$);

\coordinate (17) at ($(15) + 3*(width) + 0.5*(height)$);
\coordinate (18) at ($(16) + 3*(width) + 0.5*(height)$);

\coordinate (101) at ($15*(width)$);
\coordinate (111) at ($(101) + 2.6*(width)$);
\coordinate (121) at ($(101) + 0.7*(height) + 1.3*(width)$);
\coordinate (131) at ($(101) - 0.7*(height) + 1.3*(width)$);
\coordinate (141) at ($(111) + 2.6*(width)$);
\coordinate (151) at ($(111) + 0.7*(height) + 1.3*(width)$);
\coordinate (161) at ($(111) - 0.7*(height) + 1.3*(width)$);

\coordinate (temp) at (1);
\draw[black] ($(temp)$) -- ($(temp) + (width)$) -- ($(temp) + (width) + (height)$) -- ($(temp) + (height)$) -- ($(temp)$);
\node (label) at ($(temp) + (labelshift)$)  {$c_{0}$};

\coordinate (temp) at (2);
\draw[black] ($(temp)$) -- ($(temp) + (width)$) -- ($(temp) + (width) + (height)$) -- ($(temp) + (height)$) -- ($(temp)$);
\node (label) at ($(temp) + (labelshift)$)  {$c_{1}$};

\coordinate (temp) at (3);
\draw[black] ($(temp)$) -- ($(temp) + (width)$) -- ($(temp) + (width) + (height)$) -- ($(temp) + (height)$) -- ($(temp)$);
\node (label) at ($(temp) + (labelshift)$)  {$c_{0}^1$};

\coordinate (temp) at (4);
\draw[black] ($(temp)$) -- ($(temp) + (width)$) -- ($(temp) + (width) + (height)$) -- ($(temp) + (height)$) -- ($(temp)$);
\node (label) at ($(temp) + (labelshift)$)  {$c_{0}^2$};

\coordinate (temp) at (5);
\draw[black] ($(temp)$) -- ($(temp) + (width)$) -- ($(temp) + (width) + (height)$) -- ($(temp) + (height)$) -- ($(temp)$);
\node (label) at ($(temp) + (labelshift)$)  {$c_{2}$};

\coordinate (temp) at (6);
\draw[black] ($(temp)$) -- ($(temp) + (width)$) -- ($(temp) + (width) + (height)$) -- ($(temp) + (height)$) -- ($(temp)$);
\node (label) at ($(temp) + (labelshift)$)  {$c_{1}^1$};

\coordinate (temp) at (7);
\draw[black] ($(temp)$) -- ($(temp) + (width)$) -- ($(temp) + (width) + (height)$) -- ($(temp) + (height)$) -- ($(temp)$);
\node (label) at ($(temp) + (labelshift)$)  {$c_{1}^2$};

\coordinate (temp) at (10);
\draw[black] ($(temp)$) -- ($(temp) + (width)$) -- ($(temp) + (width) + (height)$) -- ($(temp) + (height)$) -- ($(temp)$);
\node (label) at ($(temp) + (labelshift)$)  {$c_{2m-1}$};

\coordinate (temp) at (11);
\draw[black] ($(temp)$) -- ($(temp) + (width)$) -- ($(temp) + (width) + (height)$) -- ($(temp) + (height)$) -- ($(temp)$);
\node (label) at ($(temp) + (labelshift)$)  {$c_{2m}$};

\coordinate (temp) at (12);
\draw[black] ($(temp)$) -- ($(temp) + (width)$) -- ($(temp) + (width) + (height)$) -- ($(temp) + (height)$) -- ($(temp)$);
\node (label) at ($(temp) + (labelshift)$)  {$c_{2m-1}^1$};

\coordinate (temp) at (13);
\draw[black] ($(temp)$) -- ($(temp) + (width)$) -- ($(temp) + (width) + (height)$) -- ($(temp) + (height)$) -- ($(temp)$);
\node (label) at ($(temp) + (labelshift)$)  {$c_{1m-1}^2$};

\coordinate (temp) at (14);
\draw[black] ($(temp)$) -- ($(temp) + (width)$) -- ($(temp) + (width) + (height)$) -- ($(temp) + (height)$) -- ($(temp)$);
\node (label) at ($(temp) + (labelshift)$)  {$x_1$};

\coordinate (temp) at (15);
\draw[black] ($(temp)$) -- ($(temp) + (width)$) -- ($(temp) + (width) + (height)$) -- ($(temp) + (height)$) -- ($(temp)$);
\node (label) at ($(temp) + (labelshift)$)  {$x_1^1$};

\coordinate (temp) at (16);
\draw[black] ($(temp)$) -- ($(temp) + (width)$) -- ($(temp) + (width) + (height)$) -- ($(temp) + (height)$) -- ($(temp)$);
\node (label) at ($(temp) + (labelshift)$)  {$x_1^2$};

\coordinate (temp) at (101);
\draw[black] ($(temp)$) -- ($(temp) + (width)$) -- ($(temp) + (width) + (height)$) -- ($(temp) + (height)$) -- ($(temp)$);
\node (label) at ($(temp) + (labelshift)$)  {$x_{n-1}$};

\coordinate (temp) at (111);
\draw[black] ($(temp)$) -- ($(temp) + (width)$) -- ($(temp) + (width) + (height)$) -- ($(temp) + (height)$) -- ($(temp)$);
\node (label) at ($(temp) + (labelshift)$)  {$x_n$};

\coordinate (temp) at (121);
\draw[black] ($(temp)$) -- ($(temp) + (width)$) -- ($(temp) + (width) + (height)$) -- ($(temp) + (height)$) -- ($(temp)$);
\node (label) at ($(temp) + (labelshift)$)  {$x_n^1$};

\coordinate (temp) at (131);
\draw[black] ($(temp)$) -- ($(temp) + (width)$) -- ($(temp) + (width) + (height)$) -- ($(temp) + (height)$) -- ($(temp)$);
\node (label) at ($(temp) + (labelshift)$)  {$x_n^2$};

\coordinate (temp) at (141);
\draw[black] ($(temp)$) -- ($(temp) + (width)$) -- ($(temp) + (width) + (height)$) -- ($(temp) + (height)$) -- ($(temp)$);
\node (label) at ($(temp) + (labelshift)$)  {$x_{n+1}$};

\coordinate (temp) at (151);
\draw[black] ($(temp)$) -- ($(temp) + (width)$) -- ($(temp) + (width) + (height)$) -- ($(temp) + (height)$) -- ($(temp)$);
\node (label) at ($(temp) + (labelshift)$)  {$x_{n+1}^1$};

\coordinate (temp) at (161);
\draw[black] ($(temp)$) -- ($(temp) + (width)$) -- ($(temp) + (width) + (height)$) -- ($(temp) + (height)$) -- ($(temp)$);
\node (label) at ($(temp) + (labelshift)$)  {$x_{n+1}^2$};

\node[draw=none]  at ($(8)$) {$\cdots$};
\node[draw=none]  at ($(9)$) {$\cdots$};
\node[draw=none]  at ($(17)$) {$\cdots$};
\node[draw=none]  at ($(18)$) {$\cdots$};
\end{tikzpicture}
\caption{Illustrations for the proof of Lemma~\ref{final}.}
\label{backboneApp}
\end{center}
\end{figure}

Now, for some $i$, $1 \leq i \leq n$, let $l_{j_1}^{r_1}, \dots, l_{j_q}^{r_q}$ be exactly the literal vertices corresponding to occurrences of literal $x_i$. By definition, these vertices form a path in this order and the structure of the backbone implies that the $x$-coordinates of their corresponding unit squares differ by at least two, which means that the visibilities between the unit-squares for $l_{j_1}^{r_1},\dots,l_{j_q}^{r_q}$ are all horizontal; thus, they form a horizontal path in this order and are all on the same side of the backbone. By definition, both vertices $\overset{_{\rightarrow}}{t_i}$ and $\overset{_{\leftarrow}}{t_i}$ are adjacent to all vertices $l_{j_1}^{r_1},\dots,l_{j_q}^{r_q}$. Since every literal of the formula has at least three occurrences, i.\,e., $q \geq 3$, the only possibility to place unit squares for $\overset{_{\rightarrow}}{t_i}$ and $\overset{_{\leftarrow}}{t_i}$ in order to see every unit square of the path is horizontally from opposite sides, i.\,e., either $R_{\overset{_{\rightarrow}}{t_i}} \Hvis R_{\{l_{j_1}^{r_1}, \dots, l_{j_q}^{r_q}\}}$ and $R_{\{l_{j_1}^{r_1}, \dots, l_{j_q}^{r_q}\}} \Hvis R_{\overset{_{\leftarrow}}{t_i}}$ or $R_{\overset{_{\leftarrow}}{t_i}} \Hvis R_{\{l_{j_1}^{r_1}, \dots, l_{j_q}^{r_q}\}}$ and $R_{\{l_{j_1}^{r_1}, \dots, l_{j_q}^{r_q}\}} \Hvis R_{\overset{_{\rightarrow}}{t_i}}$; since these two cases can be handled analogously, we assume the former. Since the horizontal distance between $R_{\overset{_{\rightarrow}}{t_i}}$ and $R_{\overset{_{\leftarrow}}{t_i}}$ is more than one unit, the unit squares for all other mutual neighbours of $\overset{_{\rightarrow}}{t_i}$ and $\overset{_{\leftarrow}}{t_i}$, i.\,e., the vertices $h_{t_i}^r$, $0 \leq r \leq 4$, and $t_i$, must also be placed horizontally in between $R_{\overset{_{\rightarrow}}{t_i}}$ and $R_{\overset{_{\leftarrow}}{t_i}}$. In particular, this implies that $R_{t_i}$ has to be placed on the same side as the path $R_{l_{j_1}^{r_1}},\dots,R_{l_{j_q}^{r_q}}$ with respect to the backbone. The same arguments hold for the literals corresponding to occurrences of a negated variable $\overline{x_i}$ connected to $f^1_{i}$ or $f^2_{i}$.\par
We now define an assignment $\sigma : \{x_1, x_2, \ldots, x_n\} \to \{\emph{true}, \emph{false}\}$ as follows. For every $i$, $1 \leq i \leq n$, we define $\sigma(x_i) = \emph{true}$ if and only if $R_{x_i}\Vvis R_{t_i}$. We claim that this assignment is a satisfying not-all-equal assignment for the formula $F$. To this end, let $c_j = \{y_{j,1},y_{j,2},y_{j,3}\}$ be an arbitrary clause of $F$. We first assume that $c_j$ does not contain any negated variable, i.\,e., $c_j = \{x_{\ell_1}, x_{\ell_2}, x_{\ell_2}\}$. Due to Lemma~\ref{between_all} it is not possible that $R_{L_j} \Vvis R_{c_j}$ or $R_{c_j} \Vvis R_{L_j}$, which implies that at least one of $R_{\{t_{\ell_1}, t_{\ell_2}, t_{\ell_3}\}}$ is placed below the backbone and at least one of them is placed above the backbone (since, as explained above, they are placed at the same side as their corresponding unit square from $R_{L_j}$). Consequently, at least one literal of $c_j$ is set to \emph{true} and at least one is set to \emph{false}.\par
Next, we assume that $c_j$ contains a negated variable and we recall that, by definition of $F$, there is at most one negated variable in $c_j$. Without loss of generality, we assume that $y_{j,1} = \overline{x_i}$ and we let $r\in\{1,2\}$ be such that $R_{t_i}$ and $R_{f_i^r}$ are on opposite sides of the backbone; note that such $r$ must exist, since otherwise $R_{A_i} \Vvis R_{x_i}$ or $R_{x_i} \Vvis R_{A_i}$, which contradicts Lemma~\ref{between_all}. We first assume that $r = 1$. If $R_{l^1_j} \Vvis R_{c_j}$, then also $R_{f^r_i} \Vvis R_{x_i}$, which, by definition of $r$, implies $R_{x_i} \Vvis R_{t_i}$; thus $\sigma(x_i) = \emph{true}$. Moreover, we must have $R_{c_j} \Vvis R_{l^2_j}$ or $R_{c_j} \Vvis R_{l^3_j}$ (since otherwise $R_{L_j} \Vvis R_{c_j}$), which means that $\sigma(y_{j,2}) = \emph{true}$ or $\sigma(y_{j,3}) = \emph{true}$. The case that $R_{c_j} \Vvis R_{l^1_j}$ analogously leads to $\sigma(x_i) = \emph{false}$ and either $\sigma(y_{j,2}) = \emph{false}$ or $\sigma(y_{j,3}) = \emph{false}$. Hence, $c_j$ is not-all-equal satisfied. In the case $r = 2$, we can apply the same argument, but with respect to $c_{j + m}, l^1_{j + m}, l^2_{j + m}, l^3_{j + m}$ instead of $c_{j}, l^1_{j}, l^2_{j}, l^3_{j}$. Since clause $c_{j + m}$ is a copy of $c_{j}$, we can again conclude that $c_j$ is not-all-equal satisfied.
\end{proof}

\subsection*{Proof of Theorem~\ref{nonGridNPCompletenessTheorem}}
 
\begin{proof} 
Due to Theorem~\ref{NPMembershipTheorem}, it only remains to show $\npclass$-hardness, which follows from Lemmas~\ref{nonGridReductionEasyDirectionLemma}~and~\ref{final}.
\end{proof}

\end{document}

%% file: reduction_figure.tex
\begin{tikzpicture}[scale=0.208,transform shape]
\Large
\coordinate (width) at (1,0);
\coordinate (height) at (0,1);
\coordinate (labelshift) at (0.5, 0.5);

%DEFINE COORDINATES

\foreach \y in {0,...,6}{
\coordinate (Rc_{\y}) at (3*\y,0);
}

\foreach \y in {0,...,5}{
\foreach \x in {1,2}{
\coordinate (Rc_{\y,\x}) at (3*\y+1.5,2-1.3*\x);
}
}

\foreach \y in {1,...,4}{
\coordinate (Rx_{\y}) at (7*\y+2*3*3,0);
}
\coordinate (Rx_{5}) at (7*4+2*3*3+3,0);

\foreach \y in {1,...,5}{
\foreach \x in {1,2}{
\coordinate (Rx_{\y,\x}) at (3*2*3+7*\y-5.5,2-1.3*\x);
}
}

%variables with assignment false
\foreach \y in {2,3}{
\coordinate (Rt_{\y}) at (7*\y+2*3*3,-5*\y-2+1-10/13);
\foreach \x in {1,2}{
\coordinate (Rf_{\y, \x}) at (7*\y+2*3*3+3/2-\x,5*\y+2*\x+1-10/13);
}
}

%variables with assignment true
\foreach \y in {1,4}{
\coordinate (Rt_{\y}) at (7*\y+2*3*3,5*\y+2+1-10/14);
\foreach \x in {1,2}{
\coordinate (Rf_{\y, \x}) at (7*\y+2*3*3+3/2-\x,-5*\y-2*\x+1-10/13);
}
}
%%%%%%%%%%%%%%%%%%%%%%%%%%%%%%%%%%%%%%%%%%%%%%%%%%%%%%%%%%%%%%%%%%%%%%%%%%%%%%%%%%%%%
%literals
%%%%%%%%%%%%%%%%%%%%%%%%%%%%%%%%%%%%%%%%%%%%%%%%%%%%%%%%%%%%%%%%%%%%%%%%%%%%%%%%%%%%%%
% clauses with only one satisfied literal
\foreach \x/\y/\z in {2/1/1,5/1/1}{
\coordinate (Rl_{\y, \x}) at (3*\x,5*\z+2+1-\x/13-2/13);
}

% clauses with only one unsatisfied literal
\foreach \x/\y/\z in {1/3/3,3/2/3,4/3/3,6/2/3}{
\coordinate (Rl_{\y, \x}) at (3*\x,-5*\z-2+1-\x/13-2/13);
}

%satisfied literal, shiftet to the left
\foreach \x/\y/\z in {1/1/1, 3/1/2,4/1/1}{
\coordinate (Rl_{\y, \x}) at (3*\x-0.5, 5*\z+2+1-\x/13-2/13);
}
\foreach \x/\y/\z in {6/1/2}{
\coordinate (Rl_{\y, \x}) at (3*\x-0.5, 5*\z+4+1-\x/13-2/13 );
}

%unsatisfied literal, shifted to the left
\foreach \x/\y/\z in {2/2/3,5/2/3}{
\coordinate (Rl_{\y, \x}) at (3*\x-0.5, -5*\z-2+1-\x/13-2/13);
}

%satisfied literal, shiftet to the left
\foreach \x/\y/\z in {1/2/2, 3/3/4,6/3/4}{
\coordinate (Rl_{\y, \x}) at (3*\x+0.5, 5*\z+2+1-\x/13-2/13);
}
\foreach \x/\y/\z in {4/2/2}{
\coordinate (Rl_{\y, \x}) at (3*\x+0.5, 5*\z+4+1-\x/13-2/13 );
}

%usatisfied literal, shiftet to the right
\foreach \x/\y/\z in {2/3/4}{
\coordinate (Rl_{\y, \x}) at (3*\x+0.5, -5*\z-2+1-\x/13-2/13 );
}
\foreach \x/\y/\z in {5/3/4}{
\coordinate (Rl_{\y, \x}) at (3*\x+0.5, -5*\z-4+1-\x/13-2/13 );
}

%%%%%%%%%%%%%%%%%%%%%%%%%%%%%%%%%%%%%%%%%%%%%%%%%%%%%%%%%%%%%%%%%%%%%%%%%%%%%%%%%%%%%
%arrow-vertices
%%%%%%%%%%%%%%%%%%%%%%%%%%%%%%%%%%%%%%%%%%%%%%%%%%%%%%%%%%%%%%%%%%%%%%%%%%%%%%%%%%%%%%

%satisfied arrow-t_i

\foreach \y in {1,4}{
\coordinate (Rt_{\y, rightarrow}) at (-9*\y +2,5*\y+2);
\coordinate (Rt_{\y, leftarrow}) at (6*3+7*5+9*\y-4,5*\y+3);
}
%unsatisfied arrow-t_i
\foreach \y in {2,3}{
\coordinate (Rt_{\y, rightarrow}) at (-9*\y-4 +2,-5*\y-2);
\coordinate (Rt_{\y, leftarrow}) at (6*3+7*5+9*\y-4,-5*\y-1);
}

%satisfied arrow-f_i^r
\foreach \y in {2,3}{
\foreach \x in {1,2}{
\coordinate (Rf_{\y, \x, rightarrow}) at (-9*\y-\x +2,5*\y+2*\x);
\coordinate (Rf_{\y, \x, leftarrow}) at (6*3+7*5+9*\y-\x-4,5*\y+2*\x+1);
}
}

%unsatisfied arrow-f_i^r
\foreach \y in {1,4}{
\foreach \x in {1,2}{
\coordinate (Rf_{\y, \x, rightarrow}) at (-9*\y-\x +2,-5*\y-2*\x);
\coordinate (Rf_{\y, \x, leftarrow}) at (6*3+7*5+9*\y-\x-4,-5*\y-2*\x+1);
}
}

%%%%%%%%%%%%%%%%%%%%%%%%%%%%%%%%%%%%%%%%%%%%%%%%%%%%%%%%%%%%%%%%%%%%%%%%%%%%%%%%%%%%%
%h-vertices
%%%%%%%%%%%%%%%%%%%%%%%%%%%%%%%%%%%%%%%%%%%%%%%%%%%%%%%%%%%%%%%%%%%%%%%%%%%%%%%%%%%%%%

%satisfied h^0_t_i
\foreach \y in {1,4}{
\coordinate (Rht_{\y,0}) at (6*3+7*\y-3,5*\y+3-9/13 );
}

%unsatisfied h^0_t_i
\foreach \y in {2,3}{
\coordinate (Rht_{\y,0}) at (6*3+7*\y-3,-5*\y-1-9/13 );
}
%satisfied h^0_f_i^r
\foreach \y in {2,3}{
\foreach \x in {1,2}{
\coordinate (Rhf_{\y,\x,0}) at (6*3+7*\y-2*\x,5*\y+2*\x+1-9/13);
}}

%unsatisfied h^0_f_i^r
\foreach \y in {1,4}{
\foreach \x in {1,2}{
\coordinate (Rhf_{\y,\x,0}) at (6*3+7*\y-2*\x,-5*\y-2*\x+1-9/13 );
}}

%satisfied h^r_t_i r=1,2
\foreach \y in {1,4}{
\foreach \r in {1,2}{
\coordinate (Rht_{\y,\r}) at (-9*\y+3*\r +2,5*\y+2+1-\r/13 );
}}

%unsatisfied  h^r_t_i r=1,2
\foreach \y in {2,3}{
\foreach \r in {1,2}{
\coordinate (Rht_{\y,\r}) at (-9*\y+3*\r +2,-5*\y-2+1-\r/13 );
}}
%satisfied h^r_f_i^x r=1,2
\foreach \y in {2,3}{
\foreach \r in {1,2}{
\foreach \x in {1,2}{
\coordinate (Rhf_{\y,\x,\r}) at (-9*\y+3*\r-\x  +2, 5*\y+2*\x+1-\r/13 );
}}}

%unsatisfied h^r_f_i_x r=1,2
\foreach \y in {1,4}{
\foreach \r in {1,2}{
\foreach \x in {1,2}{
\coordinate (Rhf_{\y,\x,\r}) at (-9*\y+3*\r-\x +2, -5*\y-2*\x+1-\r/13 );
}}}

%satisfied h^r_t_i r=3,4
\foreach \y in {1,4}{
\foreach \r in {3,4}{
\coordinate (Rht_{\y,\r}) at (6*3+7*5+9*\y+3*\r-19,5*\y+2 + 5/13-\r/13);
}}

%unsatisfied  h^r_t_i r=3,4
\foreach \y in {2,3}{
\foreach \r in {3,4}{
\coordinate (Rht_{\y,\r}) at (6*3+7*5+9*\y+3*\r-19,-5*\y-2 + 5/13 - \r/13);
}}
%satisfied h^r_f_i^x r=3,4
\foreach \y in {2,3}{
\foreach \r in {3,4}{
\foreach \x in {1,2}{
\coordinate (Rhf_{\y,\x,\r}) at (6*3+7*5+9*\y-\x+3*\r-19, 5*\y+2*\x + 5/13 - \r/13);
}}}

%unsatisfied h^r_f_i_x r=3,4
\foreach \y in {1,4}{
\foreach \r in {3,4}{
\foreach \x in {1,2}{
\coordinate (Rhf_{\y,\x,\r}) at (6*3+7*5+9*\y-\x+3*\r-19, -5*\y-2*\x + 5/13 - \r/13);
}}}

%DRAW VISIBILITIES:

%\coordinate (max_x_value) at (Rf_{4,2,leftarrow});
%\coordinate (min_x_value) at (Rf_{4,2,rightarrow});
\coordinate (max_x_value) at (Rt_{4, leftarrow});
\coordinate (min_x_value) at (Rf_{4, 2, rightarrow});
\coordinate (max_y_value) at ($(Rt_{4, leftarrow}) + (0,1)$);
\coordinate (min_y_value) at (Rf_{4, 2, rightarrow});

%draw vertical visibilities for clauses
\foreach \y in {1,...,6}{
\draw[fill, black!5] ($(Rc_{\y})$) -- (Rc_{\y} |- max_y_value) -- ($(Rc_{\y} |- max_y_value) + (width)$) -- ($(Rc_{\y}) + (width)$) -- ($(Rc_{\y}) + (width)$);
\draw[dashed, very thin] ($(Rc_{\y})$) -- (Rc_{\y} |- max_y_value) -- ($(Rc_{\y} |- max_y_value) + (width)$) -- ($(Rc_{\y}) + (width)$) -- ($(Rc_{\y}) + (width)$);

\draw[fill, black!5] ($(Rc_{\y})$) -- (Rc_{\y} |- min_y_value) -- ($(Rc_{\y} |- min_y_value) + (width)$) -- ($(Rc_{\y}) + (width)$) -- ($(Rc_{\y}) + (width)$);
\draw[dashed, very thin] ($(Rc_{\y})$) -- (Rc_{\y} |- min_y_value) -- ($(Rc_{\y} |- min_y_value) + (width)$) -- ($(Rc_{\y}) + (width)$) -- ($(Rc_{\y}) + (width)$);
}

%draw vertical visibilities for variables
\foreach \y in {1,...,4}{

\draw[fill, black!5] ($(Rx_{\y})$) -- (Rx_{\y} |- max_y_value) -- ($(Rx_{\y} |- max_y_value) + (width)$) -- ($(Rx_{\y}) + (width)$) -- ($(Rx_{\y}) + (width)$);
\draw[dashed, very thin] ($(Rx_{\y})$) -- (Rx_{\y} |- max_y_value) -- ($(Rx_{\y} |- max_y_value) + (width)$) -- ($(Rx_{\y}) + (width)$) -- ($(Rx_{\y}) + (width)$);

\draw[fill, black!5] ($(Rx_{\y})$) -- (Rx_{\y} |- min_y_value) -- ($(Rx_{\y} |- min_y_value) + (width)$) -- ($(Rx_{\y}) + (width)$) -- ($(Rx_{\y}) + (width)$);
\draw[dashed, very thin] ($(Rx_{\y})$) -- (Rx_{\y} |- min_y_value) -- ($(Rx_{\y} |- min_y_value) + (width)$) -- ($(Rx_{\y}) + (width)$) -- ($(Rx_{\y}) + (width)$);
}

%draw horizontal visibilities for some literals
\draw[fill, black!5] ($(Rt_{3, rightarrow}) + (width)$) -- (max_x_value |- Rt_{3, rightarrow}) -- ($(max_x_value |- Rt_{3, rightarrow}) + (height)$) -- ($(Rt_{3, rightarrow}) + (height) + (width)$) -- ($(Rt_{3, rightarrow}) + (width)$);
\draw[fill, black!5] ($(Rt_{3, leftarrow})$) -- (min_x_value |- Rt_{3, leftarrow}) -- ($(min_x_value |- Rt_{3, leftarrow}) + (height)$) -- ($(Rt_{3, leftarrow}) + (height)$) -- ($(Rt_{3, leftarrow})$);
\draw[dashed, very thin] ($(Rt_{3, rightarrow}) + (width)$) -- (max_x_value |- Rt_{3, rightarrow}) -- ($(max_x_value |- Rt_{3, rightarrow}) + (height)$) -- ($(Rt_{3, rightarrow}) + (height) + (width)$) -- ($(Rt_{3, rightarrow}) + (width)$);
%\draw[dashed, very thin] ($(Rt_{3, leftarrow})$) -- (min_x_value |- Rt_{3, leftarrow}) -- ($(min_x_value |- Rt_{3, leftarrow}) + (height)$) -- ($(Rt_{3, leftarrow}) + (height)$) -- ($(Rt_{3, leftarrow})$);
\draw[dashed, very thin] ($(Rt_{3,  rightarrow})+(height)$) --($(min_x_value |- Rt_{3, leftarrow})$) --($(min_x_value |- Rt_{3, leftarrow}) + (height)$) -- ($(Rt_{3, leftarrow}) + (height)$);

\draw[fill, black!5] ($(Rt_{1, rightarrow}) + (width)$) -- (max_x_value |- Rt_{1, rightarrow}) -- ($(max_x_value |- Rt_{1, rightarrow}) + (height)$) -- ($(Rt_{1, rightarrow}) + (height) + (width)$) -- ($(Rt_{1, rightarrow}) + (width)$);
\draw[fill, black!5] ($(Rt_{1, leftarrow})$) -- (min_x_value |- Rt_{1, leftarrow}) -- ($(min_x_value |- Rt_{1, leftarrow}) + (height)$) -- ($(Rt_{1, leftarrow}) + (height)$) -- ($(Rt_{1, leftarrow})$);
\draw[dashed, very thin] ($(Rt_{1, rightarrow}) + (width)$) -- (max_x_value |- Rt_{1, rightarrow}) -- ($(max_x_value |- Rt_{1, rightarrow}) + (height)$) -- ($(Rt_{1, rightarrow}) + (height) + (width)$) -- ($(Rt_{1, rightarrow}) + (width)$);
%\draw[dashed, very thin] ($(Rt_{1, leftarrow})$) -- (min_x_value |- Rt_{1, leftarrow}) -- ($(min_x_value |- Rt_{1, leftarrow}) + (height)$) -- ($(Rt_{1, leftarrow}) + (height)$) -- ($(Rt_{1, leftarrow})$);
\draw[dashed, very thin] ($(Rt_{1, rightarrow})+(height)$) -- ($(min_x_value |- Rt_{1, leftarrow}) $) -- ($(min_x_value |- Rt_{1, leftarrow}) + (height)$) -- ($(Rt_{1, leftarrow}) + (height)$);

\draw[fill, black!5] ($(Rf_{2, 2, rightarrow}) + (width)$) -- (max_x_value |- Rf_{2, 2, rightarrow}) -- ($(max_x_value |- Rf_{2, 2, rightarrow}) + (height)$) -- ($(Rf_{2, 2, rightarrow}) + (height) + (width)$) -- ($(Rf_{2, 2, rightarrow}) + (width)$);
\draw[fill, black!5] ($(Rf_{2, 2, leftarrow})$) -- (min_x_value |- Rf_{2, 2, leftarrow}) -- ($(min_x_value |- Rf_{2, 2, leftarrow}) + (height)$) -- ($(Rf_{2, 2, leftarrow}) + (height)$) -- ($(Rf_{2, 2, leftarrow})$);
\draw[dashed, very thin] ($(Rf_{2, 2, rightarrow}) + (width)$) -- (max_x_value |- Rf_{2, 2, rightarrow}) -- ($(max_x_value |- Rf_{2, 2, rightarrow}) + (height)$) -- ($(Rf_{2, 2, rightarrow}) + (height) + (width)$) -- ($(Rf_{2, 2, rightarrow}) + (width)$);
%\draw[dashed, very thin] ($(Rf_{2, 2, leftarrow})$) -- (min_x_value |- Rf_{2, 2, leftarrow}) -- ($(min_x_value |- Rf_{2, 2, leftarrow}) + (height)$) -- ($(Rf_{2, 2, leftarrow}) + (height)$) -- ($(Rf_{2, 2, leftarrow})$);
\draw[dashed, very thin] ($(Rf_{2, 2, rightarrow})+(height)$) --($(min_x_value |- Rf_{2, 2, leftarrow}) $) --($(min_x_value |- Rf_{2, 2, leftarrow}) + (height)$) -- ($(Rf_{2, 2, leftarrow}) + (height)$);

%\draw[fill, black!5] ($(Rf_{1, 1, rightarrow}) + (width)$) -- (max_x_value |- Rf_{1, 1, rightarrow}) -- ($(max_x_value |- Rf_{1, 1, rightarrow}) + (height)$) -- ($(Rf_{1, 1, rightarrow}) + (height) + (width)$) -- ($(Rf_{1, 1, rightarrow}) + (width)$);
%\draw[fill, black!5] ($(Rf_{1, 1, leftarrow})$) -- (min_x_value |- Rf_{1, 1, leftarrow}) -- ($(min_x_value |- Rf_{1, 1, leftarrow}) + (height)$) -- ($(Rf_{1, 1, leftarrow}) + (height)$) -- ($(Rf_{1, 1, leftarrow})$);
%\draw[dashed, very thin] ($(Rf_{1, 1, rightarrow}) + (width)$) -- (max_x_value |- Rf_{1, 1, rightarrow}) -- ($(max_x_value |- Rf_{1, 1, rightarrow}) + (height)$) -- ($(Rf_{1, 1, rightarrow}) + (height) + (width)$) -- ($(Rf_{1, 1, rightarrow}) + (width)$);
%\draw[dashed, very thin] ($(Rf_{1, 1, leftarrow})$) -- (min_x_value |- Rf_{1, 1, leftarrow}) -- ($(min_x_value |- Rf_{1, 1, leftarrow}) + (height)$) -- ($(Rf_{1, 1, leftarrow}) + (height)$) -- ($(Rf_{1, 1, leftarrow})$);

%DRAW SQUARES

\foreach \y in {0,...,6}{
\draw[fill, black!20] ($(Rc_{\y})$) -- ($(Rc_{\y}) + (width)$) -- ($(Rc_{\y}) + (width) + (height)$) -- ($(Rc_{\y}) + (height)$) -- ($(Rc_{\y})$);
\draw[black] ($(Rc_{\y})$) -- ($(Rc_{\y}) + (width)$) -- ($(Rc_{\y}) + (width) + (height)$) -- ($(Rc_{\y}) + (height)$) -- ($(Rc_{\y})$);
\node (label) at ($(Rc_{\y}) + (labelshift)$)  {$c_{\y}$};
}

\foreach \y in {0,...,5}{
\foreach \x in {1,2}{
\draw[fill, black!20] ($(Rc_{\y,\x})$) -- ($(Rc_{\y,\x}) + (width)$) -- ($(Rc_{\y,\x}) + (width) + (height)$) -- ($(Rc_{\y,\x}) + (height)$) -- ($(Rc_{\y,\x})$);
\draw[black] ($(Rc_{\y,\x})$) -- ($(Rc_{\y,\x}) + (width)$) -- ($(Rc_{\y,\x}) + (width) + (height)$) -- ($(Rc_{\y,\x}) + (height)$) -- ($(Rc_{\y,\x})$);
\node (label) at ($(Rc_{\y,\x}) + (labelshift)$)  {$c_{\y}^{\x}$};
}
}

\foreach \y in {1,...,5}{
\draw[fill, black!20] ($(Rx_{\y})$) -- ($(Rx_{\y}) + (width)$) -- ($(Rx_{\y}) + (width) + (height)$) -- ($(Rx_{\y}) + (height)$) -- ($(Rx_{\y})$);
\draw[black] ($(Rx_{\y})$) -- ($(Rx_{\y}) + (width)$) -- ($(Rx_{\y}) + (width) + (height)$) -- ($(Rx_{\y}) + (height)$) -- ($(Rx_{\y})$);
\node (label) at ($(Rx_{\y}) + (labelshift)$)  {$x_{\y}$};
}

\foreach \y in {1,...,5}{
\foreach \x in {1,2}{
\draw[fill, black!20] ($(Rx_{\y,\x})$) -- ($(Rx_{\y,\x}) + (width)$) -- ($(Rx_{\y,\x}) + (width) + (height)$) -- ($(Rx_{\y,\x}) + (height)$) -- ($(Rx_{\y,\x})$);
\draw[black] ($(Rx_{\y,\x})$) -- ($(Rx_{\y,\x}) + (width)$) -- ($(Rx_{\y,\x}) + (width) + (height)$) -- ($(Rx_{\y,\x}) + (height)$) -- ($(Rx_{\y,\x})$);
\node (label) at ($(Rx_{\y,\x}) + (labelshift)$)  {$x_{\y}^{\x}$};
}
}

%variables with assignment false
\foreach \y in {2,3}{
\draw[fill, black!20] ($(Rt_{\y})$) -- ($(Rt_{\y}) + (width)$) -- ($(Rt_{\y}) + (width) + (height)$) -- ($(Rt_{\y}) + (height)$) -- ($(Rt_{\y})$);
\draw[black] ($(Rt_{\y})$) -- ($(Rt_{\y}) + (width)$) -- ($(Rt_{\y}) + (width) + (height)$) -- ($(Rt_{\y}) + (height)$) -- ($(Rt_{\y})$);
\node (label) at ($(Rt_{\y}) + (labelshift)$)  {$t_{\y}$};
\foreach \x in {1,2}{
\draw[fill, black!20] ($(Rf_{\y, \x})$) -- ($(Rf_{\y, \x}) + (width)$) -- ($(Rf_{\y, \x}) + (width) + (height)$) -- ($(Rf_{\y, \x}) + (height)$) -- ($(Rf_{\y, \x})$);
\draw[black] ($(Rf_{\y, \x})$) -- ($(Rf_{\y, \x}) + (width)$) -- ($(Rf_{\y, \x}) + (width) + (height)$) -- ($(Rf_{\y, \x}) + (height)$) -- ($(Rf_{\y, \x})$);
\node (label) at ($(Rf_{\y, \x}) + (labelshift)$)  {$f_{\y}^{\x}$};
}
}

%variables with assignment true
\foreach \y in {1,4}{
\draw[fill, black!20] ($(Rt_{\y})$) -- ($(Rt_{\y}) + (width)$) -- ($(Rt_{\y}) + (width) + (height)$) -- ($(Rt_{\y}) + (height)$) -- ($(Rt_{\y})$);
\draw[black] ($(Rt_{\y})$) -- ($(Rt_{\y}) + (width)$) -- ($(Rt_{\y}) + (width) + (height)$) -- ($(Rt_{\y}) + (height)$) -- ($(Rt_{\y})$);
\node (label) at ($(Rt_{\y}) + (labelshift)$)  {$t_{\y}$};
\foreach \x in {1,2}{
\draw[fill, black!20] ($(Rf_{\y, \x})$) -- ($(Rf_{\y, \x}) + (width)$) -- ($(Rf_{\y, \x}) + (width) + (height)$) -- ($(Rf_{\y, \x}) + (height)$) -- ($(Rf_{\y, \x})$);
\draw[black] ($(Rf_{\y, \x})$) -- ($(Rf_{\y, \x}) + (width)$) -- ($(Rf_{\y, \x}) + (width) + (height)$) -- ($(Rf_{\y, \x}) + (height)$) -- ($(Rf_{\y, \x})$);
\node (label) at ($(Rf_{\y, \x}) + (labelshift)$)  {$f_{\y}^{\x}$};
}
}
%%%%%%%%%%%%%%%%%%%%%%%%%%%%%%%%%%%%%%%%%%%%%%%%%%%%%%%%%%%%%%%%%%%%%%%%%%%%%%%%%%%%%
%literals
%%%%%%%%%%%%%%%%%%%%%%%%%%%%%%%%%%%%%%%%%%%%%%%%%%%%%%%%%%%%%%%%%%%%%%%%%%%%%%%%%%%%%%
% clauses with only one satisfied literal
\foreach \x/\y/\z in {2/1/1,5/1/1}{
\draw[fill, black!20] ($(Rl_{\y, \x})$) -- ($(Rl_{\y, \x}) + (width)$) -- ($(Rl_{\y, \x}) + (width) + (height)$) -- ($(Rl_{\y, \x}) + (height)$) -- ($(Rl_{\y, \x})$);
\draw[black] ($(Rl_{\y, \x})$) -- ($(Rl_{\y, \x}) + (width)$) -- ($(Rl_{\y, \x}) + (width) + (height)$) -- ($(Rl_{\y, \x}) + (height)$) -- ($(Rl_{\y, \x})$);
\node (label) at ($(Rl_{\y, \x}) + (labelshift)$)  {$l_{\x}^{\y}$};
}

% clauses with only one unsatisfied literal
\foreach \x/\y/\z in {1/3/3,3/2/3,4/3/3,6/2/3}{
\draw[fill, black!20] ($(Rl_{\y, \x})$) -- ($(Rl_{\y, \x}) + (width)$) -- ($(Rl_{\y, \x}) + (width) + (height)$) -- ($(Rl_{\y, \x}) + (height)$) -- ($(Rl_{\y, \x})$);
\draw[black] ($(Rl_{\y, \x})$) -- ($(Rl_{\y, \x}) + (width)$) -- ($(Rl_{\y, \x}) + (width) + (height)$) -- ($(Rl_{\y, \x}) + (height)$) -- ($(Rl_{\y, \x})$);
\node (label) at ($(Rl_{\y, \x}) + (labelshift)$)  {$l_{\x}^{\y}$};
}

%satisfied literal, shiftet to the left
\foreach \x/\y/\z in {1/1/1, 3/1/2,4/1/1}{
\draw[fill, black!20] ($(Rl_{\y, \x})$) -- ($(Rl_{\y, \x}) + (width)$) -- ($(Rl_{\y, \x}) + (width) + (height)$) -- ($(Rl_{\y, \x}) + (height)$) -- ($(Rl_{\y, \x})$);
\draw[black] ($(Rl_{\y, \x})$) -- ($(Rl_{\y, \x}) + (width)$) -- ($(Rl_{\y, \x}) + (width) + (height)$) -- ($(Rl_{\y, \x}) + (height)$) -- ($(Rl_{\y, \x})$);
\node (label) at ($(Rl_{\y, \x}) + (labelshift)$)  {$l_{\x}^{\y}$};
}
\foreach \x/\y/\z in {6/1/2}{
\draw[fill, black!20] ($(Rl_{\y, \x})$) -- ($(Rl_{\y, \x}) + (width)$) -- ($(Rl_{\y, \x}) + (width) + (height)$) -- ($(Rl_{\y, \x}) + (height)$) -- ($(Rl_{\y, \x})$);
\draw[black] ($(Rl_{\y, \x})$) -- ($(Rl_{\y, \x}) + (width)$) -- ($(Rl_{\y, \x}) + (width) + (height)$) -- ($(Rl_{\y, \x}) + (height)$) -- ($(Rl_{\y, \x})$);
\node (label) at ($(Rl_{\y, \x}) + (labelshift)$)  {$l_{\x}^{\y}$};
}

%unsatisfied literal, shifted to the left
\foreach \x/\y/\z in {2/2/3,5/2/3}{
\draw[fill, black!20] ($(Rl_{\y, \x})$) -- ($(Rl_{\y, \x}) + (width)$) -- ($(Rl_{\y, \x}) + (width) + (height)$) -- ($(Rl_{\y, \x}) + (height)$) -- ($(Rl_{\y, \x})$);
\draw[black] ($(Rl_{\y, \x})$) -- ($(Rl_{\y, \x}) + (width)$) -- ($(Rl_{\y, \x}) + (width) + (height)$) -- ($(Rl_{\y, \x}) + (height)$) -- ($(Rl_{\y, \x})$);
\node (label) at ($(Rl_{\y, \x}) + (labelshift)$)  {$l_{\x}^{\y}$};
}

%satisfied literal, shiftet to the left
\foreach \x/\y/\z in {1/2/2, 3/3/4,6/3/4}{
\draw[fill, black!20] ($(Rl_{\y, \x})$) -- ($(Rl_{\y, \x}) + (width)$) -- ($(Rl_{\y, \x}) + (width) + (height)$) -- ($(Rl_{\y, \x}) + (height)$) -- ($(Rl_{\y, \x})$);
\draw[black] ($(Rl_{\y, \x})$) -- ($(Rl_{\y, \x}) + (width)$) -- ($(Rl_{\y, \x}) + (width) + (height)$) -- ($(Rl_{\y, \x}) + (height)$) -- ($(Rl_{\y, \x})$);
\node (label) at ($(Rl_{\y, \x}) + (labelshift)$)  {$l_{\x}^{\y}$};
}
\foreach \x/\y/\z in {4/2/2}{
\draw[fill, black!20] ($(Rl_{\y, \x})$) -- ($(Rl_{\y, \x}) + (width)$) -- ($(Rl_{\y, \x}) + (width) + (height)$) -- ($(Rl_{\y, \x}) + (height)$) -- ($(Rl_{\y, \x})$);
\draw[black] ($(Rl_{\y, \x})$) -- ($(Rl_{\y, \x}) + (width)$) -- ($(Rl_{\y, \x}) + (width) + (height)$) -- ($(Rl_{\y, \x}) + (height)$) -- ($(Rl_{\y, \x})$);
\node (label) at ($(Rl_{\y, \x}) + (labelshift)$)  {$l_{\x}^{\y}$};
}

%usatisfied literal, shiftet to the right
\foreach \x/\y/\z in {2/3/4}{
\draw[fill, black!20] ($(Rl_{\y, \x})$) -- ($(Rl_{\y, \x}) + (width)$) -- ($(Rl_{\y, \x}) + (width) + (height)$) -- ($(Rl_{\y, \x}) + (height)$) -- ($(Rl_{\y, \x})$);
\draw[black] ($(Rl_{\y, \x})$) -- ($(Rl_{\y, \x}) + (width)$) -- ($(Rl_{\y, \x}) + (width) + (height)$) -- ($(Rl_{\y, \x}) + (height)$) -- ($(Rl_{\y, \x})$);
\node (label) at ($(Rl_{\y, \x}) + (labelshift)$)  {$l_{\x}^{\y}$};
}
\foreach \x/\y/\z in {5/3/4}{
\draw[fill, black!20] ($(Rl_{\y, \x})$) -- ($(Rl_{\y, \x}) + (width)$) -- ($(Rl_{\y, \x}) + (width) + (height)$) -- ($(Rl_{\y, \x}) + (height)$) -- ($(Rl_{\y, \x})$);
\draw[black] ($(Rl_{\y, \x})$) -- ($(Rl_{\y, \x}) + (width)$) -- ($(Rl_{\y, \x}) + (width) + (height)$) -- ($(Rl_{\y, \x}) + (height)$) -- ($(Rl_{\y, \x})$);
\node (label) at ($(Rl_{\y, \x}) + (labelshift)$)  {$l_{\x}^{\y}$};
}

%%%%%%%%%%%%%%%%%%%%%%%%%%%%%%%%%%%%%%%%%%%%%%%%%%%%%%%%%%%%%%%%%%%%%%%%%%%%%%%%%%%%%
%arrow-vertices
%%%%%%%%%%%%%%%%%%%%%%%%%%%%%%%%%%%%%%%%%%%%%%%%%%%%%%%%%%%%%%%%%%%%%%%%%%%%%%%%%%%%%%

%satisfied arrow-t_i

\foreach \y in {1,4}{
\draw[fill, black!20] ($(Rt_{\y, rightarrow})$) -- ($(Rt_{\y, rightarrow}) + (width)$) -- ($(Rt_{\y, rightarrow}) + (width) + (height)$) -- ($(Rt_{\y, rightarrow}) + (height)$) -- ($(Rt_{\y, rightarrow})$);
\draw[black] ($(Rt_{\y, rightarrow})$) -- ($(Rt_{\y, rightarrow}) + (width)$) -- ($(Rt_{\y, rightarrow}) + (width) + (height)$) -- ($(Rt_{\y, rightarrow}) + (height)$) -- ($(Rt_{\y, rightarrow})$);
\node (label) at ($(Rt_{\y, rightarrow}) + (labelshift)$)  {$\overset{_{\rightarrow}}{t_{\y}}$};
\draw[fill, black!20] ($(Rt_{\y, leftarrow})$) -- ($(Rt_{\y, leftarrow}) + (width)$) -- ($(Rt_{\y, leftarrow}) + (width) + (height)$) -- ($(Rt_{\y, leftarrow}) + (height)$) -- ($(Rt_{\y, leftarrow})$);
\draw[black] ($(Rt_{\y, leftarrow})$) -- ($(Rt_{\y, leftarrow}) + (width)$) -- ($(Rt_{\y, leftarrow}) + (width) + (height)$) -- ($(Rt_{\y, leftarrow}) + (height)$) -- ($(Rt_{\y, leftarrow})$);
\node (label) at ($(Rt_{\y, leftarrow}) + (labelshift)$)  {$\overset{_{\leftarrow}}{t_{\y}}$};
}
%unsatisfied arrow-t_i
\foreach \y in {2,3}{
\draw[fill, black!20] ($(Rt_{\y, rightarrow})$) -- ($(Rt_{\y, rightarrow}) + (width)$) -- ($(Rt_{\y, rightarrow}) + (width) + (height)$) -- ($(Rt_{\y, rightarrow}) + (height)$) -- ($(Rt_{\y, rightarrow})$);
\draw[black] ($(Rt_{\y, rightarrow})$) -- ($(Rt_{\y, rightarrow}) + (width)$) -- ($(Rt_{\y, rightarrow}) + (width) + (height)$) -- ($(Rt_{\y, rightarrow}) + (height)$) -- ($(Rt_{\y, rightarrow})$);
\node (label) at ($(Rt_{\y, rightarrow}) + (labelshift)$)  {$\overset{_{\rightarrow}}{t_{\y}}$};
\draw[fill, black!20] ($(Rt_{\y, leftarrow})$) -- ($(Rt_{\y, leftarrow}) + (width)$) -- ($(Rt_{\y, leftarrow}) + (width) + (height)$) -- ($(Rt_{\y, leftarrow}) + (height)$) -- ($(Rt_{\y, leftarrow})$);
\draw[black] ($(Rt_{\y, leftarrow})$) -- ($(Rt_{\y, leftarrow}) + (width)$) -- ($(Rt_{\y, leftarrow}) + (width) + (height)$) -- ($(Rt_{\y, leftarrow}) + (height)$) -- ($(Rt_{\y, leftarrow})$);
\node (label) at ($(Rt_{\y, leftarrow}) + (labelshift)$)  {$\overset{_{\leftarrow}}{t_{\y}}$};
}

%satisfied arrow-f_i^r
\foreach \y in {2,3}{
\foreach \x in {1,2}{
\draw[fill, black!20] ($(Rf_{\y, \x, rightarrow})$) -- ($(Rf_{\y, \x, rightarrow}) + (width)$) -- ($(Rf_{\y, \x, rightarrow}) + (width) + (height)$) -- ($(Rf_{\y, \x, rightarrow}) + (height)$) -- ($(Rf_{\y, \x, rightarrow})$);
\draw[black] ($(Rf_{\y, \x, rightarrow})$) -- ($(Rf_{\y, \x, rightarrow}) + (width)$) -- ($(Rf_{\y, \x, rightarrow}) + (width) + (height)$) -- ($(Rf_{\y, \x, rightarrow}) + (height)$) -- ($(Rf_{\y, \x, rightarrow})$);
\node (label) at ($(Rf_{\y, \x, rightarrow}) + (labelshift)$)  {$\overset{_{\rightarrow}}{f_{\y}^{\x}}$};
\draw[fill, black!20] ($(Rf_{\y, \x, leftarrow})$) -- ($(Rf_{\y, \x, leftarrow}) + (width)$) -- ($(Rf_{\y, \x, leftarrow}) + (width) + (height)$) -- ($(Rf_{\y, \x, leftarrow}) + (height)$) -- ($(Rf_{\y, \x, leftarrow})$);
\draw[black] ($(Rf_{\y, \x, leftarrow})$) -- ($(Rf_{\y, \x, leftarrow}) + (width)$) -- ($(Rf_{\y, \x, leftarrow}) + (width) + (height)$) -- ($(Rf_{\y, \x, leftarrow}) + (height)$) -- ($(Rf_{\y, \x, leftarrow})$);
\node (label) at ($(Rf_{\y, \x, leftarrow}) + (labelshift)$)  {$\overset{_{\leftarrow}}{f_{\y}^{\x}}$};
}
}

%unsatisfied arrow-f_i^r
\foreach \y in {1,4}{
\foreach \x in {1,2}{
\draw[fill, black!20] ($(Rf_{\y, \x, rightarrow})$) -- ($(Rf_{\y, \x, rightarrow}) + (width)$) -- ($(Rf_{\y, \x, rightarrow}) + (width) + (height)$) -- ($(Rf_{\y, \x, rightarrow}) + (height)$) -- ($(Rf_{\y, \x, rightarrow})$);
\draw[black] ($(Rf_{\y, \x, rightarrow})$) -- ($(Rf_{\y, \x, rightarrow}) + (width)$) -- ($(Rf_{\y, \x, rightarrow}) + (width) + (height)$) -- ($(Rf_{\y, \x, rightarrow}) + (height)$) -- ($(Rf_{\y, \x, rightarrow})$);
\node (label) at ($(Rf_{\y, \x, rightarrow}) + (labelshift)$)  {$\overset{_{\rightarrow}}{f_{\y}^{\x}}$};
\draw[fill, black!20] ($(Rf_{\y, \x, leftarrow})$) -- ($(Rf_{\y, \x, leftarrow}) + (width)$) -- ($(Rf_{\y, \x, leftarrow}) + (width) + (height)$) -- ($(Rf_{\y, \x, leftarrow}) + (height)$) -- ($(Rf_{\y, \x, leftarrow})$);
\draw[black] ($(Rf_{\y, \x, leftarrow})$) -- ($(Rf_{\y, \x, leftarrow}) + (width)$) -- ($(Rf_{\y, \x, leftarrow}) + (width) + (height)$) -- ($(Rf_{\y, \x, leftarrow}) + (height)$) -- ($(Rf_{\y, \x, leftarrow})$);
\node (label) at ($(Rf_{\y, \x, leftarrow}) + (labelshift)$)  {$\overset{_{\leftarrow}}{f_{\y}^{\x}}$};
}
}

%%%%%%%%%%%%%%%%%%%%%%%%%%%%%%%%%%%%%%%%%%%%%%%%%%%%%%%%%%%%%%%%%%%%%%%%%%%%%%%%%%%%%
%h-vertices
%%%%%%%%%%%%%%%%%%%%%%%%%%%%%%%%%%%%%%%%%%%%%%%%%%%%%%%%%%%%%%%%%%%%%%%%%%%%%%%%%%%%%%

%satisfied h^0_t_i
\foreach \y in {1,4}{
\draw[fill, black!20] ($(Rht_{\y,0})$) -- ($(Rht_{\y,0}) + (width)$) -- ($(Rht_{\y,0}) + (width) + (height)$) -- ($(Rht_{\y,0}) + (height)$) -- ($(Rht_{\y,0})$);
\draw[black] ($(Rht_{\y,0})$) -- ($(Rht_{\y,0}) + (width)$) -- ($(Rht_{\y,0}) + (width) + (height)$) -- ($(Rht_{\y,0}) + (height)$) -- ($(Rht_{\y,0})$);
\node (label) at ($(Rht_{\y,0}) + (labelshift)$)  {$h^0_{t_{\y}}$};
}

%unsatisfied h^0_t_i
\foreach \y in {2,3}{
\draw[fill, black!20] ($(Rht_{\y,0})$) -- ($(Rht_{\y,0}) + (width)$) -- ($(Rht_{\y,0}) + (width) + (height)$) -- ($(Rht_{\y,0}) + (height)$) -- ($(Rht_{\y,0})$);
\draw[black] ($(Rht_{\y,0})$) -- ($(Rht_{\y,0}) + (width)$) -- ($(Rht_{\y,0}) + (width) + (height)$) -- ($(Rht_{\y,0}) + (height)$) -- ($(Rht_{\y,0})$);
\node (label) at ($(Rht_{\y,0}) + (labelshift)$)  {$h^0_{t_{\y}}$};
}
%satisfied h^0_f_i^r
\foreach \y in {2,3}{
\foreach \x in {1,2}{
\draw[fill, black!20] ($(Rhf_{\y,\x,0})$) -- ($(Rhf_{\y,\x,0}) + (width)$) -- ($(Rhf_{\y,\x,0}) + (width) + (height)$) -- ($(Rhf_{\y,\x,0}) + (height)$) -- ($(Rhf_{\y,\x,0})$);
\draw[black] ($(Rhf_{\y,\x,0})$) -- ($(Rhf_{\y,\x,0}) + (width)$) -- ($(Rhf_{\y,\x,0}) + (width) + (height)$) -- ($(Rhf_{\y,\x,0}) + (height)$) -- ($(Rhf_{\y,\x,0})$);
\node (label) at ($(Rhf_{\y,\x,0}) + (labelshift)$)  {$h^0_{f_{\y}^{\x}}$};
}}

%unsatisfied h^0_f_i^r
\foreach \y in {1,4}{
\foreach \x in {1,2}{
\draw[fill, black!20] ($(Rhf_{\y,\x,0})$) -- ($(Rhf_{\y,\x,0}) + (width)$) -- ($(Rhf_{\y,\x,0}) + (width) + (height)$) -- ($(Rhf_{\y,\x,0}) + (height)$) -- ($(Rhf_{\y,\x,0})$);
\draw[black] ($(Rhf_{\y,\x,0})$) -- ($(Rhf_{\y,\x,0}) + (width)$) -- ($(Rhf_{\y,\x,0}) + (width) + (height)$) -- ($(Rhf_{\y,\x,0}) + (height)$) -- ($(Rhf_{\y,\x,0})$);
\node (label) at ($(Rhf_{\y,\x,0}) + (labelshift)$)  {$h^0_{f_{\y}^{\x}}$};
}}

%satisfied h^r_t_i r=1,2
\foreach \y in {1,4}{
\foreach \r in {1,2}{
\draw[fill, black!20] ($(Rht_{\y,\r})$) -- ($(Rht_{\y,\r}) + (width)$) -- ($(Rht_{\y,\r}) + (width) + (height)$) -- ($(Rht_{\y,\r}) + (height)$) -- ($(Rht_{\y,\r})$);
\draw[black] ($(Rht_{\y,\r})$) -- ($(Rht_{\y,\r}) + (width)$) -- ($(Rht_{\y,\r}) + (width) + (height)$) -- ($(Rht_{\y,\r}) + (height)$) -- ($(Rht_{\y,\r})$);
\node (label) at ($(Rht_{\y,\r}) + (labelshift)$)  {$h^{\r}_{t_{\y}}$};
}}

%unsatisfied  h^r_t_i r=1,2
\foreach \y in {2,3}{
\foreach \r in {1,2}{
\draw[fill, black!20] ($(Rht_{\y,\r})$) -- ($(Rht_{\y,\r}) + (width)$) -- ($(Rht_{\y,\r}) + (width) + (height)$) -- ($(Rht_{\y,\r}) + (height)$) -- ($(Rht_{\y,\r})$);
\draw[black] ($(Rht_{\y,\r})$) -- ($(Rht_{\y,\r}) + (width)$) -- ($(Rht_{\y,\r}) + (width) + (height)$) -- ($(Rht_{\y,\r}) + (height)$) -- ($(Rht_{\y,\r})$);
\node (label) at ($(Rht_{\y,\r}) + (labelshift)$)  {$h^{\r}_{t_{\y}}$};
}}
%satisfied h^r_f_i^x r=1,2
\foreach \y in {2,3}{
\foreach \r in {1,2}{
\foreach \x in {1,2}{
\draw[fill, black!20] ($(Rhf_{\y,\x,\r})$) -- ($(Rhf_{\y,\x,\r}) + (width)$) -- ($(Rhf_{\y,\x,\r}) + (width) + (height)$) -- ($(Rhf_{\y,\x,\r}) + (height)$) -- ($(Rhf_{\y,\x,\r})$);
\draw[black] ($(Rhf_{\y,\x,\r})$) -- ($(Rhf_{\y,\x,\r}) + (width)$) -- ($(Rhf_{\y,\x,\r}) + (width) + (height)$) -- ($(Rhf_{\y,\x,\r}) + (height)$) -- ($(Rhf_{\y,\x,\r})$);
\node (label) at ($(Rhf_{\y,\x,\r}) + (labelshift)$)  {$h^{\r}_{f_{\y}^{\x}}$};
}}}

%unsatisfied h^r_f_i_x r=1,2
\foreach \y in {1,4}{
\foreach \r in {1,2}{
\foreach \x in {1,2}{
\draw[fill, black!20] ($(Rhf_{\y,\x,\r})$) -- ($(Rhf_{\y,\x,\r}) + (width)$) -- ($(Rhf_{\y,\x,\r}) + (width) + (height)$) -- ($(Rhf_{\y,\x,\r}) + (height)$) -- ($(Rhf_{\y,\x,\r})$);
\draw[black] ($(Rhf_{\y,\x,\r})$) -- ($(Rhf_{\y,\x,\r}) + (width)$) -- ($(Rhf_{\y,\x,\r}) + (width) + (height)$) -- ($(Rhf_{\y,\x,\r}) + (height)$) -- ($(Rhf_{\y,\x,\r})$);
\node (label) at ($(Rhf_{\y,\x,\r}) + (labelshift)$)  {$h^{\r}_{f_{\y}^{\x}}$};
}}}

%satisfied h^r_t_i r=3,4
\foreach \y in {1,4}{
\foreach \r in {3,4}{
\draw[fill, black!20] ($(Rht_{\y,\r})$) -- ($(Rht_{\y,\r}) + (width)$) -- ($(Rht_{\y,\r}) + (width) + (height)$) -- ($(Rht_{\y,\r}) + (height)$) -- ($(Rht_{\y,\r})$);
\draw[black] ($(Rht_{\y,\r})$) -- ($(Rht_{\y,\r}) + (width)$) -- ($(Rht_{\y,\r}) + (width) + (height)$) -- ($(Rht_{\y,\r}) + (height)$) -- ($(Rht_{\y,\r})$);
\node (label) at ($(Rht_{\y,\r}) + (labelshift)$)  {$h^{\r}_{t_{\y}}$};
}}

%unsatisfied  h^r_t_i r=3,4
\foreach \y in {2,3}{
\foreach \r in {3,4}{
\draw[fill, black!20] ($(Rht_{\y,\r})$) -- ($(Rht_{\y,\r}) + (width)$) -- ($(Rht_{\y,\r}) + (width) + (height)$) -- ($(Rht_{\y,\r}) + (height)$) -- ($(Rht_{\y,\r})$);
\draw[black] ($(Rht_{\y,\r})$) -- ($(Rht_{\y,\r}) + (width)$) -- ($(Rht_{\y,\r}) + (width) + (height)$) -- ($(Rht_{\y,\r}) + (height)$) -- ($(Rht_{\y,\r})$);
\node (label) at ($(Rht_{\y,\r}) + (labelshift)$)  {$h^{\r}_{t_{\y}}$};
}}
%satisfied h^r_f_i^x r=3,4
\foreach \y in {2,3}{
\foreach \r in {3,4}{
\foreach \x in {1,2}{
\draw[fill, black!20] ($(Rhf_{\y,\x,\r})$) -- ($(Rhf_{\y,\x,\r}) + (width)$) -- ($(Rhf_{\y,\x,\r}) + (width) + (height)$) -- ($(Rhf_{\y,\x,\r}) + (height)$) -- ($(Rhf_{\y,\x,\r})$);
\draw[black] ($(Rhf_{\y,\x,\r})$) -- ($(Rhf_{\y,\x,\r}) + (width)$) -- ($(Rhf_{\y,\x,\r}) + (width) + (height)$) -- ($(Rhf_{\y,\x,\r}) + (height)$) -- ($(Rhf_{\y,\x,\r})$);
\node (label) at ($(Rhf_{\y,\x,\r}) + (labelshift)$)  {$h^{\r}_{f_{\y}^{\x}}$};
}}}

%unsatisfied h^r_f_i_x r=3,4
\foreach \y in {1,4}{
\foreach \r in {3,4}{
\foreach \x in {1,2}{
\draw[fill, black!20] ($(Rhf_{\y,\x,\r})$) -- ($(Rhf_{\y,\x,\r}) + (width)$) -- ($(Rhf_{\y,\x,\r}) + (width) + (height)$) -- ($(Rhf_{\y,\x,\r}) + (height)$) -- ($(Rhf_{\y,\x,\r})$);
\draw[black] ($(Rhf_{\y,\x,\r})$) -- ($(Rhf_{\y,\x,\r}) + (width)$) -- ($(Rhf_{\y,\x,\r}) + (width) + (height)$) -- ($(Rhf_{\y,\x,\r}) + (height)$) -- ($(Rhf_{\y,\x,\r})$);
\node (label) at ($(Rhf_{\y,\x,\r}) + (labelshift)$)  {$h^{\r}_{f_{\y}^{\x}}$};
}}}

\end{tikzpicture}